\documentclass[10pt,english]{article}
\usepackage{geometry}
\usepackage{setspace}
\usepackage{empheq}
\usepackage[T1]{fontenc}
\usepackage[utf8]{inputenc}
\usepackage{babel}
\usepackage{varioref}
\usepackage{amsmath,nccmath}
\numberwithin{equation}{section}
\usepackage{amssymb}
\usepackage{mathrsfs} 
\usepackage{upgreek}
\usepackage{graphicx}
\usepackage{floatrow}
\usepackage[labelfont=bf]{caption}
\usepackage{subcaption}
\usepackage{appendix}
\usepackage{slashed}
\usepackage{svg}
\usepackage{bm}
\makeatletter

\usepackage{hyperref}
\usepackage{enumitem}

\usepackage[makeroom]{cancel}

\usepackage{authblk}
\usepackage{xcolor}
\usepackage{amsthm}
\usepackage{graphicx}

\graphicspath{{figures}}

\hypersetup{
	colorlinks,
	linkcolor={blue!60!black},
	citecolor={green!60!black},
	urlcolor={blue!80!black}
}
\usepackage{mathtools}
\usepackage{mathrsfs}
\usepackage{comment}
\usepackage{cases}
\usepackage{xargs}   

\usepackage[colorinlistoftodos,prependcaption,textsize=tiny,disable]{todonotes}
\newcommandx{\typo}[2][1=]{\todo[linecolor=red,backgroundcolor=red!25,bordercolor=red,#1]{#2}}
\newcommandx{\change}[2][1=]{\todo[linecolor=blue,backgroundcolor=blue!25,bordercolor=blue,#1]{#2}}
\newcommandx{\question}[2][1=]{\todo[linecolor=green,backgroundcolor=blue!25,bordercolor=green,#1]{#2}}
\newcommandx{\answer}[1]{\todo[linecolor=pink,backgroundcolor=pink!25,bordercolor=pink]{#1}}
\newcommandx{\unsure}[2][1=]{\todo[linecolor=green,backgroundcolor=green!25,bordercolor=green,#1]{#2}}
\newcommandx{\improve}[2][1=]{\todo[linecolor=violet,backgroundcolor=violet!25,bordercolor=violet,#1]{#2}}
\newcommandx{\thiswillnotshow}[2][1=]{\todo[disable,#1]{#2}}
\allowdisplaybreaks

\newtheorem{thm}{Theorem}[section]

\newtheorem{lemma}{Lemma}[section]

\newtheorem{cor}{Corollary}[section]
\newtheorem{ass}{Assumption}[section]
\newtheorem{iterate}{Iterate}[section]
\theoremstyle{remark}
\newtheorem{rem}{Remark}[section]
\newtheorem{obs}{Observation}[section]
\theoremstyle{plain}

\newtheorem*{claim*}{Claim}
\newtheorem{prop}{Proposition}[section]

\newenvironment{nalign}{
	\begin{equation}
		\begin{aligned}
		}{
		\end{aligned}
	\end{equation}
	\ignorespacesafterend
}
\theoremstyle{remark}

\newcommand{\dd}{\mathop{}\!\mathrm{d}}
\renewcommand{\O}{\mathcal{O}}
\newcommand{\M}{\mathcal{M}}
\renewcommand{\div}{\slashed{\mathrm{div}}\,}
\newcommand{\e}{\mathrm{e}}

\newcommand{\scrip}{\mathcal{I}^+}
\newcommand{\hplus}{\mathcal{H}^+}

\newcommand{\tr}{\mathrm{tr}}

\usepackage{accents}    

\newcommand{\gs}{\slashed{g}}
\newcommand{\gsh}{\hat{\slashed{g}}}

\newcommand{\trx}{\tr\chi}
\newcommand{\trxb}{\tr\underline{\chi}}
\newcommand{\otrx}{\Omega\tr\chi}
\newcommand{\otrxb}{\Omega\tr\underline{\chi}}
\newcommand{\xh}{\hat\chi}
\newcommand{\xhb}{\underline{\hat{\chi}}}

\newcommand{\et}{\eta}
\newcommand{\etb}{\underline{\eta}}

\newcommand{\pv}{\partial_v}
\newcommand{\pu}{\partial_u}
\newcommand{\Xx}{X}
\newcommand{\Xb}{\underline{X}}
\renewcommand{\sl}{\mathring{\slashed{\nabla}}}
\newcommand{\Dl}{\mathring{\slashed{\Delta}}}

\renewcommand{\div}{\slashed{\mathrm{div}}}

\renewcommand{\(}{\left(}
\renewcommand{\)}{\right)}

  \newcommand{\V}{\mathcal{V}_{\mathrm{dec}}}
\newcommand{\Vk}{\mathcal{V}_{\mathrm{Kil}}}
\newcommand{\Vc}{\mathcal{V}_{\mathrm{com}}}
\newcommand{\lesVn}{\overset{\V^n}{\lesssim}}
\newcommand{\lesV}[1][\V]{\overset{#1}{\lesssim}}

\newcommand{\philc}{\check{\phi}_\ell}
\newcommand{\philcm}{\check{\phi}_{\ell,m}}
\newcommand{\apphi}[1][\ell,m]{\phi^{\mathrm{app}}_{#1}}
\newcommand{\appsi}[1][\ell,m]{\psi^{\mathrm{app}}_{#1}}
 \newcommand{\rc}{\mathring{R}}

\newcommand{\detgm}{\sqrt{-\det g_M}}
\newcommand{\detgs}{\sqrt{\det \slashed{g}}}
\newcommand{\detgsm}{\sqrt{\det \slashed{g}_M}}

\newcommand{\R}{\mathbb{R}}
\newcommand{\C}{\mathcal{C}}
\newcommand{\Cbar}{\underline{\mathcal{C}}}
\renewcommand{\S}{\mathcal{S}}
\newcommand{\N}{\mathbb{N}}
\newcommand{\D}{\mathcal{D}}

\usepackage[capitalize,nameinlink]{cleveref}

\crefname{lemma}{Lemma}{Lemmata}
\crefname{thm}{Theorem}{Theorems}
\crefname{prop}{Proposition}{Propositions}
\crefname{conj}{Conjecture}{Conjectures}
\crefname{ass}{Assumption}{Assumptions}
\crefname{cor}{Corollary}{Corollary}
\crefname{rem}{Remark}{Remarks}
\crefname{obs}{Observation}{Observations}
\crefname{defi}{Definition}{Definitions}
\crefname{iterate}{Iterate}{Iterates}
\crefname{equation}{}{}
\crefname{enumi}{}{}
\usepackage{soul}

\crefname{subfigure}{Figure}{Figure}
\Crefname{subfigure}{Figure}{Figure}

\title{Linear and nonlinear late-time tails 
	on\\ dynamical black hole spacetimes via time integrals} 

\author[1]{Dejan Gajic\thanks{dejan.gajic@uni-leipzig.de}}
\author[1,2]{Lionor  Kehrberger\thanks{kehrberger@mis.mpg.de}}

\affil[1]{Institute for Theoretical Physics, University of Leipzig, Br\"uderstraße 16,  04103 Leipzig, Germany}
\affil[2]{Max Planck Institute for Mathematics in the Sciences,  Inselstraße 22, 04103 Leipzig, Germany}
\date{November 26, 2025} 
\allowdisplaybreaks
\makeatother

\begin{document}
	\pagenumbering{roman}
	
	\maketitle 
	\begin{abstract}
	
We prove the global leading-order late-time asymptotic behaviour of solutions to inhomogeneous wave equations on dynamical black hole exterior backgrounds that settle down to Schwarzschild backgrounds with arbitrarily small decay rates.  In particular, we show that for non-spherically symmetric solutions arising from compactly supported initial data, the late-time decay deviates from Price's law---governing the decay for stationary black hole backgrounds---by exhibiting  slower time decay by exactly one additional power. 

Our proof is based around the observation that the emergence of late-time ``tails'', featuring inverse-polynomial decay in \emph{time}, is intimately connected to conformal irregularity properties in \emph{space} (towards future null infinity) of time integrals of the solutions. This relationship is exploited through a purely physical-space approach based around energy estimates, in which the dynamical wave operator is treated as a Schwarzschild wave operator with an inhomogeneous term. Going from almost-sharp decay to global asymptotics is achieved by exploiting this relation for the difference between the solution and a carefully chosen global tail function. 

 We further apply our method to several examples of nonlinear wave equations and comment on its robustness to more general settings, such as dynamical spacetimes converging to sub-extremal Kerr spacetimes and higher-dimensional wave operators with even or odd spacetime dimensions. 

%
		
\end{abstract}
\tableofcontents
\pagenumbering{arabic}
\section{Introduction}\label{sec:intro}
Isolated gravitational systems radiate away gravitational energy \cite{ligo16}. This gravitational radiation can be studied mathematically by considering dynamical, asymptotically flat solutions $(\mathcal{M},g)$ to the Einstein vacuum equations
\begin{equation}\label{eq:intro:Einstein}
	\mathrm{Ric}[g]=0,
\end{equation}
and analysing the temporal behaviour of limits of derivatives of the metric $g$ along outgoing null hypersurfaces, i.e.\ at \emph{future null infinity} $\mathcal{I}^+$. 
 A quantitative understanding of the time-decay of gravitational radiation is intimately connected to an understanding of the stability properties of spacetimes arising from small initial data perturbations of stationary solutions to \cref{eq:intro:Einstein}.
  For dynamical black hole solutions, this radiation is moreover strongly correlated with radiation entering the black hole interior through the event horizon $\mathcal{H}^+$, and consequently, with the singularity properties of dynamical black hole interiors and the Strong Cosmic Censorship conjecture.

In the present paper, we showcase how the precise late-time asymptotics of radiation, in the form of \emph{inverse-polynomial tails}, can be derived in a dynamical spacetime setting, starting from weak time-decay estimates. In place of \cref{eq:intro:Einstein}, we consider geometric wave equations of the form:
\begin{equation}\label{eq:intro:wave}
	\Box_g \phi =F
\end{equation}
with respect to a dynamical spacetime background $(\mathcal{M},g)$ describing a black hole exterior, with $g$ settling down to a Schwarzschild metric $g_M$ at rates consistent with the nonlinear stability results in \cite{dafermos_non-linear_2021,klsz20,klainerman_kerr_2021,gks24}.
Here, $F$ can either be taken to be a fixed inhomogeneity ($F=F(x)$), or a nonlinearity ($F=F(x,\phi,\partial\phi,\partial^2\phi)$).
 
 One may think of \cref{eq:intro:wave} as modelling the nonlinear wave equations satisfied by certain curvature components of solutions to \cref{eq:intro:Einstein} (these equations reduce to the decoupled Teukolsky or Regge--Wheeler wave equations when linearising around Schwarzschild spacetimes). In the setting of \cref{eq:intro:Einstein}, these wave equations, in turn, govern the decay properties of the difference $g-g_M$, which we here treat as an \emph{assumption} in studying \cref{eq:intro:wave}.
\begin{figure}[htb]
	\includegraphics[width=0.35\textwidth]{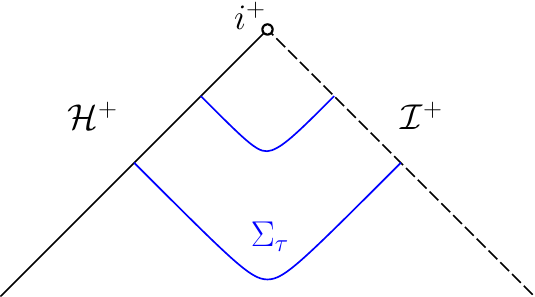}
	\caption{The spacetime $(\mathcal{M},g)$ with its foliation by level sets of $\tau$.}\label{fig:intro}
\end{figure}

The main result of the present paper can be summarised in the following informal theorem (see \cref{thm:main} in \cref{sec:main} for a precise version, and for precise formulations of the assumptions):
\newcommand{\integral}{\mathfrak{I}}
\begin{thm}[Informal version of the main theorem (\cref{thm:main})]
\label{thm:mainrough}
Assume that $|g-g_M|\lesssim (\tau+r)^{-1}\tau^{-\delta}$ and $|r^4F|\lesssim \tau^{-2-\delta}$ with $\delta>0$, $\tau$ a hyperboloidal time function with level sets $\Sigma_{\tau}$ and $r$ a radial function, and make similar assumptions on the $1/r$-expansions of $g-g_M$ and $F$. Let $\phi$ be a solution to \cref{eq:intro:wave} arising from smooth compactly supported initial data on $\Sigma_{\tau_0}$ (see~\cref{fig:intro}) and satisfying weak, global time-decay estimates.

Then $\phi$ has the following leading-order global asymptotic behaviour as $\tau\to \infty$:

\begin{align}
\label{eq:latetimerough1}
\phi\sim &\: \integral_0[\phi]\partial_{\tau}	(\tau^{-1}(\tau+r)^{-1})+\sum_{m=-1}^1\integral_{1 m}[\phi]Y_{1, m}w_1(r) \tau^{-2}(\tau+r)^{-2},\\
\label{eq:latetimerough2}
\phi_{\geq \ell}\sim &\: \sum_{m=-\ell}^{\ell} \integral_{\ell, m}[\phi]Y_{\ell, m}w_{\ell}(r) \tau^{-1-\ell}(\tau+r)^{-1-\ell}\qquad \text{for }1\leq \ell<2+\delta,\qquad \text{where:}
\end{align}
\begin{itemize}[leftmargin=*, labelindent=0pt]
\item 	$\phi_{\geq\ell}$ projects $\phi$ to all spherical harmonics $Y_{\ell',m}$ (defined with respect to the round spheres at $\mathcal{I}^+$) with $\ell'\geq \ell$.
\item The spatial profile is given by $w_{\ell}(r)=c_{\ell}P_{\ell}(M^{-1}(r-M))= r^{\ell}+O(r^{-\ell-1})$, where $P_\ell$ is a Legendre function.
\item In the case $F=0$, the constant $\integral_0[\phi]$ can be expressed as follows: 
\begin{align}\label{eq:intro:I0}
	4\pi \integral_0[\phi]:=&\: \frac{M}{2}\int_{\mathcal{H}^+\cap\Sigma_{\tau_0}}	\phi\dd\slashed{\mu}+\frac{M}{2}\int_{\Sigma_{\tau_0}} \mathbf{n}_{\Sigma_{\tau_0}}(\phi)\dd\mu_{\Sigma_{\tau_0}}+\frac{M}{2}\int_{\mathcal{H}^+} \tr \chi\cdot  \phi\dd\slashed{\mu}\dd v+\int_{\mathcal{I}^+}(\mathfrak{m}-M)r\phi\; r^{-2}\dd\slashed{\mu}\dd u\\
	=&\: \int_{\scrip} \mathfrak{m} r\phi  r^{-2} \dd \slashed{\mu}\dd u.\label{eq:intro:I0b}
\end{align}

%
Here, $\mathbf{n}_{\Sigma_{\tau_0}}$ and $d\mu_{\Sigma_{\tau_0}}$ denote the future-directed normal vector field to, and the induced volume form of $\Sigma_{\tau_0}$, $\slashed{g}$ is the induced metric on spheres foliating $\mathcal{H}^+$ and $\mathcal{I}^+$ with corresponding induced volume form $\dd\slashed{\mu}$, which are obtained by flowing along null generators with parameters $v$ and $u$, respectively. The functions $\tr \chi|_{\mathcal{H}^+}: \mathcal{H}^+\to \R$ and $\mathfrak{m}: \mathcal{I}^+\to \R$ denote the expansion of the null generators of $\mathcal{H}^+$ and the integral in $u$ of the Bondi news function, respectively.\footnote{The spherical average of $\mathfrak{m}$ is the Bondi mass function.}
\item For $\ell>0$, the constants $\integral_{\ell, m}[\phi]$ can be expressed as integrals over $r\phi|_{\mathcal{I}^+}$ and its derivatives, multiplied by functions that vanish when $g=g_M$ and $F\equiv 0$; see \cref{thm:main} for their precise forms. Both the integrals $\integral_0[\phi]$ and $\integral_{\ell,m}[\phi]$ are moreover non-vanishing for a generic subset of initial data.
\end{itemize}
Analogues of \cref{eq:latetimerough1} and \cref{eq:latetimerough2} hold in the case where $F$ is a nonlinearity; see \cref{thm:main:phi3,thm:main:phi4,thm:main:puphipvphi}.
\end{thm}

\begin{rem}[Assumptions]
	The assumptions on the background metric $g$ are stated precisely in \cref{ass:dyn:pol}.
	In contrast, the weak decay estimates that we assume on $\phi$ can easily be proved with standard methods. We keep them as part of the statement of the theorem to emphasise that our method takes as an \emph{input} a weak decay estimate and produces as an \emph{output} precise asymptotics. The precise assumptions on $F$ are stated in \cref{ass:main}.
\end{rem}

\begin{rem}[Late-time tails]
	 \cref{thm:mainrough} implies, in particular, late-time asymptotics of the Friedlander radiation field $r\phi|_{\mathcal{I}^+}$ that are governed by the following \emph{late-time tails}:
	\begin{align}\label{eq:intro:asymptoticatScrip}
		r\phi|_{\mathcal{I}^+}\sim -\integral_0[\phi]\tau^{-2}+\sum_{m=-1}^1\integral_{1 m}[\phi]Y_{1, m} \tau^{-2},&&
		r\phi_{\geq \ell}|_{\mathcal{I}^+}\sim  \sum_{m=-\ell}^{\ell} \integral_{\ell, m}[\phi]Y_{\ell, m} \tau^{-1-\ell} \quad \text{for }1\leq  \ell<2+\delta.
	\end{align}
On the other hand, the late-time asymptotics of $\phi$ along constant $\{r=r_0\}$ with $r_0\geq 2M$ are governed by:
	\begin{align}
	\label{eq:intro:asymptoticatr0}
		\phi|_{r=r_0}\sim   -2\integral_0[\phi]\tau^{-3},&&
		\phi_{\geq \ell}|_{r=r_0}\sim  \sum_{m=-\ell}^{\ell} \integral_{\ell, m}[\phi]Y_{\ell, m}w_{\ell}(r_0) \tau^{-2-2\ell}\quad \text{for } 1\leq  \ell<2+\delta
	\end{align}
	and are correlated with the asymptotics along $\mathcal{I}^+$ via the coefficients $\integral_0[\phi]$ and $\integral_{\ell, m}[\phi]$.
	
	Observe that in the case $\phi_{\geq \ell}$ with $\ell \geq 1$, the late-time tails for the dynamical setting $g\neq g_M$ decay exactly \emph{one power slower} than in the stationary setting $g= g_M$ (where they are governed by ``Price's law''); see \cite{hintz_sharp_2022, angelopoulos_prices_2023} and the discussion in \cref{sec:prevresults}.\footnote{Note that when $g=g_M$, the coefficient $\integral_0[\phi]$ in \eqref{eq:latetimerough1} agrees with the time-inverted Newman--Penrose constant introduced in \cite{angelopoulos_late-time_2018}, where it is denoted ``$I_0^{(1)}[\psi]$'', up to a factor 2; see also \cite{angelopoulos_late-time_2021}[Corollary 9.12]. The factor 2 originates from a different choice of time functions: the time functions in \cite{angelopoulos_late-time_2018,angelopoulos_late-time_2021} can be obtained from the time functions $\tau$ of the present paper by rescaling $\tau\mapsto 2\tau$. }
	
	\end{rem}
	
	\begin{rem}[Initial vs dynamical contributions in the coefficients ${\integral_{0}[\phi]}$ and ${\integral_{\ell,m}[\phi]}$]
\label{rem:initialvsdyn}
Note that the expression \cref{eq:intro:I0} for $\integral_0[\phi]$ contains an \emph{initial contribution}, the integrals over $\Sigma_{\tau_0}$ and $\mathcal{H}^+\cap \Sigma_{\tau_0}$, and a \emph{dynamical contribution}, namely the integrals over $\mathcal{H}^+$ and $\mathcal{I}^+$. If $g=g_M$, then the dynamical contribution vanishes, because $\mathfrak{m}\equiv M$ and $\tr\chi|_{\mathcal{H}^+}=0$. For $|g-g_M|$ suitably small, the initial contribution will generically be non-zero and dominate the dynamical contribution.
While cleaner, the alternative representation \cref{eq:intro:I0b}, which follows from an application of the divergence theorem, obfuscates this distinction between initial data and dynamical contributions.

In contrast, the constants $\integral_{\ell, m}[\phi]$ with $\ell\geq 1$ are \emph{purely dynamical} and vanish in the case $g=g_M$, which is why the tails in \cref{eq:intro:asymptoticatScrip,eq:intro:asymptoticatr0} decay one power slower when compared to Price's law. See \cref{eq:main:thm:integralexpression} and \cref{eq:main:thm:integral:l=1} for precise expressions for $\integral_{\ell, m}[\phi]$.
\end{rem}

	\begin{rem}[Mode coupling and the restriction to $\ell<2+\delta$]\label{rem:intro:modecoupling}
The expressions for $\integral_0[\phi]$ and $\integral_{\ell,m}[\phi]$ illustrate the presence of mode coupling, which is absent in the case $g=g_M$. 
For instance, for $\ell>0$, the expression for $\integral_{\ell,m}[\phi]$ features terms that are of the schematic form:
\begin{equation}
	\int_{\scrip} \ell [``(g-g_M)\cdot \partial\phi\cdot \partial Y_{\ell,m} \text{''}]^{(\ell+1)},
\end{equation}
where $[...]^{(\ell+1)}$ denotes the $r^{-\ell-1}$-coefficient in the $1/r$-expansion of $[...]$ towards $\scrip$. 
Since by \cref{eq:intro:asymptoticatScrip}, $r\phi|_{\scrip}\sim \tau^{-2}$, which can easily be shown to imply that $[\partial \phi]^{(\ell)}\sim \tau^{-3+\ell}$, we thus have that $ [``(g-g_M)\cdot \partial\phi\cdot \partial Y_{\ell,m} \text{''}]^{(\ell+1)}\sim \tau^{-3-\delta+\ell}$. This is integrable only for $\ell<2+\delta$.

We note that such mode-coupling would be suppressed  in strength in a quasilinear setting where $g$ depends on $\phi$, as we would then be able to improve the time-decay of $g-g_M$ as a consequence of improved decay for $\phi$ and, in particular, of higher $\ell$-modes of $\phi$. See also \cref{thm:main:phi3} along with \cref{rem:main:mode2}.

%

This mode coupling may be contrasted with the mode coupling with respect to Boyer--Lindquist spheres for $g=g_{M,a}$, with $g_{M,a}$ the rotating sub-extremal Kerr metric, derived in \cite{angelopoulos_late-time_2021}, where high-$\ell$ modes are coupled with higher-order \emph{time derivatives} of low-$\ell$ modes. 
\end{rem}
\begin{rem}[Generalisations]
\label{rem:introgennonlin}
As we will discuss further in \cref{sec:intro:gen}, while \cref{thm:mainrough} treats a specific example, our general method remains applicable also when:
\begin{itemize}
\item $g_M$ is replaced with $g_{M,a}$, the rotating, sub-extremal Kerr metric (in view of \cite{angelopoulos_late-time_2021}), 
\item \cref{eq:intro:wave} is replaced with the quasilinear wave equation $\Box_{g(\phi,\partial \phi)} \phi=F(x,\phi,\partial \phi)$, assuming suitable conditions on the nonlinearities and a priori decay estimates on $\phi$ and its derivatives, consistent with small-data global existence, see for example \cite{lt18, lindblad_weak_2022, dafermos_quasilinear_2022}. See the discussion below  \cref{thm:main:F} and the nonlinear example equations in  \cref{thm:main:phi3,thm:main:phi4,thm:main:puphipvphi}. 
\item $g_M$ is replaced by a higher-dimensional Schwarzschild or Minkowski spacetime; see also \cref{sec:prevresultsinvsq}.
\end{itemize}
In the interest of exposition, we will not consider the above extensions in great detail, as the key relevant new phenomena are already present in the study of \cref{eq:intro:wave}.
\end{rem}

\begin{rem}[Alternative method: \cite{luk_late_2024}]
	In \cite{luk_late_2024}, a method for proving the precise late-time asymptotics for a very general class of equations and metrics (excluding, however, the case of even space dimensions) was developed. This method has been announced to also apply to the setting of \cref{thm:mainrough}, see \cite{luk_late_2024}[Example 3.9].  We compare and contrast the method of \cite{luk_late_2024} with the method of the present paper in \cref{sec:difflukoh}.
\end{rem}

\subsection{Previous results}
\label{sec:prevresults}
In this section, we provide an overview of previous results in the literature pertaining to the analysis of precise late-time asymptotics and late-time tails of waves on asymptotically flat black hole backgrounds. The study of late-time tails has a rich history in the physics literature, starting from the heuristic work of Price in \cite{price_nonspherical_1972}, which featured decay rates that have since been dubbed ``Price's law''. The mathematical study of late-time tails was preceded by a plethora of works in the mathematics literature providing robust, uniform \emph{upper bound} decay estimates in linear and stationary, but also dynamical and nonlinear settings. These results play a fundamental role in the mathematically rigorous derivations of late-time tails of the present paper, and we will refer to them in the main proofs. 
In contrast, in the literature overview below, we will focus on recent developments on mathematically rigorous proofs of the existence and nature of late-time tails. We refer to the introductions of \cite{angelopoulos_late-time_2018, angelopoulos_late-time_2021,luk_late_2024} for a more extensive overview of related literature.

\subsubsection{Linear wave equations on stationary backgrounds}
\label{sec:prevresultslinwave}
The existence of late-time tails in the leading-order late-time asymptotics of solutions to the linear wave equation on sub-extremal Reissner--Nordstr\"om black hole backgrounds was first rigorously proved in \cite{angelopoulos_vector_2018, angelopoulos_late-time_2018}. In these works, precise late-time asymptotics were obtained using physical-space-based methods by exploiting the existence of conserved Newman--Penrose charges\footnote{In the present paper, we do not make use of such conservation laws!} along $\mathcal{I}^+$ \cite{np65,np68} and introducing the notion of \emph{time-inverted} Newman--Penrose charges. Let us also note the earlier work of Luk--Oh \cite{luk_proof_2017}, who derived lower bounds for $L^2$-integrals along $\mathcal{H}^+$ that are consistent with Price's law and already feature the constant $\integral_0[\phi]$ in the form \cref{eq:intro:I0b}.

The existence of late-time tails on sub-extremal Kerr backgrounds was proved in \cite{hintz_sharp_2022,angelopoulos_late-time_2021}, with the former employing spectral methods centred around the regularity of resolvent operators at the zero frequency and the latter extending the physical-space based methods from \cite{angelopoulos_vector_2018, angelopoulos_late-time_2018}. The full $\ell$-dependence in Price's law was obtained in Schwarzschild in \cite{hintz_sharp_2022, angelopoulos_prices_2023} and explored in the Kerr setting in \cite{angelopoulos_late-time_2021}. The methods of \cite{angelopoulos_vector_2018, angelopoulos_late-time_2018,angelopoulos_late-time_2021} have also been extended to the setting of the non-zero spin Teukolsky wave equations on Schwarzschild \cite{mazh22} and Kerr \cite{ma_sharp_2023} (after applying also \cite{shlcosta23} outside the very slowly-rotating regime). See also \cite{millet_optimal_2023} for a derivation of late-time tails for Teukolsky wave equations using spectral methods.

\subsubsection{Conformal regularity vs. late-time tails}
\label{sec:prevresultsnopeeling}
In the works mentioned in \cref{sec:prevresultslinwave}, it is assumed that the initial data for $r\phi$ are conformally regular  at $\mathcal{I}^+$ (i.e.~they admit $1/r$-expansions to sufficiently high order) and are possibly even more rapidly decaying (e.g.~compactly supported, as in our \cref{thm:mainrough}). In the series of works \cite{kehrberger_case_2022, kehrberger_case_2021, kehrberger_case_2022-1,kehrberger_case_2024-1, kehrberger_case_2024}, it was shown that initial data modelling the gravitational radiation arising from $N$ infalling masses with no incoming radiation cannot be expected to be conformally smooth. At the level of metric curvature components, this means that the spacetime \emph{peeling property} fails. 

In \cite{gajic_relation_2022}, we observed that for sufficiently conformally irregular data, the late-time asymptotic behaviour is entirely determined by the conformal irregularity of the initial data, i.e.~tails originating from the conformal irregularity of the initial data dominate the tails satisfying Price's law.

In particular, we put forth heuristics explaining why, for solutions to \cref{eq:intro:Einstein} settling down to a sub-extremal Kerr, the late-time asymptotics of the Newman--Penrose scalar $r\Psi^{[4]}|_{\mathcal{I}^+}$ at future null infinity should feature a $u^{-3}$-tail in the case of infalling masses following hyperbolic orbits in the infinite past, and a $u^{-11/3}$-tail in the case of parabolic orbits. 
This strongly deviates from the linear $u^{-6}$-tail in the case of compactly supported initial data derived in \cite{mazh22,ma_sharp_2023}, and also from the expected $u^{-5}$-tail in the nonlinear theory.

At first sight, the existence of late-time tails coming from conformally irregular initial data may seem more straightforward than the existence of late-time tails for compactly supported initial data.\footnote{Indeed, already in the setting of the wave equation on Minkowski in 3+1 dimensions, conformally irregular data lead to tails.} 
However, even for compactly supported initial data, the global \emph{time integral} of the corresponding solutions to $\Box_g\phi=0$ generically has conformally irregular initial data. This irregularity then determines the global leading-order asymptotic profile for the time integral, and thus, after application of a time derivative, also of the original solution. For this reason, the ideas in \cite{gajic_relation_2022} still play an important role in the present paper.

\subsubsection{Dynamical backgrounds and nonlinear wave equations}
\label{sec:intrononinear}
 In the context of \textit{linear} wave equations on \textit{dynamical} spacetime backgrounds, a deviation of Price's law consistent with \cref{thm:mainrough} was first suggested by the numerics in \cite{gupipu94}. This deviation was, however, only noticed, clarified and revisited, both numerically and heuristically, in \cite{biro10}. Similar deviations have been observed numerically in various \textit{nonlinear} models; see for example \cite{bcr07a,bcr07b,okuzumi}.
 
 The role of dynamical backgrounds and nonlinearities in the decay rates of late-time tails was first understood mathematically by Luk--Oh in the monumental work \cite{luk_late_2024} (see, however, also their previous \cite{luk_quantitative_2015,luk_strong_2019}, where the constant $\integral_0[\phi]$ already plays an important role for deducing \textit{lower bounds}), where a very general procedure for deriving late-time tails was introduced, starting from weak, a priori decay assumptions and elliptic-type properties of the stationary linear wave operator without time derivatives, i.e., in the context of \cref{eq:intro:wave}, $\square_{g_M}$ without $\partial_{\tau}$ derivatives. The work \cite{luk_late_2024} also covers the case of higher even-dimensional spacetime dimensions. In \cref{sec:difflukoh}, we list some of the key differences and similarities between the methods of \cite{luk_late_2024} and the methods introduced in the present paper. 
 
 In the context of wave equations on stationary spacetimes with power nonlinearities, late-time tails have been obtained in \cite{looi_asymptotic_2024} by extending the spectral methods developed in \cite{hintz_sharp_2022}.

 \subsubsection{The Einstein equations}
\label{sec:introeinstein}
Late-time tails in the context of the \textit{spherically symmetric} Einstein--scalar field system were studied numerically in \cite{gupipu94} and were in agreement with Price's law. This led to the incorrect belief that Price's law remains valid more generally in nonlinear and dynamical settings. It was then pointed out in \cite{bichm09} that the agreement with Price's law is a special property of the $3+1$-dimensional spherical symmetric setting. See also Remark \ref{rem:initialvsdyn}. The existence and nature of late-time tails for the Einstein--scalar field system was mathematically proved in \cite{gautam_late-time_2024}; see also the upcoming \cite{rutgprice} which, in particular, produces a class of metrics satisfying the assumptions of \cref{thm:mainrough}.

In the context of the Einstein equations outside of spherical symmetry, the current state of the art are upper bound decay estimates for metric components and their derivatives as part of nonlinear stability proofs \cite{klsz20, dafermos_non-linear_2021,klainerman_kerr_2021, gks24}. In \cref{app:motivation}, we provide an outlook on how stability results of \cite{dafermos_non-linear_2021} may be extended to obtain metrics $g$ that satisfy the assumptions in this paper. The generalisation of our methods to \cref{eq:intro:Einstein} will be the content of future work.

\subsubsection{Inverse-square potentials, odd spacetime dimensions and extremal black holes}
\label{sec:prevresultsinvsq}
The decay rates for solutions to (homogeneous) wave equations on Schwarzschild and Kerr spacetimes are closely connected to the existence of Newman--Penrose conserved charges at future null infinity, which in turn is closely related to the validity of the strong Huygens principle in even spacetime dimensions. By considering wave equations with asymptotically inverse-square potentials or odd spacetime dimensions, these conservation laws are broken. In \cite{gajic_late-time_2023}, a more robust, physical-space-based method was developed for deriving precise late-time asymptotics that applies also to these more general settings, relying on time integrals, like \cite{angelopoulos_vector_2018, angelopoulos_late-time_2018,angelopoulos_late-time_2021}, but circumventing the use of conserved charges at future null infinity and instead applying the subtraction of global tail functions. The method introduced in \cite{gajic_late-time_2023} plays an important role in the present paper, as outlined in \cref{sec:sketchproof} below. See also \cite{hin23} for an alternative derivation of late-time tails in a related, general setting using spectral methods.

Late-time asymptotics become even richer in settings where the absence of conservation laws and the phenomenon of superradiance are coupled. This occurs for wave equations on extremal Kerr black holes and for charged scalar field equations on (extremal or sub-extremal) Reissner--Nordstr\"om backgrounds. The existence and nature of late-time tails was proved in those settings in \cite{gajic_azimuthal_2023, gaj25a, gaj25b} using extensions of the method introduced in \cite{gajic_late-time_2023}, and it was shown that the tails are oscillating. See also \cite{gajic_late-time_2024} for a very different approach to derive late-time tails of a similar nature on Minkowski. We note that in the setting of uncharged scalar fields on extremal Reissner--Nordst\"om, there is no superradiance,  and conservation laws associated to the Aretakis constants \cite{are15} are present. In this setting, late-time tails were proved in \cite{aag20}. 

\subsection{A sketch of the method for $\Box_{g_M}\phi=F$}
\label{sec:sketchproof}
In this section, we sketch our general method of deriving precise late-time asymptotics and explain how it applies, in particular, to the proof of \cref{thm:mainrough}.
For the sketch, we will work with the following wave equation on Schwarzschild:\begin{equation}\label{eq:intro:inhomwave}
	\Box_{g_M} \phi =F[x,\phi,\partial \phi, \partial^2\phi],
\end{equation}
where $F$ may depend, linearly or nonlinearly, on $\phi$.
Notice that \cref{eq:intro:inhomwave} covers the case of $\Box_g\phi=0$ via:
\begin{align}\label{eq:intro:F:G}
	F[\phi,\partial \phi, \partial^2\phi]=(\Box_{g_M}-\Box_g)\phi=-\frac{1}{\sqrt{-\det g_M}}\partial_{\alpha}(G^{\alpha \beta} \partial_{\beta}\phi),\quad 
	G^{\alpha \beta}:=\:g^{\alpha \beta}\sqrt{-\det g}-g_M^{\alpha \beta}\sqrt{-\det g_M}.
\end{align}
Writing $\Box_g\phi=0$ in the form \cref{eq:intro:inhomwave} is not suitable for deriving a preliminary decay (or stability) estimate. However, since our method is to be applied only \emph{after} such an estimate has been obtained, we may freely treat the problem essentially as an inhomogeneous problem on Schwarzschild.

If $F$ is a genuine inhomogeneity (not depending on $\phi$), then our method applies provided $F$ satisfies:\footnote{We stress that it is essential that our method allows $F$ to be non-compactly supported, as the $F$ appearing in non-linear problems modelling the Einstein equations will not be compactly supported in general. This stands in sharp contrast with classical studies in the physical literature of late-time tails from a frequency-space perspective, see for example \cite{leav86}, where the compact support of $F$ is crucial and is related to the fact that standard resolvent operators must be conjugated with cut-off functions in order to meromorphically continue them to the lower complex half-plane. In these studies, late-time tails are connected to the behaviour of the cut-off resolvent near a branch cut of the meromorphic continuation of the cut-off resolvent; see also the discussion in \cite{gajic_quasinormal_2024}[Sections 3.3--3.5] for more details. Such methods are therefore \emph{not} suitable for a late-time tail analysis in nonlinear settings.}

\begin{enumerate}[label=(\Alph*)]
	\item $r$-weighted polynomial time-decay estimates, which are moreover preserved  under the action of $\{r\partial_r|_{\Sigma},\sl\}$.\label{A}
	\item Additional regularity w.r.t.~$\tau T$, i.e.~the decay estimates for $F$ in (A) are preserved when acting with the weighted vector field $\tau {T}$, $T$ denoting a Killing vector field on Schwarzschild generating time translation symmetry.\label{B}
\end{enumerate}
Here, $\sl$ denotes the covariant derivative on the unit sphere, and $r\partial_r|_{\Sigma}$ the radial derivative tangent to $\Sigma_{\tau}$; in particular, $r\partial_r|_{\Sigma}$ equals $rD^{-1}\pv$ close to $\scrip$ in double null coordinates $u, \, v$, for $D=1-2M/r$.

If $F$ depends on $\phi$, then \cref{A} and \cref{B} would follow from analogous decay estimates for $\phi$.\footnote{The reason for stating \cref{A} and \cref{B} separately is that we may think of \cref{A} for $\phi$ as a quantitative stability estimate, whereas \cref{B} contains additional information about higher-order time derivatives typically not necessary for proving stability.} 

We will first give a sketch of our method for a general inhomogeneity $F$ not depending on $\phi$, and discuss at the end the additional iterations that are necessary if $F$ depends on $\phi$.
\subsubsection{Step 0: Decay in time via integrated energy estimates (the $r^p$-hierarchy)}
In this step, we recall the $r^p$-weighted energy method of \cite{dafermos_new_2010} to derive energy decay. 
Let, for $\rc$ sufficiently large, $\mathcal{C}_{\tau}=\{u=u_0+\tau\}\cap\{r\geq\rc\}$ and $\underline{\mathcal{C}}_{\tau}=\{v=u_0+r_*(\mathring{R})+\tau=v_0+\tau\}\cap\{r\leq R\}$ be intersecting null hypersurfaces in Schwarzschild. We will moreover assume that $\tau=u$ for $r\geq\rc= r(u_0,v_0)$.
For the sketch of the $r^p$-weighted energy method, we restrict to spherically symmetric~$\phi$.\footnote{This simplification is only made in this sketch, so that we do not have to feature a loss of derivatives in integrated decay estimates due to null geodesics trapped at the photon sphere; see for example \cite{sbierski_characterisation_2015}. } Let $\dd\sigma$ denote the volume form on $S^2$, the unit round sphere, and define the non-degenerate $p$-weighted energy of $\phi$ as follows: 
\begin{equation*}
	E_p[\phi](\tau):=\int_{\underline{\mathcal{C}}_{\tau}}\left[D^{-2} |\partial_u(r\phi)|^2+\phi^2 \right]\dd \sigma \dd r+\int_{\mathcal{C}_{\tau}}r^p\left[ |\partial_v(r\phi)|^2+\phi^2\right]\dd {\sigma} \dd r.
\end{equation*}
Energy boundedness is the statement that $E_0[\phi](\tau_B)\lesssim E_0[\phi](\tau_A)+\mathfrak{X}[F]$ for all $ \tau_A< \tau_B$, where $\mathfrak{X}[F]$ is an appropriately weighted, spacetime $L^2$ integral; see \cref{prop:inhom:energy}.

The inhomogeneous version of the \emph{Dafermos--Rodnianski hierarchy of $r^p$-weighted energy estimates} \cite{dafermos_new_2010} can be stated as follows (see \cref{prop:inhom:rp}): For $0\leq p\leq 2$, for $0\leq \tau_A<\tau_B\leq \infty$, and for $\epsilon>0$ arbitrarily small:
\begin{equation}
	\label{eq:introrpest}
	E_p[\phi](\tau_B)+\int_{\tau_A}^{\tau_B}p E_{p-1}[\phi](\tau)\,d\tau\lesssim E_p[\phi](\tau_A)+\mathfrak{X}[F]+\int_{\tau_A}^{\tau_B}\int_{\underline{\mathcal{C}}_{\tau}\cup \mathcal{C}_{\tau}}r^{\min(p,\epsilon)+3}|F|^2 \dd r \dd\tau+\mathfrak{X}[F].
\end{equation}
By taking $p=1$ and applying the mean-value theorem in $\tau$, together with energy boundedness, it follows that:
\begin{equation*}
	E_0[\phi](\tau)\lesssim (\tau+1)^{-1}\left[E_1[\phi](0)+\mathfrak{X}[F]+\int_{0}^{\infty}\int_{\underline{\mathcal{C}}_{\tau}\cup \mathcal{C}_{\tau}}r^{4}|F|^2 \dd r \dd\tau\right].
\end{equation*}
Hence, we obtain decay in time of the energy quantity $E_0$ if the right-hand side above is finite, which requires suitable decay of the initial data in the spatial coordinate $r$, as well as suitable decay of $F$ in $r$ and $\tau$.

Assuming even more initial data decay in $r$ and decay of $F$ in $\tau$ and $r$, we can carry out the above argument with $p=2$ and then repeat the $p=1$ argument to obtain an additional power of $\tau^{-1}$:
\begin{equation*}
	E_0[\phi](\tau)\lesssim (\tau+1)^{-2}\left[E_2[\phi](0)+\mathfrak{X}[F]+\int_{0}^{\infty}\int_{\underline{\mathcal{C}}_{\tau}\cup \mathcal{C}_{\tau}}r^{5}|F|^2 \dd r \dd\tau\right].
\end{equation*}
This energy decay can easily be shown to imply the $L^{\infty}$-decay estimate: $|r\phi|\lesssim \tau^{-\frac{1}{2}}$. 
The above estimates remain valid (with a loss of derivatives) also for non-spherically symmetric $\phi$, and for commutations $Z^k \phi$ and $Z^k F$ for $Z\in\Vc= \{T, r\partial_r|_{\Sigma}, \sl\}$. 

Notice that the $\tau^{-2}$ energy-decay corresponds exactly to the length of the hierarchy of $r^p$-weighted estimates, i.e.~$2=\max(p)-\min(p)$. If the initial data only had finite $r^{p}$-energies with $p<2$, then we would have energy decay $\tau^{-p}$ and pointwise decay $r\phi\lesssim \tau^{\frac{p-1}{2}}$.

\subsubsection{Step 1: Improved decay for time derivatives and the relevance of time integrals $T^{-n}\phi$}
 Improved decay for time derivatives via commutation with $r\partial_r|_{\Sigma}$ was first observed in \cite{schlue_decay_2013, moschidis_rp_2016} and used in combination with elliptic-type estimates to show that $|\phi|\lesssim \tau^{-\frac{3}{2}}$.
In \cite{angelopoulos_vector_2018}, this method was extended to obtain faster decay for arbitrarily many time derivatives; we give a brief sketch:
By taking the square norm of \cref{eq:intro:inhomwave} and integrating, we can obtain the following additional energy estimate:
\begin{equation}
	E_{p-1}[T\phi](\tau)\lesssim \sum_{Z\in\{r\partial_r|_{\Sigma}, \sl\}}E_{p-3}[Z\phi](\tau)+\int_{\Cbar_{\tau}\cup \C_{\tau}}r^{p+1}F^2\dd r
\end{equation}
\begin{equation}
	\label{eq:introconvtrder}
	\int_{\tau_A}^{\tau_B}E_{p-1}[{T}\phi](\tau)\,d\tau \lesssim\sum_{Z\in\Vc} \sum_{k\leq 1}\int_{\tau_A}^{\tau_B}E_{p-3}[Z^k\phi](\tau)\,d\tau+\int_{\tau_A}^{\tau_B}\int_{\underline{\mathcal{C}}_{\tau}\cup \mathcal{C}_{\tau}} r^{p+1}|Z^kF|^2\dd r\dd\tau.
\end{equation}
Taking $p=3$ and $p=4$ in \cref{eq:introconvtrder} and applying \cref{eq:introrpest} with $\phi$ replaced by ${T}\phi$ then leads to $\tau^{-4}$ decay for $E_{0}[{T}\phi](\tau)$. The inequality \cref{eq:introconvtrder} therefore allows us to obtain additional decay for extra time derivatives, at the expense of commuting with vector fields in $\Vc$ and requiring additional decay in time for ${T}F$ (cf.~\cref{B}). 
Moreover, by considering additional, $r$-weighted elliptic estimates (see \cref{prop:inhom:elliptic1}) and assuming suitably fast decay for $F_{\geq \ell}$, we may also infer that $r^{-\ell}\phi_{\geq \ell}\lesssim \tau^{-\ell-3/2}$.

While the above argument gives improved decay for ${T}^k\phi$, it does not directly lead to improved decay for~$\phi$. 
In order to make use of the improved decay for time derivatives, we write: $\phi=T^{k}(T^{-k}\phi)$, for $T^{-k}$ the $k$-th time-integral. 
Since we already have pointwise decay $\phi\lesssim \tau^{-3/2}$, we may directly define the time integral of $h=\phi, F$ as follows:
\begin{equation}
	T^{-1}h=-\int_{\tau}^{\infty} h \dd \tau', \qquad \implies \qquad\Box_{g_M}T^{-1}\phi=T^{-1}F.
\end{equation}
(Alternatively, we may also define $T^{-1}\phi$ \textit{without prior decay knowledge of $\phi$} by first inverting a suitable elliptic operator---see \cref{eq:intro:l=0wave} below---along the initial hypersurface to define initial data for $T^{-1}\phi$ and then evolving this data according to $\Box_{g_M}T^{-1}\phi=T^{-1}F$ to obtain $T^{-1}\phi$.) We define higher-order time integrals inductively via $T^{-k-1}\phi=T^{-1}T^{-k}\phi$.

Thus, if we can show that we can apply the hierarchy of $r^p$-estimates with $p=2$ to $T^{-k}\phi$, which, in particular, requires $E_2[T^{-k}\phi](\tau_0)<\infty$, then we can deduce that $r\phi\lesssim \tau^{-3/2-k}$. 
We have therefore reduced the question of finding sharp \textit{time decay} estimates for $\phi$ to the question of understanding the \textit{asymptotic behaviour of the initial data} for $T^{-k}\phi$.
\subsubsection{Step 2: Asymptotics of the initial data of $T^{-k}\phi$ and almost-sharp decay for $\phi$}
We now sketch how to find the initial data asymptotic behaviour of  time integrals for fixed $\ell$-modes, assuming the initial data for $\phi$ to be compactly supported along $\Cbar_0\cup \C_0$. First, we consider $\ell=0$, then we consider higher $\ell$.
For $\Xx=\partial_r|_{\Sigma}=D^{-1}\pv$ the radial derivative along $\C_{\tau}$, the wave equation $\Box_{g_M}T^{-1}\phi=T^{-1}F$ can be written as:
\begin{equation}\label{eq:intro:l=0wave}
\mathcal{L}T^{-1}\phi:=	\Xx(r^2 D\Xx T^{-1}\phi)+\Dl T^{-1}\phi=r^2T^{-1}F+r\Xx(r TT^{-1}\phi).
\end{equation}
\textbf{Projecting onto $\ell=0$,} we can integrate $\mathcal{L}T^{-1}\phi$ from the sphere  $\S_{0}^{\rc}=\Cbar_{0}\cap \C_{0}$ to find:
\begin{equation}\label{eq:intro:T-1l0}
	r^2D\Xx T^{-1}\phi_{\ell=0}(u_0,r)=	r^2D\Xx T^{-1}\phi_{\ell=0}(u_0,\rc) +\int_{\rc}^{r} r'\Xx(r'\phi_{\ell=0}) \dd r'-\int_{\rc}^r\int_{\tau_0}^{\infty}{r'}^2F_{\ell=0 }\dd \tau' \dd r'.
\end{equation}
There is an analogous expression for the term $ r^2D\Xx T^{-1}\phi_{\ell=0}(u_0,\rc)$ purely in terms of $\phi|_{\Cbar_0}$ and the inhomogeneity~$F$. 

Now, assuming that the initial data for $\phi$ are compactly supported and that $F$ has an expansion
\begin{equation}\label{eq:intro:F:expansion}
	F-\sum_{i=3}^n \frac{F^{(i)}(u,\theta^A)}{r^i}\lesssim r^{-n-} u^{-\beta-3+n+}, \qquad \text{for} \quad F^{(i)}\lesssim u^{-\beta-3+i},
\end{equation}
with $n\geq 3$  and $\beta>1$, we may deduce from \cref{eq:intro:T-1l0} that
\begin{equation}
 r^2D \Xx T^{-1}\phi_{\ell=0}(u_0,r) \sim -\log r \int_{\scrip} F^{(3)}\dd u \dd \sigma \implies \Xx(rT^{-1}\phi_{\ell=0})\sim r^{-1}  \int_{\scrip} F^{(3)}\dd u \dd \sigma.
\end{equation}
In particular, we  have that $E_p[T^{-1}\phi_{\ell=0}](0)<\infty$ iff $p< 1$. Thus, writing $\phi=T T^{-1}\phi$, we may deduce $\tau^{-3+}$-decay for the energy (corresponding to a hierarchy of $r^p$-estimates of length $3-$), and $\tau^{-1+}$ pointwise decay for $r\phi_{\ell=0}$; this is sharp, up to an arbitrarily small loss.

On the other hand, if $F^{(3)}$ in \cref{eq:intro:F:expansion} vanishes, and if $\beta>2$ and $n\geq 4$, we find that
\begin{equation}\label{eq:intro:l=0F4}
	r^2D\Xx T^{-1}\phi_{\ell=0}(u_0,r)=C_0 +\frac{\int_{\scrip} F^{(4)}\dd u \dd \sigma}{r}+\dots \implies \Xx (rT^{-1}\phi_{\ell=0}) \sim  r^{-2} \int_{\scrip} F^{(4)}\dd u \dd \sigma+\frac{C'_0}{r^2},
\end{equation}
for $C_0, C'_0$  constants depending on initial data for $\phi$ and on the spacetime integral of $r^2F$. In particular, $E_p[T^{-1}\phi_{\ell=0}]<\infty$ for $p\leq 2$. 
However, if we now consider the initial data for the second time integral $T^{-2}\phi$, then the integral $ \int_{\rc}^{r} r\Xx(rT^{-1}\phi_{\ell=0}) \dd r'\sim \log r$ in \cref{eq:intro:T-1l0} will grow logarithmically, and we have that $E_{p}[T^{-2}\phi_{\ell=0}]=\infty$ for $p\geq 1$. 
This then results in $\tau^{-5+}$-decay for the energy, and pointwise $\tau^{-2+}$-decay for $r\phi$. These decay estimates are sharp, up to an arbitrarily small loss.

\textbf{For higher $\ell$,} we cannot directly integrate \cref{eq:intro:l=0wave} along $\C_{0}$ due to the term $\Dl T^{-1}\phi$. 
The function $w_\ell$ appearing in \cref{thm:mainrough} is non-vanishing everywhere and solves $\mathcal{L}w_\ell=0$. Note that $\square_{g_M}(w_{\ell}Y_{\ell,m})=0$, so $w_{\ell}Y_{\ell,m}$ may be interpreted as a stationary solution to the wave equation on Schwarzschild. We then \textit{twist} the operator $\mathcal{L}$ as follows using $w_{\ell}$:
\begin{equation}
\label{eq:introtwistop}
\check{\mathcal{L}}_{\ell}\philc:=\Xx(r^2 D w_\ell^2 \Xx (\philc))=w_{\ell}\mathcal{L}\phi_{\ell}, \qquad \text{for } \philc=w_{\ell}^{-1}\phi_\ell.
\end{equation}
Thus, by projecting \cref{eq:intro:l=0wave} onto fixed $\ell$ and integrating along $\C_{0}$, we obtain:
\begin{equation}\label{eq:intro:higherell}
	r^2Dw_{\ell}^2 \Xx T^{-1}\philcm(u_0,r)=	r^2Dw_{\ell}^2 \Xx T^{-1}\philcm(u_0,\rc)+\int_{\rc}^r r' w_{\ell}\Xx(r'\phi_{\ell,m}) \dd r'-\int_{\rc}^r \int_{\tau_0}^{\infty} {r'}^2w_{\ell} F_{\ell,m} \dd \tau \dd r'.
\end{equation}
In particular, if \cref{eq:intro:F:expansion} holds with $n\geq \ell+3$ and with $\beta>\ell+1$, then we obtain (the $\dots$-terms denote lower-order terms in the $1/r$-expansion)
\begin{multline}\label{eq:intro:owe}
	\Xx T^{-1}\philcm \sim1+\dots+ \int_{\scrip}(r^{-\ell}w_{\ell}F_{\ell,m})^{(3+\ell)} \dd u\cdot  r^{-2-2\ell}\log r\\ \implies rT^{-1}\phi_{\ell,m}\sim 1+\dots+ \int_{\scrip}(r^{-\ell}w_{\ell}F_{\ell,m})^{(3+\ell)} \dd u\cdot r^{-\ell}\log r.
\end{multline}
In particular, for $\ell>0$, the first time integral $T^{-1}\phi_{\ell}$ has finite $p=2$ initial energy.
Re-inserting these asymptotics into \cref{eq:intro:higherell} for higher-order time integrals, the integral over $rw_{\ell}\Xx(r\phi_\ell)$ then pushes this log-term one order forward for each time integral. We find that
\begin{equation}
	rT^{-\ell-1}\phi_{\ell,m}\sim    \int_{\scrip}(r^{-\ell}w_{\ell}F_{\ell,m})^{(3+\ell)} \dd u \cdot \log r \implies \Xx(rT^{-\ell-1}\phi_{\ell,m})\sim r^{-1}  \int_{\scrip}(r^{-\ell}w_{\ell}F_{\ell,m})^{(3+\ell)} \dd u .
\end{equation}
This means that the initial energy $E_p[T^{-\ell-1}\phi_{\ell}](\tau_0)<\infty$ iff $p<1$, so we deduce (sharp) energy decay $E_0[\phi_\ell]\lesssim \tau^{-2\ell-3+}$ and pointwise decay $r\phi_\ell\lesssim \tau^{-\ell-1+}$.

We can complement the fixed angular mode computations above with a high-$\ell$ elliptic estimate (see \cref{prop:inhom:elliptic2}) to conclude summability in $\ell$ (i.e.~that $\phi_{>\ell}$ decays faster). We can moreover obtain improved decay for higher $\ell$-modes away from $\scrip$ using further elliptic estimates (see \cref{prop:inhom:elliptic1}).
 We thus get that if $\beta>\ell+1$ and $n\geq\ell+3$ in \cref{eq:intro:F:expansion}:
 \begin{equation}
 	r\phi_{\geq\ell}\lesV \tau^{-\ell-1}\min_{q\in[0,\ell+1]} \left(\frac{r}{\tau}\right)^q.
 \end{equation}
This inequality moreover remains valid under commutations with vector fields in $\V=\{r\partial_r|_{\Sigma}, \tau T, \sl\}$.

\subsubsection{Step 3: Precise asymptotics for $\phi$ via subtraction of a global tail function}
\label{sec:introstep3}
The obstruction to proving faster decay in the above steps arose from the conformal irregularity of the initial data of the time integral. On the other hand, it follows from \cite{gajic_relation_2022} that for conformally irregular data, the global asymptotic profile of the solution is fixed completely by the initial data, which is related to the fact that this conformal regularity is conserved along null infinity.

We make use of these insights by subtracting a global approximate solution $T^{-1}\apphi$ whose initial data have exactly the same (conformally irregular) logarithmic term as \cref{eq:intro:owe}. Since the difference $\phi-\phi^{\mathrm{app}}$ then has a time integral with better conformal regularity, it turns out that we can infer that the $(\ell+1)$-th time integral satisfies $E_p[T^{-\ell}(\phi_{\ell}-\apphi[\ell])](\tau_0)<\infty$ for $p\leq 2$, which means that we can prove one power better energy decay (and half a power better pointwise decay) for the difference. Since $\square_{g_M}(T^{-1}\apphi)\neq 0$, this procedure generates an additional inhomogeneity, but this inhomogeneity does not form an obstruction to faster decay. In particular, we may infer, for $\ell>0$, $r\phi_{\ell,m}-r\apphi[\ell,m]\lesssim \tau^{-\ell-3/2}$. Since $r\apphi\sim \tau^{-\ell-1}$, this (together with elliptic estimates away from $\scrip$) proves the global asymptotics for $\phi_{\ell,m}$, namely
\begin{equation}
	r\phi_{\ell,m}\sim r\apphi \int_{\scrip} F^{(\ell+3)} \dd u.
\end{equation}

Notice that for $\ell=0$, in the case where $F^{(3)}=0$, we obtain \cref{eq:intro:l=0F4}, as the bad logarithmic term only appears for the second time integral: Thus, rather than subtracting $\apphi[\ell=0]$ from $\phi_{\ell=0}$, we simply subtract $T\apphi[\ell=0]$ from $\phi_{\ell=0}$.

We can motivate the choice of approximate function $\apphi$ as follows. Any constant multiple of the function $r^{\ell}\tau^{-1-\ell}(\tau+r)^{-1-\ell}$ is a solution to the wave equation on Minkowski and has a time integral with the desired behaviour in $r$. Since $g_M$ is asymptotically flat, $\Box_{g_M}(r^{\ell}\tau^{-1-\ell}(\tau+r)^{-1-\ell})$ therefore decays suitably fast in $r$. In order to obtain a function $\apphi$ for which $\square_{g_M}(\apphi Y_{\ell,m})$ also decays suitably in $\tau$ globally, we recall the twisted elliptic operator $\check{\mathcal{L}}_{\ell}$ from~\cref{eq:introtwistop}. Define $\apphi=w_{\ell}\tau^{-1-\ell}(\tau+r)^{-1-\ell}$, then:
\begin{equation*}
w_{\ell}r^2\square_{g_M}(\apphi Y_{\ell,m})=\tau^{-1-\ell}\check{\mathcal{L}}_{\ell}((\tau+r)^{-1-\ell}) Y_{\ell,m}-rw_{\ell}\Xx(r T \apphi)Y_{\ell,m}.
\end{equation*}
Observe that both the terms on the right-hand side above decay faster in $\tau$ than $\apphi$. By replacing $w_{\ell}$ with $\chi(\frac{r}{\tau})r^{\ell}+(1-\chi(\frac{r}{\tau}))w^{\ell}$, with $\chi$ a suitable cut-off function, we can in fact estimate globally (see \cref{sec:approx} for more details):
\begin{equation*}
 |\Box_{g_M}(\apphi Y_{\ell,m})| \lesssim \frac{1}{(\tau+r)r^2}|\apphi|.
 \end{equation*}

In the setting of inhomogeneities discussed in this sketch, making the arguments above precise results in  \cref{thm:main:F}.

\subsubsection{The case where $F$ depends on $\phi$}
In the sketch above, we took $F$ to be a given inhomogeneity with given decay. 
If $F$ depends on $\phi$ and its derivatives, we need to start out with an initial decay estimate on $\phi$ and its derivatives, which translates into a weak decay estimate on $F$. 
Depending on the decay of $F$, we can then use the above procedure to improve the decay estimate for $\phi$ and its derivatives, which in turn improves the decay of $F$ via an iteration argument. See also the example in the \cref{sec:intro:gen} as well as the more extensive discussion of nonlinear examples in \cref{sec:main:additional}.

In the specific case where $F$ takes the form \cref{eq:intro:F:G}, we then obtain \cref{thm:mainrough} by also exploiting the divergence structure of $F$.

\subsection{Generalisations, further results and future work}\label{sec:intro:gen}
We mention here several generalisations of the results and methods of the present paper.
\paragraph{Conformally irregular data or inhomogeneity:} In the sketch above, we have taken $F$ to have a smooth $1/r$-expansion \cref{eq:intro:F:expansion}, and we have assumed $\phi$ to be initially compactly supported. 
If either of these assumptions are dropped, then we simply have another source of conformally irregular terms appearing in the initial data of suitable time integrals, and the method applies without further changes (except that one may need to take fewer time integrals and subtract different approximate solutions).\footnote{In the physically relevant case where all $\ell$-modes decay the same \cite{kadar_scattering_2025}, one needs an additional argument addressing the issue of summability.}

\paragraph{Nonlinearities:} As discussed above, via iteration, the method applies also to $F(x,\phi,\partial\phi,\partial^2\phi)$, provided an initial global existence and weak quantitative decay result holds. 
Rather than providing a general class of nonlinearities, we simply treat a few examples covering the case of null form nonlinearities and power nonlinearities (\cref{thm:main:phi3,thm:main:phi4,thm:main:puphipvphi}). 
\paragraph{Inverse-square potentials and higher dimensions:} The key ingredients of our method, namely $r^p$-weighted energy estimates, time integrals and subtraction of approximate solutions, have also been shown to hold for $\Box_{g_M}\phi=V(r)\phi$, where $V(r)\sim \alpha r^{-2}$ as $r\to \infty$, with $\alpha>-\frac{1}{4}$, in \cite{gajic_late-time_2023}.\footnote{In this setting, one can in general only take $r^p$-estimates for $p<2$, but the method still applies. Moreover, one needs to introduce an appropriate twisted operator $\check{\mathcal{L}}$ that accounts for the potential.} This setting covers the case of higher even and odd spacetime dimensions: The methods in the present paper therefore extend also to these settings. 

\paragraph{Kerr backgrounds:}
By the same reasoning as above, in view of \cite{angelopoulos_late-time_2021}, our methods also apply to quasilinear perturbations of the wave equation on Kerr.

\paragraph{The Einstein equations:}
The class of dynamical metrics $g$ that we allow for in \cref{thm:mainrough} should in particular include the metrics constructed in the stability proof \cite{dafermos_non-linear_2021}, see \cref{app:motivation}. 
In future work, we will discuss how to apply the methods of the present paper within the framework of the Regge--Wheeler/Teukolsky equations to infer the precise late-time asymptotics for solutions to \cref{eq:intro:Einstein} settling down to Schwarzschild.

\paragraph{Higher-order asymptotics:}
Our methods can be generalised to include higher-order terms in the asymptotic expansion in $\frac{1}{\tau}$ of $\phi$ as $\tau\to\infty$. This will be the subject of future work. 
See also \cref{rem:main:higher}.

 \subsection{A short comparison with \cite{luk_late_2024}}
 \label{sec:difflukoh}
 Since the results of the present paper and \cite{luk_late_2024} are closely related, we list below some of the key differences and similarities between the methods employed in \cite{luk_late_2024} and the present paper. For simplicity, we will mainly focus on the setting of \cref{eq:intro:wave}. 
\begin{enumerate}[label=\arabic*),leftmargin=*]
\item Both \cite{luk_late_2024} and the present paper are purely based in physical space (with a partial decomposition into spherical harmonic modes). Both also take as an assumption decay estimates for the metric difference $g-g_M$ together with weak decay assumptions on $\phi$ that remain valid after commutation with $\{\tau T, r\partial_r|_{\Sigma},\sl\}$, and then use iteration schemes to gain improved decay estimates.
\item In both \cite{luk_late_2024} and the present paper, \cref{eq:intro:wave} is decomposed as a simpler, stationary wave equation with inhomogeneity.  In \cite{luk_late_2024}, \cref{eq:intro:wave} is treated as an inhomogeneous wave equation on \emph{Minkowski} $\square_{\eta} \phi=F$ in a large-$r$ region in order to apply the strong Huygens Principle.

In the present paper, however, \cref{eq:intro:wave} is instead treated globally as an inhomogeneous equation on \emph{Schwarzschild} $\square_{g_M}\phi=F$. The role of the Minkowski operator appears only in the explicit global tail function $\phi^{\mathrm{app}}$ used in the subtraction $\phi-\phi^{\mathrm{app}}$, as $\phi^{\mathrm{app}}$ is closely related to solutions to the Minkowski wave equation; see \cref{sec:introstep3}. This difference is already present in the stationary setting $g=g_M$.

\item In \cite{luk_late_2024}, the key mechanism for deriving late-time tails is an analysis of recurrence relations for higher-order radiation fields along $\mathcal{I}^+$, which is used to derive detailed information about $\frac{1}{r}$-expansions of $r\phi$ after a commutation with vector fields of the form $r^2\partial_r$ in Bondi-type coordinates. This expansion is then used to obtain sufficient control of $F$ when writing $\square_{\eta} \phi=F$ (as in 1)) and to derive late-time tails in the large-$r$ region via the strong Huygens principle. The late-time decay is determined by the lowest $j$ satisfying $\lim_{u\to\infty}(r\phi)^{(j)}\neq 0$ (in the expansion $r\phi= r^{-j}(r\phi)^{(j)}$).

 In contrast, in the present paper, we do not use $r^2\partial_r$ as a commutator for the solution. Instead, late-time decay is determined by the conformal irregularity of the initial data of sufficiently many time integrals. See however \cref{prop:inhom:elliptic2} for the relation between $r^2\partial_r$ commutations and time integration $T^{-1}$ at the level of initial data.


\item In both \cite{luk_late_2024} and the present paper, a key role is played by $\mathcal{L}$, the stationary wave operator $\square_{g_M}$ \emph{without} time derivatives. In fact, the operator $\mathcal{L}$ also plays a crucial role as the resolvent operator at zero frequency in the spectral methods of \cite{tataru_local_2010, hintz_sharp_2022}.

In \cite{luk_late_2024}, the existence of appropriate elliptic-type estimates for $\mathcal{L}$, as proved in \cite{angelopoulos_late-time_2021}, is assumed in order to extend the tails from a large $r$-region to the rest of the spacetime. 

In the present paper, elliptic-type estimates for $\mathcal{L}$ are proved and  used only for \textit{large} angular frequency spherical harmonic modes to prove improved decay away from $\scrip$, and to construct the initial data for time integrals of these modes. However, a \emph{twisted modification}, $\check{\mathcal{L}}$, of $\mathcal{L}$ is employed to construct the time integral initial data for bounded spherical harmonic modes via explicit integration. $\check{\mathcal{L}}$ is also used to prove improved decay away from $\scrip$.


\item  
In \cite{luk_late_2024}, the strong Huygens principle is appealed to in order to obtain asymptotic estimates,  which is the main technical reason for their restriction to odd spacetimes.
In contrast, our present work is applicable in any dimension. Furthermore, the present work is applicable to extremal black hole settings.

\end{enumerate}


\subsection{Structure of the paper}
We outline here the remainder of the paper.
\begin{itemize}
\item We introduce the main geometric objects in Schwarzschild spacetimes and general notational conventions in \cref{sec:Sch}. We then introduce the class of dynamical metrics $g$ in \cref{sec:classmetrics}.
\item We give precise statements of the assumptions and main theorems in \cref{sec:main}: We state a precise version of \cref{thm:mainrough} in the form of \cref{thm:main}, and we state additional results applicable to nonlinear settings in \cref{thm:main:F,thm:main:phi3,thm:main:phi4,thm:main:puphipvphi}.
\item In \cref{sec:inhom}, we provide brief derivations of the main energy, elliptic and pointwise estimates for inhomogeneous wave equations on Schwarzschild. With the exception of some refinements and \cref{prop:inhom:elliptic2}, the results of this section are known in the literature.
\item In \cref{sec:approx}, we prove the relevant decay estimates on the global tail functions $\apphi$ and show that they may indeed be thought of as approximate solutions to the wave equation on Schwarzschild.
\item \cref{sec:time} constitutes the heart of the paper. We prove the main theorem \cref{thm:main} here.
\item In \cref{sec:additionalproof}, we prove \cref{thm:main:F,thm:main:phi3,thm:main:phi4,thm:main:puphipvphi}.
\item In \cref{sec:gen}, we discuss the non-vanishing of our the coefficients appearing in the late-time asymptotics for generic initial data. While we give a proof of genericity in the linear setting, we only give a sketch of a proof in the nonlinear setting. 
\item In \cref{app:motivation}, we motivate our assumptions on the class of dynamical metrics studied in this paper, and, finally, in \cref{app:derivapriori}, we explain how to recover the decay estimates assumed on $\phi$ in our main theorem.
\end{itemize}

\subsection{Acknowledgements}
We thank Istvan Kadar, Jonathan Luk and Sung-Jin Oh for their comments on a previous version of the manuscript, and we thank Gustav Holzegel for helpful discussions. Both authors acknowledge funding through the ERC Starting Grant 101115568.

 \section{The Schwarzschild spacetime and notation}\label{sec:Sch}
 In this section, we first {introduce} the Schwarzschild {spacetime} and {set the geometric notation that will be used throughout the rest of the paper}. 
 We also introduce specific notation used for inequalities in \cref{sec:Sch:notation}.

 \subsection{The Schwarzschild spacetime $(\mathcal{M},g_M)$}
 Fix some $\tau_0=u_0=v_0>0$. The following manifold with corners is a subset of the Schwarzschild black hole exterior manifold:
$
 	\M=[u_0,\infty)_u\times[v_0,\infty)_v\times S^2,
$ 
with $S^2$ the unit round sphere. The Schwarzschild metric in double null coordinates on $\M$ is given by:
 \begin{equation}
 	g_M=-4 D\dd u \dd v +\gs_M, \qquad D=1-\frac{2M}{r}=\Omega^2_M.
 \end{equation}
 Here, $\gs_M=r^2(\dd \vartheta^2+\sin^2\theta \dd \varphi^2)$ is the metric on the sphere of radius $r$, where $r$ is defined implicitly via
 \begin{equation}\label{eq:Sch:tortoise}
 	v-u=r^{\ast}=r+2M\log\((2M)^{-1}{Dr}\)+c_0, 
 \end{equation}
 for some constant $c_0$ to be fixed later.
 We record the following expressions for the inverse and determinant of $g_M$:
\begin{equation}
	g_M^{-1} =-\frac{1}{2D}(\pu \otimes \pv+\pv \otimes \pu)+\slashed{g}_M^{-1}, \qquad \sqrt{-\det{g_M}}=2D r^2\sin \theta=2D \sqrt{\det \slashed{g}_M}.
\end{equation}

We also consider the following null frame on $\mathcal{M}$:
$
 	e^M_3=\Omega_M^{-1}\pu,\,  e^M_4=\Omega_M^{-1}\pv, \, e^M_A=\partial_{\theta^A},
$
for $\{\theta^A\}=\{\vartheta,\varphi\}$ local coordinates on $S^2$. 
 With respect to this frame, the metric has the following ingoing and outgoing null expansions $\otrx_M$ and $\otrxb_M$:
 \begin{equation}
 \Omega_M g_M^{AB}g_M(\nabla_A e^M_4, e^M_B)=:	\otrx_M=\frac{2\Omega^2_M}{r}=-\otrxb_M:= \Omega_M g_M^{AB}g_M(\nabla_A e^M_3, e^M_B),
 \end{equation}
 where $\nabla$ denotes the Levi-Civita connection of $g_M$. 
 We will denote the connection induced on the spheres by $\slashed{\nabla}$, and we will moreover write $\sl=r\slashed{\nabla}$ to denote the connection on the unit sphere.
 Later on, we will introduce different metrics $g$ on $\M$; we will then denote the associated Levi-Civita connection by $\nabla^g$. 

 \subsection{Foliations, functions and vector fields on $\M$}
 \begin{figure}[htb]
 	\includegraphics[width=0.3\textwidth]{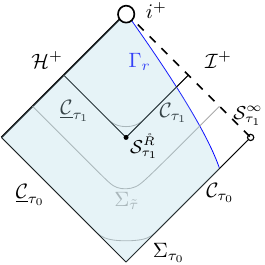}
 	\caption{Depiction of the foliations that we work with.}\label{fig:foliation}
 \end{figure}

Fix a large parameter $\rc> 100 M$, and  fix $c_0$ in \cref{eq:Sch:tortoise} such that $r^{\ast}(\rc)=0$, i.e.~such that $v=u$ along the hypersurface $r=\rc$.
We then define, for any $s\geq \tau_0$, and for any $2M\leq r_1<r_2\leq \infty$ (see~\cref{fig:foliation}):
\begin{nalign}
	\Cbar_{s} & = \M \cap \{v=s\} \cap \{r \leq \rc\}, 
	& \C_s & = \M \cap \{u=s\} \cap \{r \geq \rc\}, 
	& \Sigma_s & = \Cbar_s \cup \C_s, \\
	\Cbar_s^{r_1,r_2} & = \Cbar_s \cap \{r_1 < r < r_2\}, 
	& \C_s^{r_1,r_2} & = \C_s \cap \{r_1 < r < r_2\}, 
	& \Sigma_s^{r_1,r_2} & = \Sigma_s \cap \{r_1 < r < r_2\},
\end{nalign}
and, for any $R\in(2M,\infty)$, we write
\begin{equation}
	\Gamma_R=\M\cap\{r=R\}.
\end{equation}

By working in coordinates $(r,v,\theta^A)$ and $(u, 1/r, \theta^A)$, respectively, we may attach to $\M$ two boundaries\\ $\hplus= [v_0,\infty)_v\times \{2M\}_r\times S^2$ and $\scrip=[u_0,\infty)_u\times \{0\}_{\frac{1}{r}}\times S^2$, called the \emph{event horizon} and \emph{future null infinity}, respectively.
We then write, for $r_1\geq 2M$:
 \begin{equation}
	\S_s^{r_1}=\{r=r_1\}\cap \Sigma_s, \quad  \S_s^{\infty}=\Sigma_s \cap \scrip.
 \end{equation}
 
 Next, we introduce the (continuous, but not differentiable) function $\tau:\mathcal{M}\to \R$ whose level sets are $\Sigma_\tau$: \begin{equation}
 	\tau=\begin{cases}
 		u, \text{ if } r\geq \rc,\\
 		v, \text{ if } r\leq \rc.
 	\end{cases}
 \end{equation}
We also require a smoothed out version of $\tau$: We therefore fix $\tilde{\tau}$ to be a smooth function that agrees with $\tau$ on $\{r\leq \rc-1\}\cup \{r\geq \rc+1\}$ and such that $T\tilde{\tau}=1$ and  the level sets $\tilde{\Sigma}_{\tilde{\tau}}$ of $\tilde{\tau}$ are achronal. We set $\tilde{\tau}_0=\tau_0$.

 For $\infty\geq \tau_2>\tau_1>0$, we then write
 \begin{equation}
 	\D_{\tau_1,\tau_2}^{r_1,r_2}=\M\cap \{r_1<r<r_2\}\cap \{\tau_1<\tau<\tau_2\}, 
 \end{equation}

 
 Given a coordinate chart $\{x^{3},x^4,\varphi, \vartheta\}$, we will use the notation $\partial_{x^3}|_{x^4}$ and $\partial_{x^4}|_{x^3}$ for the corresponding coordinate basis vector fields to emphasize that $x^4$ and $x^3$ are held fixed, respectively. In this notation, we can write:
 \begin{equation}
 	T=\pu|_v+\pv|_u=2\partial_t|_r, \quad \Xb=-D^{-1}\pu|_v, \quad \Xx   =D^{-1}\pv|_u.
 \end{equation}
 This implies
 \begin{equation}
 	\pv|_u= T-D\Xb, \qquad \pu|_v =T-D\Xx  .
 \end{equation}
 We also note that in coordinates $(r,v)$, we have $T=\pv|_r$ and $\Xb=\partial_r|_v$, while, in coordinates $(u,r)$, we have $T=\pu|_r$, $\Xx=\partial_r|_u$. 
 From now on, whenever we write $\pv$ or $\pu$, it will be with respect to the double null coordinates. Furthermore, we will write $\partial_r|_{\Sigma}$ to denote the vector field $\Xx$ for $r\geq\rc$ and $\Xb$ for $r<\rc$.

 \subsection{The wave operator on Schwarzschild and its stationary solutions}\label{sec:Sch:stationary}
For any function $\phi$ and for any metric $g$ on $\M$, the wave operator $\Box_g$ is given by
\begin{equation}
	\Box_g\phi=g^{\mu\nu}\nabla^g_\mu \nabla^g_\nu\phi=\sqrt{\det{g}^{-1}}\partial_\mu((g^{-1})^{\mu\nu}\sqrt{\det g}\partial_\nu \phi).
\end{equation}
Thus, in the case of Schwarzschild, we can write:
\begin{equation}\label{eq:Sch:wave}
 	\frac{\sqrt{-\det g_M}}{2D\sin \theta}\Box_{g_M}\phi=r^2\Box_{g_M}\phi=\Xb(r^2D\Xb\phi)+r\Xb(rT\phi)+\Dl \phi= \Xx  (r^2D \Xx  \phi)-r\Xx  (rT\phi)+\Dl \phi.
 \end{equation}
 Setting $\psi=r\phi$, this may equivalently be written in the form
\begin{equation}\label{eq:inhom:wave:twisted}
	P_{g_M}\psi:= r \Box_{g_M} (r^{-1}\psi)=D^{-1}\pu\pv\psi-\frac{\Dl\psi}{r^2}+\frac{2M\psi}{r^3}.
\end{equation}
 
Given the orthonormal and real-valued spherical harmonics $Y_{\ell,m}$, we denote by $\phi_\ell$ the projection \begin{equation}\phi_{\ell}=\sum_{m=-\ell}^{\ell}\phi_{\ell,m} Y_{\ell,m}= \sum_{m=-\ell}^{\ell} {Y}_{\ell,m}\langle \phi,Y_{\ell,m}\rangle_{L^2(S^2)}.
	\end{equation}
 Projecting the wave equation onto fixed $\ell$ and looking for stationary solutions reduces \cref{eq:Sch:wave} to  Legendre's ODE; thus $w_\ell Y_{\ell,m}=b_\ell P_{\ell}(r/M-1) Y_{\ell,m}$ solves $\Box_{g_M}\phi=0$.  We fix the constants $b_{\ell}$ by requiring that $\lim_{r\to\infty}r^{-\ell}w_{\ell}=1$.
 In particular, $w_\ell$ solves
$
 	\Xb(r^2D\Xb  w_\ell)-\ell(\ell+1)w_\ell=0=	\Xx  (r^2D\Xx   w_\ell)-\ell(\ell+1)w_\ell,
$
implying:
 \begin{equation}\label{eq:Sch:EllEllbar}
 	\Xb(r^2Dw_{\ell}^2\Xb\frac{\phi}{w_\ell})=\Xb(r^2D\Xb\phi)\cdot w_\ell-\phi \Xb(r^2D\Xb  w_{\ell})=\Xb(r^2D\Xb\phi)\cdot w_\ell -\ell(\ell+1)w_{\ell}\phi,
 \end{equation}
 with the same expression being valid for $\Xb$ replaced by $\Xx$.
 Hence, writing $\check\phi_{\ell}=w_{\ell}^{-1}\phi_\ell$, we have that
 \begin{equation}
 	w_{\ell}	\frac{\sqrt{-\det g_M}}{2D\sin \theta}\Box_{g_M}\phi_{\ell}=r^2w_\ell \Box_{g_M}\phi_{\ell}=\Xb(r^2Dw_{\ell}^2\Xb\check{\phi}_{\ell})+w_{\ell}r\Xb(rT\phi_{\ell})=\Xx   (r^2Dw_{\ell}^2\Xx   \check{\phi}_{\ell})-w_{\ell}r\Xx  (rT\phi_{\ell}).
 \end{equation}

\newcommand{\Ell}{\mathcal{L}}
\newcommand{\Ellbar}{\underline{\mathcal{L}}}
\newcommand{\Ellc}{\check{\mathcal{L}}}
\newcommand{\Ellbarc}{\underline{\check{\mathcal{L}}}}

 \subsection{Notation: integration and volume forms}
We consider the following volume forms on the 2-spheres $S^2_{u_1,v_1}:=\{u=u_1\}\cap \{v=v_1\}$: 
 \begin{equation}
 	\dd \sigma= \sin \vartheta \dd \vartheta \dd\varphi, \quad \dd \tilde{\sigma}=\dd \vartheta \dd\varphi.
 \end{equation}
 
 When integrating over spacetime regions, we write $\dd \mu= \dd u \dd v \dd \sigma$, when integrating over sets like $\Cbar$ or $\C$, we write $\dd \mu= \dd u \dd \sigma$ or $\dd \mu= \dd v \dd \sigma$ instead. In all cases, we also write  $\dd \mu_D=D \dd \mu$. 
We also write $\dd\tilde{\mu}$ if $\dd \sigma$ is replaced by $\dd\tilde{\sigma}$. For instance, for a function $f$, we then have the relations
 \begin{equation}
 	\int_{\Cbar_{v}} f \dd \mu=\int_{\Cbar_v} f \dd u \dd \sigma=\int_{\Cbar_v}  \frac{f}{D} \dd r \dd \sigma=\int_{\Cbar_v}  \frac{f}{D} \dd \mu_D.
 \end{equation}
 
 Lastly, we introduce the following convention:
 An integral $\int_{\Sigma_\tau} f \dd \mu$ will always denote $\int_{\Cbar_\tau}f\dd \mu+\int_{\C_\tau}f\dd \mu$. In particular, when integrating over $\Sigma_\tau$, we will interpret an integral of a $\partial_r|_{\Sigma}$ derivative to mean:
 \begin{equation}
 	\int_{\Sigma_\tau} \partial_r|_{\Sigma}f \dd \mu=\int_{\Cbar_\tau}\Xb f \dd \mu+\int_{\C_\tau} \Xx f \dd \mu.
 \end{equation}
 

\subsection{Integral inequalities}
The proofs of the following two inequalities are standard:
\begin{lemma}[Poincar\'e's inequality]
	Let $f\in C^{\infty}(S^2)$. Then, for any $1\leq \ell_0\in\N$, we have
	\begin{equation}
		\int_{S^2}|\sl f_{\ell\geq\ell_0}|^2 \dd \sigma \geq \ell_0(\ell_0+1)\int_{S^2} f_{\ell\geq\ell_0}^2\dd \sigma. 
	\end{equation}
\end{lemma}
\begin{lemma}[Hardy's inequality]
	Let $f\in C^{\infty}([a,b])$ for $b>a$. Let moreover $q\neq-1$. Then
	\begin{equation}
		\int_a^b |s|^{q}|f(s)|^2\dd s\leq \frac{2}{q+1}(|s|^q f^2(s))|^{s=b}_{s=a}+\frac{4}{(q+1)^2} \int_{a}^b |s|^{q+2} |f'(s)|^2\dd s.		
	\end{equation}
\end{lemma}
Throughout the paper, we will also  make use of the Sobolev embedding $H^2(S^2)\hookrightarrow L^{\infty}(S^2)$ without explicit mention.

  \subsection{{Notation: commutator vector fields and inequality symbols}}\label{sec:Sch:notation}

 As usual, we write $f\lesssim F$ (or $f=O(F)$) if there exists a global (we call a constant global if it only depends on background parameters or if they feature in global (or initial data) assumptions such as \cref{ass:main}) constant $C>0$ such that $f\leq CF$. If this constant $C$ has dependencies on other parameters, e.g.~$C=C(n)$, we will write $f\lesssim_n F$.
 
Next, for $\mathcal{O}_i$ denoting the Killing vector fields on $S^2$, we denote
 \begin{align}
 	\Vk=\{T,\O_1,\O_2,\O_3\}, \quad \Vc=\{r\partial_r|_{\Sigma}\}\cup \Vk=\{r\partial_r|_{\Sigma}, T, \O_1,\O_2,\O_3\},\quad \V=\{\tau T, r \partial_r|_{\Sigma}, \mathcal{O}_1, \mathcal{O}_2, \mathcal{O}_3\}.
 \end{align}
 The vector fields $\Vk$ and $\Vc$ will be the vector fields with which we commute, the vector fields $\V$ will be the vector fields with which we measure pointwise decay.
 
 Let $\mathcal{V}$ denote a finite set of vector fields $Z_i$, like $\Vk$, $\Vc$ $\V$ (in an abuse of notation, we also allow for $\sl$ to feature in $\mathcal{V}$; in expressions like $\mathcal{V}^\alpha$, $\sl$ will then always be taken to act last). Let $\alpha\in \N^{|\mathcal{V}|}$ be a multi-index and $k\in \N$. Denote:
 \begin{equation*}
  \mathcal{V}^{\alpha}=\prod_{i=1}^{|\mathcal{V}|}Z_i^{\alpha_i}, \qquad  \mathcal{V}^{k}=\sum_{|\alpha|\leq k}  \mathcal{V}^{\alpha}.
 \end{equation*}

 Now, for a \textit{functional inequality} of the form $\mathcal{H}[\phi]\lesssim \mathcal{G}[\psi]$ we will write:
 \begin{equation}
 	\mathcal{H}[\phi]\lesV[\mathcal{V}^N] \mathcal{G}[\phi] \quad \text{if }\sum_{|\alpha|\leq N}\mathcal{H}[\mathcal{V}^{\alpha} \phi]\lesssim_N  \sum_{|\alpha|\leq N}\mathcal{G}[\mathcal{V}^{\alpha} \phi].
 \end{equation}
 We shall write $\mathcal{H}[\phi]\lesV[\mathcal{V}]\mathcal{G}[\phi]$ if 	$\mathcal{H}[\phi]\lesV[\mathcal{V}^N] \mathcal{G}[\phi]$ holds for all $N$.
We also generalise this notation for inequalities of the form $\sum_{i=1}^n \mathcal{H}_i[\phi_i]\lesssim \sum_{i=1}^{k}\mathcal{G}_i[\phi_i]$, or for inequalities of the form $\mathcal{F}[\phi]\lesV[\mathcal{V}] 1$.
\textit{In most cases, our inequalities will not verbatim look like $\mathcal{H}[\phi]\lesssim \mathcal{G}[\phi]$, but rather  take the form of integral inequalities, see for instance \cref{prop:inhom:rp}. These integrals will feature either $\phi$, $\psi$ or $F$, and we will accordingly interpret them as functionals of either $\phi$, $\psi$ or $F$.
}

On the other hand, for \textit{pointwise estimates}, we use the following notation capturing decay with respect to the vector fields $\V$: 
 Given $n\in\mathbb{N}$ and two functions $f$, $F$ on $\M$, with $F$ strictly positive, we then write
 \begin{equation}
 	f\lesVn F \text{ if $\V^n (f)\lesssim_n F$ holds both in $\{r< \rc\}$ and $\{r\geq \rc\}$,}
 \end{equation}
and we write $
 	f \lesV F$ (or $f=O_{\V}(F)$)
 if $f\lesVn F$ for all $n\in\mathbb{N}$. We extend this notation to general sets of vector fields as well.

 Finally, for $f$ a function with a $1/r$-expansion towards $\scrip$,  we denote by $f^{(i)}$ the $i$th coefficient in
 \begin{equation}
 	f-\sum_{i=0}^n \frac{f^{(i)}}{r^i}=O(r^{-n-}).
 \end{equation}
 For instance, we have
$
 	(fg)^{(1)}=f^{(0)}g^{(1)}+g^{(0)}f^{(1)}.
$
 \section{The class of polynomially decaying dynamical metrics on $\M$}
 \label{sec:classmetrics}
 \subsection{Metrics in a double null gauge}
In this section, we introduce a class of dynamical metrics $g$ on $\mathcal{M}$ with respect to which we will eventually consider the wave equation $\Box_g\phi=0$. In order to define this class, we first need to discuss metrics in a double null gauge.
A metric is said to be in double null form if it takes the form 
 \begin{equation}\label{eq:dyn:doublenullbdv}
 	g=-4\Omega^2 \dd u \dd v+\slashed{g}_{AB}(\dd \theta^A-b^{A}\dd v)(\dd \theta^B-b^{B}\dd v).
 \end{equation}
 Alternatively, the $b$-term may appear in front of the $\dd u$:
 \begin{equation}\label{eq:dyn:doublenulldu}
 	g=-4\Omega^2 \dd u \dd v+\slashed{g}_{AB}(\dd \theta^A-b^{A}\dd u)(\dd \theta^B-b^{B}\dd u).
 \end{equation}
 
Let us for now assume that $g$ is of the form \cref{eq:dyn:doublenulldu} (all formulae below have obvious analogues for \cref{eq:dyn:doublenullbdv}). Then its inverse and determinant are given by:
 	\begin{equation}
 	g^{-1}=-\frac{1}{2\Omega^2}(\pu\otimes\pv+\pv\otimes \pu)-\frac{b^A}{2\Omega^2}(\pu\otimes \partial_A+\partial_A\otimes \pu)+\slashed{g}^{-1}, \quad \det{g}=2\Omega^2\detgs.
 \end{equation}
 
 We may associate to \cref{eq:dyn:doublenulldu} a double null frame
$e_3=\Omega^{-1}(\pu+b^A\partial_{\theta^A}),\,e_4=\Omega^{-1}\pv,\, e_A=\partial_{\theta^A}$, and we recall the definitions
 \begin{equation}\label{eq:dyn:trx}
 	\otrx=\Omega g^{AB} g(\nabla^g_A e_4, e_B), \qquad \otrxb=\Omega g^{AB}g(\nabla^g_A e_3, e_B),
 \end{equation}
 where $\nabla^g$ denotes the Levi-Civita connection of $g$.
 From \cref{eq:dyn:trx}, one may then readily derive the following equations:\footnote{We use $\pu \log \sqrt{\det \gs}=\frac12 \gs^{AB} \pu g_{AB}$ and  write $\pu g_{AB}$ as the Lie-derivative in $e_3$ minus the Lie-derivative in the $b^A\partial_A$-direction.}
 \begin{align}\label{eq:dyn:DlogG}
 	\pu \log \sqrt{\det\slashed{g}}=\otrxb-\div^g b=\otrxb-g^{AB}\nabla^g_A b_B, && \pv  \log \sqrt{\det\slashed{g}}=\otrx.
 \end{align}
 \subsection{The class of polynomially decaying metrics on $\M$}
 We now define on $\M$ a class of dynamical metrics $g$ settling down to $g_M$ with rate $(\tau+r)^{-1}\tau^{-\delta}$, where $\tau$ and $r$ denote the functions on $\M$ defined in \cref{sec:Sch}.
These metrics should be thought of as arising from a stability proof of the Schwarzschild solution as in \cite{dafermos_non-linear_2021}.

 In particular, since it is not yet known how to prove stability with respect to a single choice of double null gauge (ignoring, of course, the failure of the spheres to be parallelisable), we allow for metrics covered by the Schwarzschild coordinates $u,v,\theta^A$ such that they are of double null form \cref{eq:dyn:doublenulldu} near $\scrip$, of double null form \cref{eq:dyn:doublenullbdv} near $\hplus$, but not necessarily of double null form in some intermediate matching region.
 
 \begin{ass}[Polynomial decay]\label{ass:dyn:pol}
 	Let $\delta>0$ and let $1\leq N\in\mathbb{N}$. We say that a metric $g$ settles down polynomially to $g_M$ (with rate $\delta$ and conformal regularity $N$) if $g$  takes the form \cref{eq:dyn:doublenullbdv} for $r\leq \rc-1$, \cref{eq:dyn:doublenulldu} for $r\geq \rc$, and if the following bounds hold: 
 	
 	For $r\leq \rc$, and for $\mu,\nu\in\{A, 3,4\}$ ($A=1,2$ always denoting spherical indices), we require
 	\begin{equation}\label{eq:dyn:improved1}
 		g_{\mu\nu} -(g_{M})_{\mu\nu}\lesV \tau^{-1-\delta}.
 	\end{equation}
 	
 	For $r\geq \rc$, we assume that there exists a collection of functions $(\Omega^2)^{(i)}, \otrx^{(i)}$, $S^2$-1-forms $b^{(i)}$ and symmetric $S^2$-$(0,2)$-tensor fields\footnote{$b^{(i)}$ and $\slashed{g}^{(i)}$ being tensors along $\scrip$, their indices are raised and lowered with respect to the metric on the unit sphere.} $\slashed{g}^{(i)}$ along $\scrip$ for any ${-1}\leq n\leq N$:\footnote{Whenever the lower index in a sum exceeds the upper one, we take the sum to vanish.}
 	\begin{align}
 		&	\Omega^2(u,v,\theta)-\Omega^2_M(u,v,\theta)-\sum_{i= 2}^n\frac{(\Omega^2)^{(i)}(u, \theta)}{r^i}\lesV_n r^{-1}u^{-\delta}  \left(\frac{u}{r}\right)^{n} &\forall -1\leq n\leq N+1,\label{eq:dyn:pol:Omega} \\
 		&	rb^A(u,v,\theta)-\sum_{i= 1}^n\frac{r(b^{(i)})^{A}(u,\theta)}{r^i}\lesV_n r^{-1}u^{-\delta} \left(\frac{u}{r}\right)^{n} &\forall -1\leq n\leq N,\\
 		&	r^{-2}\slashed{g}_{AB}(u,v,\theta)-r^{-2}\gs_{M,AB}(u,v,\theta)-\sum_{i=1}^n\frac{r^{-2}(\slashed{g}^{(i)})_{AB}(u,\theta)}{r^i}\lesV_n  r^{-1}u^{-\delta}  \left(\frac{u}{r}\right)^{n}&\forall -1\leq n\leq N, \label{eq:dyn:pol:g}\\
 		&	\otrx(u,v,\theta)-\otrx_M(u,v,\theta)-\sum_{i= 3}^{n+1} \frac{\otrx^{(i)}(u,\theta)}{r^i}\lesV_n r^{-2}u^{-\delta}  \left(\frac{u}{r}\right)^{n}&\forall -1\leq n\leq N+1,\label{eq:dyn:pol:trx}
 	\end{align}
 \end{ass}
 \begin{rem}
 	The reason for the $r$-weights in front of $b^A$ and $\gs_{AB}$ is to emphasise the $r$-weights coming from having the indices downstairs/upstairs! 
 	The reason for allowing for $n=-1$, where all the sums above are empty, is to capture improved $\tau$-decay away from $\scrip$.
 	Furthermore, notice, for instance, that  \cref{eq:dyn:pol:Omega}, applied with $n=1$ and $n=2$, implies:
 	\begin{equation}
 		\frac{(\Omega^2)^{(2)}}{r^2}\leq| \Omega^2-\Omega^2_M|+ \left| \Omega^2-\Omega^2_M-\frac{(\Omega^2)^{(2)}}{r^2}\right| \lesV r^{-1}u^{-\delta}\(\frac{u}{r}+\(\frac{u}{r}\)^2\) \implies (\Omega^2)^{(2)}\lesV u^{-\delta+1};
 	\end{equation}
 	more generally, \cref{eq:dyn:pol:Omega} implies that $(\Omega^2)^{(i+1)}\lesV u^{-\delta+i}$, similarly for the other quantities.
 \end{rem}

 \begin{rem}
 	\cref{ass:dyn:pol} implies that $\Omega^2$, $b^A$, $\gs_{AB}$ and $\otrx$ extend continuously to $\hplus$, in view of the regularity with respect to $\Xb$ encoded in the $\lesV$-notation.
  \end{rem}
\begin{rem}
The reader should think of the metrics of \cref{ass:dyn:pol} as being \emph{Bondi normalised}; i.e.~as satisfying
\begin{equation}
	0=\lim_{\scrip} b^A=\lim_{\scrip} r^{-2}(\gs_{AB}-(\gs_M)_{AB})=\lim_{\scrip} r(\Omega^2-\Omega^2_M)=\lim_{\scrip} r^2(\otrx-\otrx_M),
\end{equation}
 and being characterised by a News function $\xhb^{(1)}$ along $\scrip$ decaying like $\xhb^{(1)}\lesssim u^{-1-\delta}$. 
  Notice that the gauge assumption $\lim r^2(\otrx-\otrx_M)=0$ implies
 \begin{equation}
 	\pv\left( \log \sqrt{\det\slashed{g}}- \log \sqrt{\det\slashed{g}_M}\right)=\otrx-\otrx_M \implies \frac{\sqrt{\det\slashed{g}}}{\sqrt{\det\slashed{g}_M}}=O(\e^{r^{-2}})=1+\O(r^{-2}),
 \end{equation}
 i.e.~the area radius of $S^2_{u,v}$ with respect to $g$ approaches the area radius with respect to $g_M$ at a sufficiently fast rate as $v\to \infty$. This assumption is equivalent to saying that $\gs^{(1)}$ is traceless.\footnote{Without this assumption, already quantities such as the $\ell=0$-Newman--Penrose constant, given by $\lim_{v\to\infty}r^2\pv(r\phi)$ in Schwarzschild, would need to be modified for the dynamical metrics.}
 See also \cref{app:motivation} and \cite{dafermos_non-linear_2021}[Section 17].
 \end{rem} 
 
\begin{rem}
	The metrics constructed in \cite{dafermos_non-linear_2021} satisfy \cref{ass:dyn:pol} with some finite $N$ and with $\delta=1$, apart from: 1) the regularity with respect to $\tau T$, and 2) the improved $\tau$-decay away from $\scrip$. Note that 2) can be shown to follow from 1), combined with suitable elliptic estimates. 
	 Establishing 1) would require additional work within the framework of \cite{dafermos_non-linear_2021}. {On the other hand, an increase in $N$ (establishing additional conformal regularity) would follow directly from stricter conformal regularity assumptions at the level of the initial data}.
 \end{rem}
 \begin{rem}
 	One may alternatively work with metrics that are of the form \cref{eq:dyn:doublenullbdv} near $\scrip$. Such classes of metrics are constructed in \cite{dafermos_scattering_2024}, for instance. However, the position of the $b$-term turns out to not affect the results of the present paper in any way.
 \end{rem}
  \begin{rem}
 	The different upper bounds for $n$ in \cref{eq:dyn:pol:Omega}--\cref{eq:dyn:pol:trx}, measuring the conformal regularity of the different quantities, turn out to be natural in the context of the present paper. 
 	They are also somewhat natural in the context of the Einstein vacuum equations, though the situation there is slightly more complicated; see also \cref{app:motivation}.
 \end{rem}

\subsection{The dynamical wave operator vs the Schwarzschild wave operator}
We now study the consequences of \cref{ass:dyn:pol} when comparing the dynamical and Schwarzschild wave operator. For any function $\phi$, we have
\begin{equation}\label{eq:dyn:G}
	\sqrt{-\det g}\Box_g \phi=\sqrt{-\det g_M}\Box_{g_{M}}\phi+\partial_\alpha (G^{\alpha\beta}\partial_\beta\phi), \quad\text{where}\quad 	G^{\alpha\beta}=g^{\alpha\beta}\sqrt{-\det g}-g_M^{\alpha\beta}\sqrt{-\det g_M}.
\end{equation}

We will therefore typically write $\Box_g\phi=0$ as
\begin{equation}\label{eq:wave:dynamical:inhom}
	\Box_{g_M}\phi = \frac{\partial_{\alpha}(G^{\alpha\beta}\partial_\beta\phi)}{-\detgm}=F, \quad \text{or, equivalently,}
\end{equation}

\begin{equation}\label{eq:dyn:wave:X}
	\Xb(r^2D\Xb\phi)+r\Xb(rT\phi)+\Dl\phi=-\frac{\partial_\alpha\left(G^{\alpha\beta}\partial_\beta\phi\right)}{2D\sin \theta}=\frac{\sqrt{-\det g_M}F}{2D\sin \vartheta}=	\Xx   (r^2D\Xx   \phi)-r\Xx   (rT\phi)+\Dl\phi.
\end{equation}

We observe, in particular, that, when $g$ takes the double null form \cref{eq:dyn:doublenullbdv} or \cref{eq:dyn:doublenulldu}:
\begin{nalign}\label{eq:dyn:Gspelledout}
	G^{uu}&=G^{vv}=0,\qquad G^{uv}=-(\sqrt{\det \slashed{g}}-\sqrt{\det \slashed{g}_M}), 
	\qquad G^{uA}=g^{uA}\sqrt{-\det g}, \qquad G^{vA}=g^{vA}\sqrt{-\det g}\\
	 G^{AB}&= (\slashed{g}^{AB}-\slashed{g}_M^{AB})\sqrt{-\det g}+\slashed{g}_M^{AB}(\sqrt{-\det g}-\sqrt{-\det g_M}),
\end{nalign}
where we have used that $\sqrt{-\det g}=2\Omega^2 \sqrt{\det{\slashed{g}}}$. 

Recall also that $g^{uA}=-\frac{b^A}{2\Omega^2}$, $g^{vA}=0$ when $g$ is of the form \cref{eq:dyn:doublenullbdv}, and vice versa if $g$ is of the form \cref{eq:dyn:doublenulldu}. In the latter case, we then have that $\lim_{r\to\infty} G^{vA}=- (b^{(1)})^A \sin \vartheta$ and that $\lim_{r\to\infty} r G^{AB}= -2\sin \vartheta \gsh^{(1)}$. More generally, we have
\begin{lemma}\label{lem:dyn:G}
 \cref{ass:dyn:pol}  implies that there exists functions $(G^{\mu\nu})^{(i)}$ along $\scrip$ such that, componentwise: 
	\begin{align}\label{eq:dyn:Guv}
		G^{uv}-\sum_{i=0}^{n-2} \frac{(G^{uv})^{(i)}}{r^i}&\lesV r \tau^{-\delta}\left(\frac{\tau}{r}\right)^{n} &\forall -1\leq n\leq N+1,
		\\\label{eq:dyn:GvA}
		G^{vA}-\sum_{i=0}^{n-1} \frac{(G^{vA})^{(i)}}{r^i}&\lesV  \tau^{-\delta} \left(\frac{\tau}{r}\right)^{n} &\forall -1\leq n\leq N,\\
		G^{AB}-\sum_{i=1}^{n} \frac{(G^{AB})^{(i)}}{r^i}&\lesV  r^{-1}\tau^{-\delta} \left(\frac{\tau}{r}\right)^{n}&\forall -1\leq n\leq N.\label{eq:dyn:GAB}
	\end{align}
\end{lemma}

In order to re-write a few limits further, we recall the definition of the Bondi mass along $\scrip$ as the limit of the quasilocal Hawking mass:
	\begin{equation}
	M_B(u)=\lim_{v\to\infty}\frac{r}{2}\left(1-\frac{1}{16\pi}\int_{S^2_r}\trx\trxb \right).
\end{equation}

\begin{lemma}\label{lem:dyn:Bondi}
Denote $\otrxb^{(2)}=4\mathfrak{m}(u,\theta)$, i.e.~$\otrxb=-2/r+\frac{4\mathfrak{m}(u)}{r^2}+O(r^{-3})$.  Then under \cref{ass:dyn:pol}, we have $\overline{\mathfrak{m}}(u)=M_B(u)$, with $\overline{\mathfrak{m}}$ denoting the spherical average of $\mathfrak{m}$.
\end{lemma}
\begin{proof}
\begin{nalign}
	\trx\trxb&=\frac{1}{\Omega^2}\otrx\otrxb=\frac{1}{\Omega^2_M}\otrx\otrxb+O(r^{-3-})=\frac{1}{\Omega^2_M}(\otrx)_M\otrxb+O(r^{-3-})\\
&	=	\frac{1}{\Omega^2_M}\(\frac2r-\frac{4M}{r^2}\)\(-\frac{2}{r}+\otrxb+\frac{2}{r}\)+O(r^{-3-})	
	=-\frac4{r^2}+\frac2r\(\otrxb+\frac2r\)+O(r^{-3-}).
\end{nalign}
Thus, using the definition of the Bondi mass, we have
$
	M_B(u)=\lim_{v\to\infty}\frac r4 (r\overline{\otrxb}+2).
$
\end{proof}
\begin{rem}
	In \cref{app:motivation}, we show that $\partial_u \mathfrak{m}(u)=-\frac14|\xhb^{(1)}|^2 $, with $\xhb^{(1)}$ denoting the News function along $\scrip$.
\end{rem}
For later reference, we summarise the main properties of $G^{\mu\nu}$ in the following 
\begin{lemma}\label{lem:dyn}
Let $g$ be as in \cref{ass:dyn:pol} and let $\mathring{\div}$ denote the divergence operator on the round unit sphere. Then the components of $G^{\alpha\beta}$ as defined in \cref{eq:dyn:G} satisfy:
\begin{align}\label{eq:dyn:TGuvlimit}
	\lim_{v\to\infty} G^{vA}=-(b^{(i)})^A\sin\vartheta, \quad \lim_{v\to\infty}rG^{AB}=-2\sin\vartheta (\gsh^{(1)})^{AB}, \quad \lim_{v\to\infty}-TG^{uv}=4m(u,\theta^A)-4M-\mathring{\div} b^{(1)}
\end{align}
	\end{lemma}
	\begin{proof}
		We have already proved the first two limits. For the third, note that $	TG^{uv}=-(\pu+\pv)\sqrt{\det\slashed{g}}=-\sqrt{\det{\slashed{g}}}(\otrxb+\otrx-\div^g b)$.
	\end{proof}

\section{Precise statement of the main theorems}\label{sec:main}
In \cref{sec:mainmain}, we state the precise version of \cref{thm:mainrough}, along with a precise formulation of the corresponding assumptions and with further commentary on the results. In \cref{sec:main:additional}, we state additional results for nonlinear wave equations.
\subsection{The main result (\cref{thm:main}): Precise asymptotics for dynamical spacetimes}\label{sec:mainmain}
We consider solutions to the equation:
\begin{equation}\label{eq:main:wave}
	\Box_g \phi=\sqrt{\frac{\det g_M}{\det g}} \mathfrak{F} \qquad \(\iff \Box_{g_M} \phi= -\frac{\partial_\mu (G^{\mu\nu}\partial_\nu \phi)}{\detgm} +\mathfrak{F}\),
\end{equation}
for some inhomogeneity {$\mathfrak{F}:\mathcal{M}\to \R$}.
We {impose} the following assumptions on $g$, $\phi$ and $\mathfrak{F}$:

\begin{ass}\label{ass:main}
\begin{enumerate}[label=\emph{(\Roman*)}]
\item Let $g$ be as in \cref{ass:dyn:pol} with $\delta>0$ and $N=\infty$.

\item Let $2<\beta\in\R$. Assume that the inhomogeneity $\mathfrak{F}$ satisfies for all $\ell_0<\beta$:
\begin{equation}\label{eq:main:assonF}
	\mathfrak{F}_{\ell\geq \ell_0}-\sum_{i=4}^{n+3} \frac{(\mathfrak{F}_{\ell\geq\ell_0})^{(i)}}{r^i}\lesV \frac{1}{r^4\tau^{\beta}}\left(\frac{\tau}{r}\right)^n \qquad \forall -2-\ell_0\leq n\leq \max(\ell_0,1).
\end{equation}

\item Let $\phi$ solve \cref{eq:main:wave} and arise from smooth compactly supported data along $\tilde{\Sigma}_{\tilde{\tau_0}}$. Assume the following \textit{a priori} decay assumption on $\psi=r\phi$:
\begin{equation}\label{eq:main:ass:phi}
	\psi \lesV \min_{q\in[0,1]} \left(\frac{r}{\tau}\right)^q.
\end{equation}
\end{enumerate}
\end{ass}
\begin{rem}
	One can show that (I) and (II) imply (III), see \cref{app:derivapriori}. We chose to nevertheless state (III) as a global assumption in order to emphasise that a weak a priori decay estimate for $\phi$ suffices to apply our methods. Furthermore, it would, in fact, be sufficient to merely assume $\psi \lesV[\Vc] 1$; as one could readily derive \cref{eq:main:ass:phi} (i.e.~regularity w.r.t.~$\tau T$ and improved decay away from $\scrip$) from this using the inhomogeneous energy estimates of \cref{sec:inhom}.
	We also note that, while we pose initial data along $\tilde{\Sigma}_{\tilde{\tau}_0}$, we will later on perform our (Schwarzschildean) energy estimates starting at $\Sigma_{\tau_0}$.
\end{rem}

We will now introduce a second assumption on the metric, relevant for capturing the precise late-time asymptotics for higher $\ell$-modes away from $\scrip$. Roughly speaking, with \cref{ass:dyn:pol} already telling us that $g-g_M\lesssim \tau^{-1-\delta}$ away from $\scrip$, we want to demand that $(g-g_M)_{\ell\geq \ell_0}\lesssim r^{\ell_0}\tau^{-1-2\ell_0-\delta}$. So as to avoid introducing tensorial mode decompositions, we capture this in the following
\begin{ass}\label{ass:main:aux}
Fix $L\in N$. Let $g$ and $\phi$ be as in \cref{ass:main}. We assume the following \underline{\emph{implication}}: For any $\tilde{\alpha}\in \R$:
\begin{equation}\label{eq:main:ass:aux}
	\text{If} \quad r^{-\ell_1}\phi_{\ell\geq \ell_1}\lesV \tau^{-\tilde{\alpha}-2\ell_1} \text{ for all $1\leq \ell_1\leq L$ }, \quad \text{then} \quad r^{-\ell_1} \partial_\mu(G^{\mu\nu} \partial_\nu \phi)_{\ell \geq \ell_1}\lesV \tau^{-\tilde{\alpha}-2\ell_1-1-\delta} \text{ for all $1\leq \ell_1\leq L$}.
\end{equation}
\end{ass}
\begin{rem}Notice that \cref{ass:dyn:pol} already implies \cref{eq:main:ass:aux} with $\ell_1=0$ in view of  \cref{eq:dyn:Guv}--\cref{eq:dyn:GAB}.\end{rem}
\begin{rem}
The point of \cref{ass:main:aux} is the following: 
Say we wanted to show that the $\ell=3$-mode decayed four powers faster than $\ell=1$ in a finite-$r$-region. For this to be consistent, the term $\partial_\mu(G^{\mu\nu} \partial_\nu \phi)_{\ell=3}$ also needs to decay faster, but this term now includes coupled modes that schematically look like $G_{\ell=3}\phi_{\ell=1}+G_{\ell=1}\phi_{\ell=3}$. 
\cref{ass:main:aux} (with, say, $L=3$) ensures that $G_{\ell=3}$ decays four powers faster than $G_{\ell=1}$. 
\end{rem}

Finally, we introduce the following class of approximate solutions: We denote for $\ell\in\N$ and for $m=-\ell,\dots,\ell$:
\begin{equation}
	\apphi=\frac{\chi r^{\ell}+(1-\chi) w_{\ell}}{\tilde{\tau}^{\ell+1}(\tilde{\tau}+r)^{\ell+1}}, \qquad \appsi=r\cdot \apphi.
\end{equation}
Here, $w_\ell$ denotes the stationary solution introduced in \cref{sec:Sch:stationary}, and $\chi(u/r)$ denotes a smooth cut-off identically one for $u/r\geq 3/4$, and vanishing for $u/r\leq 1/4.$\footnote{Introducing this cut-off is natural for a better choice of approximate solution, however, \cref{thm:main} remains valid if $\chi\equiv0$.}

We are now ready to state the main theorem of the present paper. The proof can be found in \cref{sec:time}:
\begin{thm}\label{thm:main}
	Let $g$ and $\phi$ be as in \cref{ass:main}, and assume that \cref{ass:main:aux} holds for some $L\in\N$. 
	Then the following estimates hold in $\M$:
	For $\ell=0$, we have that:
	\begin{equation}\label{eq:main:mainthm:l=0}
		\psi_{\ell=0}-T \appsi[0,0]\cdot \integral_0[\phi]  \lesV \tau^{-5/2}\min_{q\in[0,1]} \left(\frac{r}{\tau}\right)^q,
	\end{equation}
	with
	\begin{equation}
	\label{eq:main:thm:I0}
	\integral_0[\phi]:=\frac{1}{4\pi}\int_{\scrip} \mathfrak{m} \psi(u,\theta^A)+\frac{\mathfrak{F}^{(4)}}{2} \dd u \dd \sigma.
	\end{equation}
	Furthermore, for $1\leq \ell_0< \min(2+\delta,\beta)$, we have that:
	\begin{equation}\label{eq:main:mainthm:higherell}
		\psi_{\ell\geq\ell_0}-\sum_{m=-\ell_0}^{\ell_0}\appsi[\ell=\ell_0,m] Y_{\ell_0,m}\cdot \integral_{\ell_0,m}[\phi] \lesV \left( \tau^{-\ell-3/2}+\tau^{-3-\delta+ }+\tau^{-1-\beta+}\right) \min_{q\in[0,\min(L+1, \ell_0+1)]} \left(\frac{r}{\tau}\right)^q,
	\end{equation}
	 where the integrals above are given by
	\begin{equation}\label{eq:main:thm:integralexpression}
		\integral_{\ell,m}[\phi]:=\frac{(-1)^{\ell}}{(4\ell+2)\binom{2\ell}{\ell}}\int_{\scrip}\underbrace{ \left[r^{-\ell}\left(w_{\ell}' (G^{v\beta}-G^{u\beta})\partial_\beta\phi Y_{\ell,m}+D^{-1}w_{\ell} G^{A\beta}\partial_\beta \phi\partial_A Y_{\ell,m}+2r^2 w_\ell\mathfrak{F} Y_{\ell,m}\sin\vartheta\right)\right]^{(\ell+1)}}_{:=-2H^{(\ell+1)}_{\ell,m}} \dd u \dd \tilde{\sigma}.
	\end{equation}
\end{thm}
\begin{rem}[Improved error estimates]
	The $\tau^{-5/2}$-decay in \cref{eq:main:mainthm:l=0} can be improved to $\tau^{-3+}$.
		Similarly, the $\tau^{-3/2-\ell_0}$ in \cref{eq:main:mainthm:higherell} to $\tau^{-2-\ell_0+}$; see \cref{obs:error}. 
\end{rem}
\begin{rem}[Higher-order asymptotics]\label{rem:main:higher}
		Our method can be extended to also prove the \textbf{next-to-leading-order global asymptotics} (e.g.~for $\ell=0$ at order $\tau^{-3}\log \tau$, see \cref{rem:time:secondorder}). This will be the content of future work.
\end{rem}
\begin{rem}[Weakening the assumptions on $g$]
	We could optimise the assumptions for \cref{thm:main}. For instance, the improved $1/\tau$-decay for $g$ away from $\scrip$ assumed on $g$ in \cref{ass:dyn:pol} is not actually necessary to infer global asymptotics. For instance, if we merely had $g\lesssim r^{-1}\tau^{-\delta}$,\footnote{For $g$ settling down to Schwarzschild, one still has integrated local energy decay in this case (see \cite{metcalfe_prices_2011}; this is not known for $g$ settling down to Kerr, however.)} we would still be able to prove \cref{eq:main:mainthm:l=0} with RHS replaced by $\tau^{-5/2}\min_{q\in[0,\min(1,1/2+\delta-)]}(\frac{r}{\tau})^q$; in particular, away from $\scrip$, the error term would decay like $\tau^{-3-\delta+}$. See \cref{obs:weaken:away}. 
\end{rem}
\begin{rem}[Alternative expressions for the integrals]
	The integral in \cref{eq:main:mainthm:l=0} can be recast as (see~\cref{prop:time:rewritingl=0})
	\begin{nalign}
4\pi\integral_0[\phi]=\int_{\scrip}  \mathfrak{m} \psi(u,\theta^A)+\frac12\mathfrak{F}^{(4)} \dd u \dd \sigma =\int_{\scrip} ( \mathfrak{m}-M)\psi +\frac12\mathfrak{F}^{(4)}\dd u \dd \sigma +\frac12\int_{\hplus} \otrx \phi \detgs \dd\tilde{\sigma}\dd v\\
+\frac{M}2\int_{\M} r^2\mathfrak{F} \dd \mu_D+\frac{M}2 \int_{\S_{\tau_0}^{2M}}  \phi \detgs \dd\tilde{\sigma} +\frac{M}2\int_{\Cbar_{\tau_0}} n_{\Cbar}^\mu \partial_\mu \phi \detgs \dd\tilde{\sigma} \dd u+\frac{M}2\int_{\C_{\tau_0}} n_{\C}^\mu \partial_\mu \phi \detgs \dd\tilde{\sigma} \dd v,
	\end{nalign}
Applying the divergence theorem in a spacetime region bounded by $\Sigma_{\tau_0}$ and $\tilde{\Sigma}_{\tilde{\tau}_0}$ produces the identity \cref{eq:intro:I0}.

We can also give simpler expressions for the higher-order integrals $\integral_{\ell,m}[\phi]$ in terms of merely (integrals over) $\psi^{(0)}$ and its angular derivatives. For instance, for $\ell=1$, we have (see~\cref{cor:time:l=1alternative})
\begin{align}\nonumber
	\integral_{\ell=1,m}[\phi]&=-\frac16 	\int_{\scrip} (2 \mathfrak{m}-2M) Y_{1,m} \psi -(\gsh^{(1)})^{AB} \sl_A Y_{1,m} \sl_B \psi+(b^{(1)})^A \psi \sl Y_{1,m} \sl_B Y_{1,m} +(r w_\ell\mathfrak{F})^{(4)} Y_{\ell,m}\Big) \dd u \dd {\sigma}
	\\&=-\frac16
	\int_{\scrip}\psi\Big( (2 \mathfrak{m}-2M) Y_{\ell=1,m}  +(\gsh^{(1)})^{AB} \sl_A \sl_B Y_{\ell=1,m} \Big)+(r w_\ell\mathfrak{F})^{(4)} Y_{\ell,m} \dd u \dd {\sigma},\label{eq:main:thm:integral:l=1}
\end{align}
where the second relation holds if $\div \gsh^{(1)}=b^{(1)}$, which is the case if $g$ solves the Einstein vacuum equations.
\end{rem}
\begin{rem}[Generic nonvanishing of the coefficients]
	We show in \cref{sec:gen} (see \cref{prop:gen}) that the coefficients $\integral_{\ell,m}[\phi]$ each only vanish for a subset in the space of compactly supported initial data of codimension at least 1.
\end{rem}

\begin{rem}[Mode coupling]
	The restriction $\ell_0<2+\delta$ has the following origin: The integrand in \cref{eq:main:thm:integralexpression} in particular will always feature the $\ell=0$ modes of $\psi$, which decays like $\psi_{\ell=0}\lesssim u^{-2}$. Thus, the integrals in \cref{eq:main:thm:integralexpression} turn out to be well-defined only for $\ell_0<2+\delta$, since the metric decay along $\scrip$ is always just governed by $\delta$. 
	For nonlinear problems where $g$ depends on $\phi$, we would, however, also find that higher $\ell$-modes of $g$ decay with faster rates; mode coupling would then be further suppressed. 
	For a simple nonlinear example, see \cref{thm:main:phi3} and \cref{rem:main:mode2} below.
\end{rem}

\subsection{{Additional} results (\cref{thm:main:F,thm:main:phi3,thm:main:phi4,thm:main:puphipvphi}): Power nonlinearities and null forms}\label{sec:main:additional}
In \cref{eq:main:assonF}, the inhomogeneity is assumed to decay like $r^{-4}$ towards $\scrip$. Of course, it may just as well be assumed to decay like $r^{-3}$. That is, we may assume that, for some $0<\beta\in\R$, and for all $\ell_0<\beta$:
\begin{equation}\label{eq:main:assonF2}
	\mathfrak{F}_{\ell\geq \ell_0}-\sum_{i=3}^{n+3} \frac{(\mathfrak{F}_{\ell\geq\ell_0})^{(i)}}{r^i}\lesV \frac{1}{r^4\tau^{\beta}}\left(\frac{\tau}{r}\right)^n \qquad \forall -2-\ell_0\leq n\leq \ell_0.
\end{equation}
In this case, we obtain
\begin{thm}\label{thm:main:F}
	Under the assumptions of \cref{thm:main}, if \cref{eq:main:assonF} is replaced with \cref{eq:main:assonF2}, then \cref{eq:main:mainthm:higherell} holds for all $0\leq \ell_0< \min(2+\delta,\beta)$. In particular, if $g=g_M$, then we have
	\begin{equation}\label{eq:main:Fthm}
			\psi_{\ell\geq\ell_0}-\sum_{m=-\ell_0}^{\ell_0}\appsi[\ell=\ell_0,m] Y_{\ell_0,m}\cdot \frac{(-1)^{\ell}}{(2\ell+1)\binom{2\ell}{\ell}} \int_{\scrip} \left[r^{-\ell+2}w_\ell \mathfrak{F}_{\ell,m}\right]^{(\ell+1)}\dd \mu \lesV \left( \tau^{-\ell-3/2+}+\tau^{-1-\beta+}\right) \min_{q\in[0,\ell_0+1]} \left(\frac{r}{\tau}\right)^q.
	\end{equation}
\end{thm}
We could, in principle, also introduce a class of admissible smooth and bounded nonlinearities (as in \cite{kadar_scattering_2025}) and allow $\mathfrak{F}$ to depend of $\phi$ and prove that \cref{thm:main:F} still holds with $\mathfrak{F}$ replaced by $\mathfrak{F}(x,\phi, \partial\phi,\partial^2\phi)$ so long as we have the following two ingredients:
\begin{enumerate}
	\item A weak quantitative \textbf{stability} (and global existence) result for $\Box_{g}\phi=\mathfrak{F}(x,\phi,\partial\phi,\partial^2\phi)$ along with a suitable algebraic condition on $\mathfrak{F}$.
	\item \textbf{\textbf{Mode coupling:}} Assumption a.~is enough to infer the asymptotics for $\ell=0$. Asymptotics for higher $\ell$ are determined by checking if the integrals $\integral_{\ell,m}[\phi]$ (involving in particular $\mathfrak{F}(x,\phi,\partial\phi,\partial^2\phi)$) can eventually be inferred to be finite. See also \cref{rem:main:mode2}.
\end{enumerate}

Rather than treating the allowed class of nonlinear equations in full generality, we will simply give a few examples. For all these examples, a version of a quantitative stability estimate sufficient for our purposes would be the global estimate
\begin{equation}\label{eq:main:additional:ass}
	\psi\lesV[\Vc] \tau^{-\epsilon} \min_{q\in[0,1]}\left(\frac{r}{\tau}\right)^q,
\end{equation}
for some arbitrarily small $\epsilon>0$; in particular, with the methods of the present paper, we could derive the $\tau T$-regularity \emph{a posteriori} once this estimate is obtained. In fact, we could also weaken \eqref{eq:main:additional:ass} by restricting, for example, to $q\in[0,1/2]$.

We first discuss power nonlinearities (these results generalise results of~\cite{looi_asymptotic_2024}):
\begin{thm}\label{thm:main:phi3}
	For solutions to $\Box_{g_M}\phi=\phi^3$ arising from smooth, compactly supported and sufficiently small initial data, such that \cref{eq:main:additional:ass} holds for some $\epsilon>0$, we have the following global asymptotics for any $\ell_0\in\N$:
	\begin{equation}\label{eq:main:phi3thm}
	\psi_{\ell\geq\ell_0}-\sum_{m=-\ell_0}^{\ell_0}\appsi[\ell=\ell_0,m] Y_{\ell_0,m}\cdot \frac{(-1)^{\ell}}{(4\ell+2)\binom{2\ell}{\ell}} \int_{\scrip} \left[r^{-\ell-1}w_\ell \psi^3 Y_{\ell,m}\right]^{(\ell+1)}\dd \mu \lesV \tau^{-\ell-3/2+} \min_{q\in[0,\ell_0+1]} \left(\frac{r}{\tau}\right)^q.
	\end{equation}
\end{thm}
Note that small data global existence and decay for $\Box_{g_M}\phi=\phi^P$ with $3\leq P\in\N$ is known e.g.~by \cite{tohaneanu_pointwise_2022}. See  \cref{sec:additionalproof}.

 We may also consider a higher-order power-nonlinearity such as $\Box_{g_M}\phi=\phi^4$, in which case the $\ell=0$-mode will decay one power faster:
\begin{thm}\label{thm:main:phi4}
	For solutions to $\Box_{g_M}\phi=\phi^{4}$ arising from smooth, compactly supported and sufficiently small initial data such that \cref{eq:main:additional:ass} holds for some $\epsilon>0$, then we have the following global asymptotics for any $1\leq\ell_0\in\N$:
	\begin{equation}
		\psi_{\ell\geq\ell_0}-\sum_{m=-\ell_0}^{\ell_0}\appsi[\ell=\ell_0,m] Y_{\ell_0,m}\cdot \frac{(-1)^{\ell}}{(4\ell+2)\binom{2\ell}{\ell}} \int_{\scrip} \left[r^{-\ell-2}w_\ell \psi^4 Y_{\ell,m}\right]^{(\ell+1)}\dd \mu \lesV \tau^{-\ell-3/2} \min_{q\in[0,\ell_0+1]} \left(\frac{r}{\tau}\right)^q,
	\end{equation}
	while, for $\ell=0$, we have both a linear and a nonlinear contribution (cf.~\cref{eq:main:mainthm:l=0})
	\begin{equation}
		\psi_{\ell=0}- T\appsi[0,0]\frac{1}{4\pi}\int_{\scrip}M\psi+ \frac12[r^{-4}\psi^4]^{4}\dd \mu \lesV \tau^{-5/2}\min_{q\in[0,1]}\left(\frac{r}{\tau}\right)^q.
	\end{equation}
\end{thm}
\begin{rem}[Higher power-nonlinearities]
More generally, when considering $\Box_{g_M}\phi=\phi^{P}$ for $5\leq P\in\N$, the leading-order global asymptotics for $\ell\leq P-5$ will only have linear contributions (when $M\neq 0$). For $\ell= P-4$, there will be a linear and a nonlinear contribution, while, for $\ell>P-4$, the nonlinear contribution will dominate.
\end{rem}
\begin{rem}[Mode coupling for power-nonlinearities]\label{rem:main:mode2}
	In contrast to \cref{thm:main:F}, we can take $\ell_0$ arbitrarily large in \cref{thm:main:phi3} and \cref{thm:main:phi4}; essentially, this is because the nonlinearity $\psi^3$ projected onto higher $\ell$ will also only see the higher $\ell$-behaviour of $\psi$. Notice that for the equation $\Box_{g_M}\phi=a(\theta^A)\phi^3$, with $a(\theta^A)$ a function on the sphere supported on all $\ell$, \cref{eq:main:phi3thm} would only hold for $\ell_0\leq2$.
\end{rem}

While the leading-order asymptotics are thus typically determined by the decay of the nonlinearity towards $\scrip$, there may also be cancellations. 
For instance, consider the equation $\Box_{g_M}\phi=\mathfrak{F}[\phi]=\pu\phi\pv\phi$ (for which global existence is known by \cite{luk_null_2010,dafermos_quasilinear_2022}).
Clearly, $r^3\mathfrak{F}$ attains a limit; however, this limit is a total derivative:
\begin{multline}\label{eq:main:tralala}
\pu(\psi/r)\pv(\psi/r)=-\frac{\pu\psi \psi}{r^3}-\frac{\psi^2}{r^4}+\frac{\pu \psi r^2\pv\psi}{r^4}+\frac{\psi r^2\pv \psi}{r^5} \\
\implies [\pu\phi\pv\phi]^{(3)}= -\frac12 \pu(\psi^{(0)})^2, \quad  [\pu\phi\pv\phi]^{(4)}= -(\psi^{(0)})^2+\pu\psi^{(0)}\int_{u_0}^u \Dl\psi^{(0)}-\frac12\pu(\psi^{(0)})^2\dd u'
\end{multline}
so the total integral along $\scrip$ of the $r^{-3}$-term will vanish, and the asymptotics then turn out to be given by:
\begin{thm}\label{thm:main:puphipvphi}
	For solutions to $\Box_{g_M}\phi=\pu\phi\pv\phi$ arising from smooth, compactly supported and sufficiently small data such that \cref{eq:main:additional:ass} holds for some $\epsilon>0$, the following global asymptotics hold for any $\ell_0\geq 1$:
		\begin{equation}
		\psi_{\ell\geq\ell_0}-\sum_{m=-\ell_0}^{\ell_0}\appsi[\ell=\ell_0,m] Y_{\ell_0,m}\cdot \frac{(-1)^{\ell+1}}{(4\ell+2)\binom{2\ell}{\ell}} \int_{\scrip} \left[r^{-\ell-2}w_\ell \pu\phi\pv\phi Y_{\ell,m}\right]^{(\ell+1)}\dd \mu \lesV \left( \tau^{-\ell-3/2+}\right) \min_{q\in[0,\ell_0+1]} \left(\frac{r}{\tau}\right)^q,
	\end{equation}
	while, for $\ell=0$, there is both a linear and a nonlinear contribution (cf.~\cref{eq:main:mainthm:l=0}):
	\begin{equation}\label{eq:main:pupvl=0}
		\psi_{\ell=0}- T\appsi[0,0]\frac{1}{4\pi}\int_{\scrip}M\psi+ \frac12|\sl\psi^{(0)}|^2\dd \mu \lesV \tau^{-5/2+}\min_{q\in[0,1]}\left(\frac{r}{\tau}\right)^q.
	\end{equation}
\end{thm}
\begin{rem}[Comments on the coefficient]
Note that the coefficient in \cref{eq:main:pupvl=0} is not equal to $\int_{\scrip}M\psi+\frac12 [\pu\phi\pv\phi]^{(4)}\dd\mu$, but can be shown to be equal to $\int_{\scrip}M\psi+\frac12 [\pu\phi\pv\phi]^{(4)}+\frac12(\psi^{(0)})^2\dd\mu$ (after integrating by parts in $u$ and on the sphere).
Furthermore, if we considered instead the nonlinear wave equation  $\Box_{g_M}\phi=(g_M^{-1})^{\mu \nu}\partial_\mu\phi  \partial_\nu\phi={D^{-1}}\pu\phi\pv\phi-r^{-2}|\sl\phi|^2$, this coefficient would also vanish; see also \cref{rm:gennullform}.
\end{rem}
\begin{rem}[General null forms]
\label{rm:gennullform}
In general, for constant coefficient null form nonlinearities, the $\ell=0$-tails will appear at the same order as in Price's law, while, for higher $\ell$, the nonlinear contribution leads to a tail decaying one power slower. For non-constant coefficients, the nonlinear contribution will also lead to a slower $\ell=0$-tail. 
For particular null forms such as $\Box_{g_M}\phi=(g_M^{-1})^{\mu \nu}\partial_\mu\phi  \partial_\nu\phi$, the contributions will satisfy Price's law, which also follows from the fact that in that case, the equation can be recast the linear equation: $\Box_{g_M}(e^{-\phi}-1)=0$.
\end{rem}
\begin{rem}[Nonstationary backgrounds]
 \cref{thm:main:phi3,thm:main:phi4,thm:main:puphipvphi} all hold true for $\Box_{g_M}$ replaced by $\Box_g$, with the obvious modifications. 
\end{rem}
\begin{rem}[Generic nonvanishing of the integrals]
	All the integrals in \cref{thm:main:phi3,thm:main:phi4,thm:main:puphipvphi} can be shown to be generically nonvanishing, see \cref{sec:gen}.
\end{rem}

\section{Estimates for the inhomogeneous wave equation on Schwarzschild}\label{sec:inhom}
In this section, we will prove estimates for solutions to the inhomogeneous wave equation on Schwarzschild,
\begin{equation}\label{eq:inhom:wave}
	\Box_{g_M}\phi=F,\qquad  (\iff P_{g_M}\psi=rF)
\end{equation} 
under suitable assumptions on the inhomogeneity $F$ and on the initial data. Throughout this section, we will tacitly make the minimal assumption 
\begin{equation}\label{eq:inhom:tacit}
	F\lesV r^{-1-\epsilon}.
\end{equation}

With the nondegenerate energy $E_N[\phi](\tau_0)$ defined in \cref{eq:inhom:EN} below, we already state the main result of this section (see~\cref{prop:inhom:energydecay} and~\cref{cor:inhom:decay} for the precise versions):
\begin{prop}\label{prop:inhom:main}
Let $\tilde{p}\in(0,2]$, and let $\phi$ solve $\Box_{g_{M}}\phi=F$. Assume that $F$ and the initial data satisfy
	\begin{align}\label{eq:inhom:prop:main:ass}
		r^3F\lesV \tau^{-\eta}\min_{q\in[-1,1]}\left(\frac{r}{\tau}\right)^q,&&	E_N[\phi](\tau_0)+\int_{\C_{\tau_0}} r^{\tilde{p}} (\pv\psi)^2+r^2(\pv T\psi)^2 \dd \mu \lesV[\Vc] 1.
	\end{align}
Then 
	\begin{equation}\label{eq:inhom:prop:main:con}
	\tau T	\psi\lesV  (\tau^{-\frac{\tilde{p}-1}{2}}+\tau^{-\eta+})\min_{q\in[0,1]}\left(\frac{r}{\tau}\right)^{q}, \qquad \text{and} \qquad 	\psi\lesV (\tau^{-\frac{\tilde{p}-1}{2}}+\tau^{-\eta+})\min_{q\in[0,1]}\left(\frac{r}{\tau}\right)^{q} \quad \text{if $\tilde{p}>1$}.
	\end{equation}
	Under suitable additional assumptions on $F$, we can also take $q$ up to $\ell_0+1$ when replacing $\psi$ with $\psi_{\ell \geq \ell_0}$.
\end{prop}

\begin{rem}
The assumption on $F$ in \cref{eq:inhom:prop:main:ass} is directly motivated by the inhomogeneity appearing in \cref{thm:main}, cf.~\cref{eq:time:aprioriF}. It can easily be weakened. In particular, if we can merely take $q\in[-1,0]$ in \cref{eq:inhom:prop:main:ass} (this is the case of \cref{thm:main:F}), then the result \cref{eq:inhom:prop:main:con} still holds with arbitrarily small loss in decay. See also \cref{rem:inhom:rp:u,rem:inhom:rp:generalised}.
	
\end{rem}
We will apply \cref{prop:inhom:main} to the time integral $T^{-1}\psi$. 
More precisely, we will start from a decay estimate for $\psi$ and $F$, then establish that \cref{eq:inhom:prop:main:ass} holds for the time integrals $T^{-1}\psi$ and $T^{-1}F$, and then deduce faster decay for $\psi$ via \cref{eq:inhom:prop:main:con}.

We emphasise that \cref{prop:inhom:main} follows from well-known techniques in the literature. Apart from minor improvements, the only genuinely novel part of this section is the elliptic estimate of \cref{prop:inhom:elliptic2}. For the sake of completeness, we will nevertheless include brief proofs.

\subsection{Boundedness of the $T$-energy and of the $N$-energy}

We first introduce notation for our energies:
\begin{align}
	E_T[\phi](\tau)&=\int_{\Cbar_{\tau}} r^2(\pu\phi)^2 +D|\sl \phi|^2 +D\phi^2\dd \mu+\int_{\C_\tau} r^2(\pv\phi)^2 +D  |\sl\phi|^2+D\phi^2\dd \mu 
	\\
	E_N[\phi](\tau)&=\int_{\Cbar_{\tau}} D r^2(\Xb\phi)^2 +D|\sl \phi|^2+D\phi^2 \dd \mu+\int_{\C_\tau} r^2(\pv\phi)^2 +D  |\sl\phi|^2+D\phi^2\dd \mu\label{eq:inhom:EN}
\end{align}
\newcommand{\fbulk}[2]{\mathfrak{X}_{#1,#2}^{\epsilon}}
\newcommand{\fbulkp}[3]{\mathfrak{X}_{#1,#2}^{\epsilon,#3}}
We also introduce notation for {spacetime} integrals coming from the inhomogeneity:
\begin{align}
\fbulk{\tau_1}{\tau_2}[F]&=\int_{\tau_1}^{\tau_2} \int_{\Sigma_\tau}\tau ^{1+\epsilon} Dr^2(TrF)^2 \dd \mu +\sup_{\tau\in[\tau_1,\tau_2]} \int_{\Sigma_\tau} D r^2(rF)^2\dd \mu\\
	\fbulkp{\tau_1}{\tau_2}{p}[F]&=\fbulk{\tau_1}{\tau_2}[F]+\int_{\tau_1}^{\tau_2}\int_{\Sigma_\tau} D r^{\min(p,\epsilon)+1}(rF)^2 \dd \mu .\label{eq:inhom:fbulkp}
\end{align}
Notice that, by the fundamental theorem of calculus:
\begin{equation}
	\fbulk{\tau_1}{\tau_2}[F]
\lesssim  \int_{\tau_1}^{\tau_2} \int_{\Sigma_\tau}\tau ^{1+\epsilon} Dr^2(TrF)^2 \dd \mu + \int_{\Sigma_{\tau_1}} D r^2(rF)^2\dd \mu
\end{equation}

\begin{prop}[Degenerate and non-degenerate energy boundedness]\label{prop:inhom:energy}
Let $\epsilon>0$ arbitrarily small, let $\tau_2>\tau_1$ and let $\phi$ solve \cref{eq:inhom:wave}. Then there exists $C(\epsilon)$ such that
	\begin{align}\label{eq:inhom:energyT}
		E_T[\phi](\tau_2)\leq C(\epsilon) \left(E_T[\phi](\tau_1)+\fbulk{\tau_1}{\tau_2}[F]\right),&&	E_N[\phi](\tau_2)\leq C(\epsilon) \left(E_N[\phi](\tau_1)+\fbulk{\tau_1}{\tau_2}[F]\right)
	\end{align}

\end{prop}
See \cite{dafermos_lectures_2008} for a more detailed discussion of these estimates. We include a proof for completeness:
\begin{proof}
\textit{The degenerate $E_T$-estimate:}	We consider the identities
	\begin{equation}
	2	r^2DT\phi\cdot \Box_{g_M}\phi=2r^2 D T\phi \cdot F, \quad 2DT\psi \cdot P_{g_M}\psi=2DT\psi\cdot rF,
	\end{equation}
	integrate over spacetime in $\dd u \dd v \dd \sigma$; this gives (the terms along $\hplus$ and $\scrip$ are positive definite)
	\begin{nalign}
&	\int_{\Sigma_{\tau_2}} Dr^2(\partial_r|_{\Sigma}\phi)^2+D(\partial_r|_{\Sigma}(r\phi))^2+ 2|\sl\phi|^2\dd \mu_D\\
		\leq&  \int_{\Sigma_{\tau_1}} Dr^2(\partial_r|_{\Sigma}\phi)^2+D(\partial_r|_{\Sigma}(r\phi))^2+  2|\sl\phi|^2\dd \mu_D+2\int_{\tau_1}^{\tau_2} \dd \tau \int_{\Sigma_\tau} D T(r\phi) \cdot (rF) \dd \mu.
	\end{nalign}
	Notice that the control over $r\partial\phi$ and $\partial\psi$ already gives control over zeroth order terms. 
We treat the inhomogeneous term via integration by parts in $T$: By \cref{eq:inhom:tacit}, there are no boundary terms along $\scrip$ and $\hplus$. The boundary terms along $\Sigma_{\tau_1}$ and $\Sigma_{\tau_2}$ can be absorbed via Young's inequality into the boundary terms on the LHS, and we estimate the bulk term with a weighted Young's inequality using	\begin{equation}
		\int_{\tau_1}^{\tau_2}\int_{\Sigma_\tau} \tau^{-1-\epsilon}r^{-2} \psi^2 \leq C(\epsilon) \sup_{\tau} E_T[\phi](\tau).
	\end{equation}
	This proves the first of \cref{eq:inhom:energyT}.
	
	\textit{The non-degenerate $E_N$-estimate:} To remove the degeneracy at $r=2M$ of the $\pu\phi$-term, we let $\chi$ be a cutoff localising to $r=2M$, identically 1 in $[2M,\rc/2]$ and identically zero for $r>\rc$, and we integrate the identity
	\begin{equation}
	\chi	 \pu \psi \cdot P_{g_M}\psi=\frac1D \pu\psi \left(\pu\pv\psi-\frac{D\Dl\psi}{r^2}+\frac{2MD\psi}{r^3}\right)=\chi \pu\psi \cdot rF.
	\end{equation}
in $\dd \mu$. This gives
	\begin{nalign}\label{eq:inhom:redshift-estimate}
	&	\int_{\Cbar_{\tau_2}^{2M,\rc/2}} D^{-1} (\pu\psi)^2 \dd \mu +\int_{\D_{\tau_1,\tau_2}^{2M,\rc/2}} D^{-1}(\pu\psi)^2 \dd \mu \\
	&	\lesssim 	\int_{\Cbar_{\tau_1}^{2M,\rc/2}} D^{-1} (\pu\psi)^2 \dd \mu +\int_{\tau_1}^{\tau_2}\dd \tau \int_{\Cbar_{\tau}} (\pu\phi)^2+D |\sl\phi|^2+D\phi^2 \dd \mu +\int_{\D_{\tau_1,\tau_2}^{2M,\rc}} D (rF)^2 \dd \mu.
	\end{nalign}
After adding the term $\int_{\tau_1}^{\tau_2} E_T[\phi](\tau)$ this produces an inequality of the form
	\begin{equation}
		E_N[\phi](\tau_2)+{c}\int_{\tau_1}^{\tau_2} E_N[\phi](\tau)\dd\tau{\leq}	E_N[\phi](\tau_1)+{C}\int_{\tau_1}^{\tau_2} E_T[\phi](\tau)\dd \tau+{C}\int_{\tau_1}^{\tau_2}\int_{\Sigma_\tau} Dr^2 F^2\dd \mu ,
	\end{equation}
	with $c,C>0$ uniform constants. Subtract $E_N[\phi](\tau_1)$, divide by $\tau_2-\tau_1$, take the limit $\tau_2-\tau_1\to0$ and introduce an integrating factor of $\e^{c(\tau-\tau_1)}$ to obtain the $N$-energy estimate in \cref{eq:inhom:energyT}.
	\end{proof}
\subsection{Integrated local energy decay estimates}

\begin{prop}[Degenerate and nondegenerate integrated local energy estimates]\label{prop:inhom:iled}
We have, for any $\epsilon>0$, and for any $r'>2M$,
\begin{equation}\label{eq:inhom:iledbasic}
	\int_{\tau_1}^{\tau_2}\dd \tau \int_{\Sigma_\tau^{2M,r'}} \(1-\frac{3M}{r}\)^2\((Dr\partial_r|_{\Sigma}\phi)^2+D|\sl\phi|^2\)+D\phi^2 \dd \mu \leq C(\epsilon, r') \(E_T[\phi](\tau_1)+\fbulkp{\tau_1}{\tau_2}{0}[F] \) .
\end{equation}	
Furthermore, we can remove the degeneracies at $r=2M$ and $r=3M$ via:
\begin{equation}\label{eq:inhom:iledr=3m}
	\int_{\tau_1}^{\tau_2}\dd \tau \int_{\Sigma_\tau^{2M,r'}} r^2(\Xb\phi)^2+|\sl\phi|^2+(T\phi)^2+\phi^2 \dd \mu_D
	\leq C(\epsilon, r') \(E_N[\phi](\tau_1)+E_{T}[\Vk^1\phi](\tau_1) +\fbulkp{\tau_1}{\tau_2}{0}[\Vk^1 F]    \).
\end{equation}
Finally, we have for any $2M<r_1<r_2<\infty$, and for any $N\in\mathbb{N}$, the commuted estimates
\begin{equation}\label{eq:inhom:iledcommute}
	\int_{\tau_1}^{\tau_2}\int_{\Sigma_\tau^{r_1,r_2}}	|\Vc^{N}\phi|^2+(1-\tfrac{3M}{r})^2|\Vc^{N+1}\phi|^2 \dd \mu
	\leq C(\epsilon, r_1,r_2,N) \(E_{T}[\Vk^{N}\phi](\tau_1) +\fbulkp{\tau_1}{\tau_2}{0}[\Vk^{N} F]    \).
\end{equation}
\end{prop}

\begin{proof}
Estimates as above  in were first proved in \cite{dafermos_note_2007,marzuola2010strichartz}, see also \cite{holzegel2024noteintegratedlocalenergy} and \cite{gajic_late-time_2023} for more recent proofs, as well as \cite{metcalfe_prices_2011} and \cref{app:derivapriori} for a proof for perturbations of Schwarzschild.
Here, we only briefly recall the relevant ingredients:
	\cref{eq:inhom:iledbasic} follows from a multiplier estimate with multipliers $f(r)(\pv-\pu)\psi$ and $y(r)\psi$, with suitable choices for $f$ and $y$ in different regimes regimes of angular frequency $\ell$. (In fact, this also gives control over $r^2(
	(\pv-\pu)\phi)^2$ that does not degenerate at $r=3M$.)
The degeneracy at $r=2M$ can be removed using  \cref{eq:inhom:redshift-estimate} and by also including the non-degenerate energy $E_N$ on the RHS.
{Furthermore, by commuting with the Killing vector field $T$ and the angular momentum operators, using that the zeroth-order term $\phi^2$ in \cref{eq:inhom:iledbasic} does not degenerate, we can control all derivatives of $\phi$ without degeneracy at $r=3M$. Alternatively, it suffices to commute only with $T$ and then apply $f(r)(\pv-\pu)\psi$ with a suitable choice of $f$ supported near $r=3M$ to control the angular derivatives without degeneracy.}
\end{proof}

\subsection{The hierarchy of $r^p$-weighted energy estimates and commutations}
Next, we consider the hierarchy of $r^p$-weighted energy estimates introduced in \cite{dafermos_new_2010}:
\begin{prop}[$r^p$-weighted estimates]\label{prop:inhom:rp}
	Let $p\in[0,2]$ and let $\epsilon>0$ arbitrarily small. Then we have
		\begin{align}\label{eq:inhom:rp}
		&E_N[\phi](\tau_2)+	\int_{\C_{\tau_2}} r^p (\pv \psi)^2\dd \mu +\int_{\tau_1}^{\tau_2}\dd \tau \int_{\C_{\tau}} pr^{p-1}(\pv\psi)^2+(2-p)r^{p-3}|\sl\psi|^2\dd \mu\\
		\lesV[\Vc^{N}] &C(\epsilon,\rc)    \(	\int_{\C_{\tau_1}} r^p (\pv \psi)^2  \dd \mu+E_N[\phi](\tau_1)+E_T[\Vk^1\phi](\tau_1) +\fbulkp{\tau_1}{\tau_2}{p}[\Vk^1 F]   \)\nonumber
	\end{align}
\end{prop}
\begin{proof}
	The uncommuted version of this estimate was first proved in \cite{dafermos_new_2010}: It follows from simply integrating the identity $\chi\pv(r\phi)\cdot DP_g(r\phi)=\chi \pv(r\phi)\cdot D rF$ over the spacetime region $\D_{\tau_1,\tau_2}^{\rc,\infty}$, for $\chi(r)$ a cutoff vanishing for $r\leq \rc-1/2$ and identically one for $r\geq \rc$. The term arising from $\chi'$ is controlled by \cref{eq:inhom:iledcommute}.  We treat the inhomogeneous term by using Young's inequality, where we note that we can directly absorb the spacetime integral of $r^{p-1}(\pv\psi)^2$ into the LHS only for $p>0$---this is the reason for the $\min(p,\epsilon)$ in \cref{eq:inhom:fbulkp}.
	
	The statement that we can commute with $\Vc$ follows after applying \cref{prop:inhom:rX} and \cref{prop:inhom:redshift} below. Notice already that commutations with $\Vk$ are trivial, and commutations in the bounded-$r$ region away from $r=2M$ are already covered by \cref{eq:inhom:iledcommute}.
	Thus, it remains to show that we can commute with $r\Xx$ near $\scrip$ (first done in \cite{angelopoulos_late-time_2021}, with $\Xx$-commutations having appeared in this context already in \cite{moschidis_rp_2016}), and that we can commute with $\Xb$ near $\hplus$ \cite{dafermos_lectures_2008}.
	\end{proof}
	\begin{rem}\label{rem:inhom:rp:u}
		In the above, we have estimated the inhomogeneous bulk term simply by using Young's inequality and losing an $r$-weight. In certain applications, given that $F$ has sufficient $\tau$-decay and less $r$-decay, one may instead wish to estimate it against the fluxes at the expense of a $\tau$-weight!
	\end{rem}
	\subsubsection{Commutations with $r\Xx$ near $\scrip$}
	In this section, we will derive the estimates that are necessary to be able to commute the $r^p$-weighted energy estimates of \cite{dafermos_new_2010} with vector fields in $\Vc$.
\begin{prop}[The $r\pv$-commuted $r^p$-weighted energy estimates]\label{prop:inhom:rX}
	Let $N\in\mathbb{N}$, let $\epsilon>0$, and let $p\in[-2N,2]$. Then we have
	\begin{align}\nonumber
	&	\sum_{0\leq n_1+n_2\leq N}	\int_{\C_{\tau_2}}  r^p(\pv(r\Xx)^{n_1} T^{n_2}\psi )^2\dd \mu \\\nonumber
		&+\int_{\tau_1}^{\tau_2}\dd \tau \int_{\C_\tau} (p+2N)r^{p-1}(\pv(r\Xx)^{n_1}T^{n_2}\psi)^2+(2-p)r^{p-3} (\sl(r\Xx)^{n_1}T^{n_2}\psi)^2+(2-p)r^{p-3}((r\Xx)^{n_1}T^{n_2}\psi)^2\dd \mu\\
		&\lesssim_{N,\epsilon, \rc} \sum_{0\leq n_1+n_2\leq N}\int_{\C_{\tau_1}} r^p(\pv(r\Xx)^{n_1} T^{n_2}\psi )^2\dd \mu \\
		& +E_T[\Vk^N\phi](\tau_1)+\fbulk{\tau_1}{\tau_2}[\Vk^N F]+\sum_{0\leq n_1+n_2\leq N}\int_{\tau_1}^{\tau_2}\dd \tau	\int_{\C_{\tau}} D r^{\min(p,\epsilon)+1}((r\Xx)^{n_1}T^{n_2}(rF))^2 \dd \mu\nonumber
		\end{align}
\end{prop}
\begin{proof}
	The starting point of the proof is the following inductive relation for $n\in\N$, where we write $y=D/r$: 	\begin{nalign}
		\pu(y^{n+1}\pv(r\Xx)^{n+1} \psi)&=\pv\pu(y^n \pv (r\Xx)^n \psi)-y^{n+1}\pu\pv \log (y^{n+1}) (r\Xx)^{n+1} \psi -\pv \log (y^{n+1}) \pu(y^n\pv (r\Xx)^n\psi)\\
		&	=y^{n+1}\pv (y^{-n-1} \pu(y^n \pv(r\Xx)^n\psi))-y^{n+1}\pu\pv \log (y^{n+1}) (r\Xx)^{n+1} \psi 
	\end{nalign}
	This relation then directly implies, for some $A_k^n\in\N$ and for some smooth bounded functions $f_k^n(r), g_k^n(r)$:
	\begin{nalign}\label{eq:inhom:rXcommutation}
		\frac{r^n}{D^n}\pu\(\frac{D^n}{r^n}\pv (r\Xx)^n\psi\)&=\frac{D\Dl(r\Xx)^n\psi}{r^2}-\frac{nD\Dl (r\Xx)^{n-1}\psi}{r^2}+D\sum_{k=0}^n A_k^n (r\Xx)^k (rF)\\
		&+n\sum_{k=0}^n \frac{f_k^n(r)}{r^2} (r\Xx)^k \psi+n(n-1)\sum_{k=0}^{n-2} \frac{g_k^n(r)}{r^2}\Dl(r\Xx)^n\psi.
	\end{nalign}

We then multiply \cref{eq:inhom:rXcommutation} by $\chi r^p\pv(rX)^n\psi$ and integrate over $\D_{\tau_1,\tau_2}^{r',\infty}$, where $\chi$ is a cutoff vanishing for $r\leq \rc-1$. (The arising bulk terms are controlled by \cref{eq:inhom:iledcommute}.) All terms in the second line of \cref{eq:inhom:rXcommutation} can be treated either by largeness of $r$ or by induction assumption, and we crucially use the sign of the Laplacian term multiplied by $n$ in the first line of \cref{eq:inhom:rXcommutation}. See also \cite{gajic_late-time_2023}[Proposition 7.2] for a more detailed proof (with $F=0$).
\end{proof}
\cref{prop:inhom:rX}, in addition to establishing $rX$-regularity, also allows us to later on prove faster decay for time derivatives (an insight going back to \cite{angelopoulos_vector_2018}), using that we can convert time derivatives into $r\pv$ derivatives while gaining $r$ weights:
\begin{lemma}[Converting $T$ into $L$ derivatives]\label{lem:inhom:Tconversion}
	Let $p\in \R$ and $N\in\mathbb N$. Then
	\begin{multline}
		\sum_{n_1+n_2\leq N} \int_{\C_\tau} r^{p+1}|\sl^{n_1} \pv T (r\Xx)^{n_2}\psi|^2 \lesssim_N \sum_{n_1+n_2\leq N+1} \int_{\C_\tau} r^{p-1}|\sl^{n_1}\pv(r\Xx)^{n_2}\psi|^2+r^{p-3}|\sl^{n_1+1}(r\Xx)^{n_2}\psi|^2\\
		+	\sum_{n_1+n_2\leq N} \int_{\C_\tau} r^{p+1}|\sl^{n_1}(r\Xx)^{n_2}(rF)|^2
	\end{multline}
\end{lemma}
\begin{proof}
	This follows directly from \cref{eq:inhom:rXcommutation}, which allows to rewrite schematically $r\pv T=\pv(r\pv)-D\pv+r\pu\pv=\pv(r\pv)-D\pv+r^{-1}\Dl+r\cdot rF$.
\end{proof}
\subsubsection{Red-shift estimates commuted with $\Xb$ near $\hplus$}
The red-shift estimates commute with vector fields in $\Vc$. This follows from the general construction in \cite{dafermos_lectures_2008}[Section 7]. We include a brief proof for the sake of completeness.
\begin{prop}[Redshift-commutations]\label{prop:inhom:redshift}
	For $N\in\mathbb{N}$, $\epsilon>0$, and for any $r'>2M$, we have
	\begin{nalign}\label{eq:inhom:full}
	&E_N[\Vc^N\phi](\tau_2)+E_T[\Vc^{N+1}\phi](\tau_2)+	\int_{\tau_1}^{\tau_2}\dd \tau \int_{\Sigma_\tau^{2M,r'}} r^2(\partial_r|_{\Sigma}\Vc^N\phi)^2+|\sl\Vc^N\phi|^2+(\Vc^N\phi)^2 \dd \mu_D\\
	\leq &C(\epsilon,r',N)   \(	E_N[\Vc^N\phi](\tau_1)+E_T[\Vc^{N+1}\phi](\tau_1) +\fbulk{\tau_1}{\tau_2}[\Vc^{N+1} F]\)
\end{nalign}
\end{prop}
\begin{proof}
	Restricted to the region $r\geq \rc$, this estimate follows from \cref{prop:inhom:rX} with $p=0$. Control over the commuted energies in the intermediate region, i.e.~away from $\hplus$, follows from an elliptic estimate as for \cref{eq:inhom:iledcommute}.
	It remains to extend the result to the horizon. For this, we commute with $\Xb$:
	We observe the following commutation relations:
	\begin{nalign}
		&\pu\pv(\Xb^n\psi)=\Xb^n \pu\pv\psi-\frac{2Mn}{r^2}\pu \Xb^n \psi+\sum_{k=0}^{n-1} \underline{f}_n^k(r) \pu(\Xb)^k \psi\\
		=&\frac{D\Dl\Xb^n\psi}{r^2}-\frac{2MD\Xb^n\psi}{r^3}+\sum_{k=0}^{n-1} \underline{g}_k^n(r) D (\Dl \Xb^k\psi)+\underline{h}_k^n(r) D \Xb^k\psi -\frac{2Mn}{r^2}\pu \Xb^n \psi+\sum_{k=0}^{n-1} \underline{f}_n^k(r) \pu(\Xb)^k \psi,
			\end{nalign}
			where $\underline{f}_k^n, \underline{g}_k^n, \underline{h}_k^n$ are some smooth functions in $r$.
We integrate this identity over a spacetime region close to the horizon, where we crucially exploit the good sign $(-2Mn)$ of the $\pu\Xb^n\psi$-term.\footnote{{We can express: $\frac{2Mn}{r^2}|_{r=r_+}=2n \kappa_+$, with $\kappa_+=\frac{1}{4M}$ the surface gravity of the Schwarzschild event horizon. The increase of the factor $2n \kappa_+$  with $n$ is called the enhanced redshift effect. See also \cite{dafermos_lectures_2008}[Section 7].}}
\end{proof}

\subsection{Elliptic estimates}\label{sec:inhom:elliptic}
In this subsection, we prove estimates for the part of the wave operator without time-derivatives, viewing it as a degenerate (at $r=2M$) elliptic operator along $\Sigma_{\tau}$.
Along $\C_\tau$ and $\Cbar_\tau$, respectively, we have
\begin{equation}\label{eq:inhom:ELL}
	\Ell \phi :=\Xx(r^2D\Xx\phi)+\Dl\phi= r\Xx (rT\phi)+r^2F, \qquad \Ellbar\phi:= \Xb(r^2D\Xb\phi)+\Dl\phi=-r\Xx(rT\phi)+r^2F.
\end{equation}
We shall prove two different kinds of elliptic estimates. 
First, we prove an elliptic estimate that controls $r^{-q}$-weighted $L^2$ norms along $\Sigma_\tau$ in terms of $\Ell\phi$, where $q$ can be taken to be larger for higher $\ell$. The purpose of this estimate is to prove faster time decay for higher angular modes, such estimates first appeared in \cite{angelopoulos_late-time_2021}[Section~7].

Second, we shall prove a new (but essentially entirely Minkowskian) $(r^2\Xx)^N$-commuted elliptic estimate along $\C_{\tau}$ (i.e.~for large $r$); the purpose of this estimate is to later on control the $p=2$ flux of the $N$-th time integral of $\psi$ against the $p=2$ flux of $(r^2\pv)^N\psi$. 
\subsubsection{$r^{-q}$-weighted elliptic estimates}
\begin{prop}[$r^{-q}$-weighted estimates for $\mathcal{L}$]\label{prop:inhom:elliptic1}
	Let $\phi$ solve $\Box_{g_M}\phi=F$ and be supported on $\ell\geq \ell_0\geq 0$. Let $q\in(-2\ell_0-1, 2\ell_0+1)$. 
	Then we have the following estimate:
	\begin{align}\label{eq:inhom:ellipticmain1}\nonumber
&	\int_{\Cbar_{\tau}}r^{-q}\((D r^2\Xb^2\phi )^2+r^2(\Xb \phi)^2+|\sl\phi|^2+\phi^2\)\dd \mu_D+\int_{\C_{\tau}} r^{-q}\((D r^2\Xx^2\phi )^2+r^2(\Xx \phi)^2+|\sl\phi|^2+\phi^2\)\dd \mu_D\\
	&	\lesV[\Vk]_{q,\ell_0}  \int_{\Cbar_{\tau}}r^{-q}(\Ellbar\phi)^2\dd \mu_D+\int_{\C_{\tau}}r^{-q} (\Ell\phi)^2\dd \mu_D+\sum_{k=0}^1 \int_{\S_{\tau}}(\sl^k  T \phi)^2 \dd \sigma.\\
		&	\lesssim \sum_{k=0}^1\int_{\Sigma_{\tau}}r^{-q+2}\left((\partial_r|_{\Sigma}\sl^kT\psi)^2+r^{-2}(\sl^k r^2 F)^2\right)\dd \mu,\nonumber \text{	where only the last {inequality}  uses $\Box_{g_M} \phi=F$.}
	\end{align}
	Furthermore, for $N\in\N$, we  have, for $\underline{\phi}^{[n]}=(r\Xb)^n\phi$ and $\phi^{[n]}=(r\Xx)^n\phi$:
	\begin{align}\label{eq:inhom:ellipticmain1:commuted}\nonumber
	&\sum_{n=0}^{N}	\int_{\Cbar_{\tau}}r^{-q}\((D r^2\Xb^2\underline{\phi}^{[n]} )^2+r^2(\Xb \underline{\phi}^{[n]})^2+|\sl\underline{\phi}^{[n]}|^2\)\dd \mu_D+\int_{\C_{\tau}} r^{-q}\((D r^2\Xx^2\phi^{[n]} )^2+r^2(\Xx \phi^{[n]})^2+|\sl\phi^{[n]}|^2\)\dd \mu_D\\
	&	\lesV[\Vk]_{q,\ell_0,N}\sum_{n_1+n_2+n_3\leq N}\int_{\Sigma_{\tau}}r^{-q+2}\left((\partial_r|_{\Sigma}\sl^{n_2}T^{n_1+1}(r\partial_r|_{\Sigma})^{n_3}\psi)^2+r^{-2}(\sl^{n_2}T^{n_1}(r\partial_r|_{\Sigma})^{n_3} (r^2 F))^2\right)\dd \mu.
\end{align}
\end{prop}	
\begin{proof}
	Notice that the second {inequality} of \cref{eq:inhom:ellipticmain1} follows from \cref{eq:inhom:ELL} and the fundamental theorem of calculus to estimate the $\S_\tau$-terms. 	For the proof of the first estimate of \cref{eq:inhom:ellipticmain1}, we write $\phi=\phi_{\ell<\ell_1}+\phi_{\ell\geq \ell_1}$ for some $\ell_1$ suitably large. 
		For the remainder of the proof, we assume without loss of generality that $\phi\in C^{\infty}_c(\Sigma_{\tau})$. The more general case then follows from a standard density argument. We closely follow the proof of \cite{gajic_late-time_2023}[Proposition~9.1].

	\textit{The lower $\ell$-modes $\ell_0\leq \ell<\ell_1$:}
	We use  $w_{\ell}\Ell u_{\ell}=\Xx(r^2Dw_{\ell}^2\Xx\philc)$ (cf.~\cref{eq:Sch:EllEllbar}) and  integrate the following product over $\C_{\tau}$:
	\begin{nalign}
		&r^{-q+1}\Xx \philc\cdot w_\ell \Ell \phi_\ell=r^{-q+1} \Xx \philc \cdot \Xx(r^2D w_{\ell}^2\Xx \philc)\\
		=&\Xx\(\frac12 (\Xx\philc)^2r^{3-q}Dw_{\ell}^2\)+(\Xx\philc)^2w_{\ell}^2r^{-q+1}\((r-M)+\frac{q-1}{2}(r-2M)+r(r-2M)(\log w_{\ell})'\).
	\end{nalign}
	For $r$ sufficiently large, since $\log w_\ell'= r^{-1}\ell+O(r^{-2})$, the bulk term is coercive provided that $q>-1-2\ell$; for the intermediate $r$-region, we simply choose $q$ suitably large.
	Applying Young's inequality thus gives
	\begin{equation}\label{eq:inhom:elliptic:suitablyadd1}
		\int_{\C_{\tau}}r^{-q+2}w_{\ell}^2(\Xx\philc)^2 \dd \mu_D\lesssim \int_{\S_{\tau}}(\Xx\philc)^2\dd \sigma+\int_{\C_{\tau}} r^{-q}(\Ell\phi)^2 \dd \mu_D.
	\end{equation}

	We perform exactly the same computation along $\Cbar_{\tau}$, 
		\begin{nalign}
		&r^{-\tilde{q}+1} \Xb \philc \cdot \Xb(r^2D w_{\ell}^2\Xb \philc)\\
		=&\Xb\(\frac12 (\Xb\philc)^2r^{3-\tilde{q}}Dw_{\ell}^2\)+(\Xb\philc)^2w_{\ell}^2r^{-\tilde{q}+1}\((r-M)+\frac{\tilde{q}-1}{2}(r-2M)+r(r-2M)\log w_{\ell}'\).
	\end{nalign}
	Very close to $r=2M$, the bulk term is manifestly positive. We may thus choose $\tilde{q}$ sufficiently large so that the $\log w_\ell'$ term is dominated. 
	Since $r$ is bounded along $\Cbar_{\tau}$, the following estimate then holds for any $q$:
		\begin{equation}\label{eq:inhom:elliptic:suitablyadd2}
		\int_{\Cbar_{\tau}}r^{-q+2} w_{\ell}^2(\Xb\philc)^2 \dd \mu_D+\int_{\S_{\tau}}(\Xb\philc)^2\dd \sigma\lesssim \int_{\Cbar_{\tau}} r^{-q}(\Ellbar\phi)^2 \dd \mu_D.
	\end{equation}
	We may add \cref{eq:inhom:elliptic:suitablyadd1,eq:inhom:elliptic:suitablyadd2} to obtain the desired result \cref{eq:inhom:ellipticmain1} using also that
	\begin{equation}
		\int_{\S_{\tau}}(\Xx\philc)^2\dd \sigma\lesssim \int_{\S_{\tau}}(\Xb\philc)^2+(T\philc)^2\dd \sigma,
	\end{equation}
	as well as an application of Hardy's inequality (for $q<2\ell+1$) to gain control over the zeroth order terms.\footnote{Later on, we will also make use of the estimates above for $(\Xx \philc)^2$ directly without applying Hardy's inequality, in which case there is no upper bound on $q$.\label{footnote:hardy}}
	
	\textit{The remaining high $\ell$-modes $\phi_{\geq \ell_1}$, degenerate at $\hplus$:}
	We consider directly the square:
	\begin{align}
		(\Ell\phi)^2=(\Xx(r^2D\Xx\phi))^2+\(\Dl\phi\)^2+2\Xx(r^2D\Xx\phi)\cdot \Dl\phi.
	\end{align}
	Now, we integrate the mixed term by parts over $\C_{\tau}$:
	\begin{equation}
		\int_{\C_{\tau}} r^{-q}2\Xx(r^2D\Xx\phi)\cdot \Dl\phi\dd\mu_D\
		=\int_{\C_{\tau}} -2\Xx(r^{2-q}D(\sl\Xx\phi)\cdot\sl\phi)+2Dr^{2-q}|\sl\Xx\phi|^2-2q r^{-q+1}r^2D(\sl\Xx\phi)\cdot \sl\phi \dd\mu_D.
	\end{equation}
We can control the last term by applying Young's inequality and Poincar\'e's inequality on $S^2$ (for $\ell_1$ sufficiently large) to establish the following bound:
	\begin{equation}\label{eq:inhom:proof:absorb}
		\int_{\C_\tau} r^{-q}\((\Xx(r^2D\Xx\phi))^2+(\Dl\phi)^2+Dr^2|\sl\Xx\phi|^2\) \dd \mu_D-\int_{\S_{\tau}}\sl\Xx\phi\cdot \sl \phi \dd \mu_D\lesssim \int_{\C_{\tau}}r^{-q}(\Ell\phi)^2\,d\mu_D.
	\end{equation}
	We perform an analogous argument along $\Cbar_{\tau}$ and add the corresponding estimate to \cref{eq:inhom:proof:absorb}. The resulting integrals on $\S_{\tau}$ satisfy:
	\begin{equation}
		\int_{\S_{\tau}} \sl(\Xx-\Xb)\phi\cdot \sl \phi\leq \int_{\S_{\tau}}\epsilon^{-1}|T\sl\phi|^2 +\epsilon |\sl \phi|^2.
	\end{equation}
	We may absorb the $\epsilon$-term into the LHS of \cref{eq:inhom:proof:absorb} via Hardy inequality; this gives the result up to the degeneracy in $D$ near the event horizon.
	
	\textit{Removing the degeneracy at $\hplus$ for $\phi_{\geq \ell_1}$:}
	We apply the redshift multiplier $\chi r\Xb \phi$, for $\chi$ a cutoff localising to $r\leq 4M$, say, and we observe that
	\begin{nalign}
		&\chi r\Xb\phi \Ellbar \phi=\chi r\Xb \phi\cdot (r^2D \Xb^2\phi +(r^2D)' \Xb \phi+\Dl\phi)\\
		=&\chi \Xb(\frac12 (\Xb\phi)^2 r^3 D)+\chi \frac{r}{2} ((r^2D)'-(rD)) (\Xb\phi)^2 -\chi \frac 12 \Xb(r|\sl\phi|^2)+\chi \frac12 |\sl\phi|^2.
	\end{nalign}
We integrate over $\Cbar_\tau$; the boundary term at the horizon comes with a good sign, we have $(r^2D)'-rD=r>0$, and we can control the bulk terms arising from $\Xb\chi$ from our already established degenerate bound.

\textit{Notice that this argument works directly for all angular frequencies. Furthermore, the cut-off is not actually necessary if we also consider the multiplier $r\Xx\phi$ along $\C_{\tau}$; that is, the red-shift multiplier works globally. It is this estimate that carries over also to more general spacetimes such as that of Kerr; see \cite{angelopoulos_late-time_2021}[Proposition 9.1].}

\textit{Commutations with $r\partial_r|_\Sigma$ for $\phi_{\ell\geq\ell_1}$.}
We can directly commute with $\sl$ and $T$, since both differential operators commute with $\mathcal{L}$. 
For commutations with $r\partial_r|_{\Sigma}$, we first observe the following commutation relations (entirely analogous ones hold for $\Ellbar$):
\begin{equation}
	\Ell(r\Xx)^N \phi+2MN \Xx(r\Xx)^N\phi +\sum_{n=0}^N O_{\infty}(r^{-1}) (r\Xx)^n \phi=(r\Xx)^N \Ell \phi.
\end{equation}
The proof now proceeds inductively in $N$: Denoting $(r\Xx)^N\phi=\phi^{[N]}$, we consider the square $r^{-q}((r\Xx)^N\Ell \phi)^2=r^{-q}(\Xx(r^2D\Xx \phi^{[N]}+\Dl\phi^{[N]}+2MN\Xx\phi^{[N]})+\dots)^2$. The lower order $\dots$-terms are controlled inductively using Young's inequality. We deal with the remaining three cross terms as follows:
Firstly, 
\begin{multline}\label{eq:inhom:ell:commute1}
	\int r^{-q}2 \Dl\phi^{[N]} \Xx(r^2D\Xx \phi^{[N]})=\int \Xx (2r^{-q}\Dl \phi^{[N]} r^2 D \Xx \phi^{[N]})+\int q r^{-q-1}2 \Dl \phi^{[N]} r^2 D\Xb \phi^{[N]}+\int r^{-q+2} 2D ( \sl\Xx \phi^{[N]})^2,
\end{multline}
where the second term on the RHS can be absorbed exactly as in the $N=0$-case; 
secondly,
\begin{equation}\label{eq:inhom:ell:commute2}
	\int r^{-q} 4MN \Xx(r^2D\Xx \phi^{[N]})\Xx \phi^{[N]}=\int 2MN\Xx(r^{2-q} D (\Xx \phi^{[N]})^2) +\int 2MN\left(q rD+(r^2D)'\right) r^{-q} (\Xb \phi^{[N]})^2,
\end{equation}
where the last term is lower-order in $r$ and nonnegative for $q\geq-1$;
 thirdly, 
\begin{equation}\label{eq:inhom:ell:commute3}
	\int r^{-q} 4MN \Xx \phi^{[N]} \cdot \Dl \phi^{[N]}=\int -\Xx(r^{-q}2MN |\sl\phi^{[N]}|^2)+\int -q r^{-q-1} 2MN |\sl\phi^{[N]}|^2,
\end{equation}
where the last-term is again lower-order in $r$ and nonnegative for $q\leq0$.
Thus, all $q$- and $N$-dependent terms can be absorbed by using the largeness of $\ell_1$ and of $r$. 
In order to deal with the boundary terms at $\S_\tau$, we note that in \cref{eq:inhom:ellipticmain1}, we also control $\int_{\S_{\tau}} (\Xx\phi)^2+(\Xb\phi)^2$  via Hardy's inequality, and analogous $T$- and $\sl$-commuted quantities. We can thus control all boundary terms coming from \cref{eq:inhom:ell:commute1}--\cref{eq:inhom:ell:commute3} by simply writing e.g.~$X^2\phi\sim\Xx(\Xb+T)\phi$, applying the wave equation, and then proceeding inductively.

Notice that \cref{eq:inhom:ell:commute1}--\cref{eq:inhom:ell:commute3} also hold with $\Xx$ replaced by $\Xb$.
Thus, if we replace $\Xx$ by $\Xb$ and  integrate along $\Cbar_\tau$ and choose, say, $q=-1$, then this already gives \textit{nondegenerate} control for $N>0$ (in view of the last term in \cref{eq:inhom:ell:commute2}). Note that this indeed works \textit{for all} $\ell$, as the only problematic term, namely the penultimate term in~\cref{eq:inhom:ell:commute1}, vanishes for $\ell=0$ and can be controlled for $\ell\geq 1$ as before. Thus, we already get coercive control along $\Cbar_\tau$ for all $\ell$.

\textit{Commutations with $r\partial_r|_\Sigma$ for $\phi_{\ell<\ell_1}$:}
Finally, we restrict to fixed $\ell$-modes:
Observe first that 
\begin{align}
	w_\ell \left(\Ell(r\Xx)^N \phi_\ell+2MN \Xx(r\Xx)^N \phi_\ell\right)=\Xx (r^2D w_\ell^2 \Xx(w_\ell)^{-1}(r\Xx)^N\phi_\ell)+2MN w_\ell \Xx(r\Xx)^N \phi_\ell.
\end{align}
We can then use largeness in $r$ (the bounded $r$-region, where we may choose $q$ freely, is already covered by the above) and argue as in the uncommuted case.
More precisely, we consider the product
\begin{align}
	&r^{-q}\Xx (w_\ell^{-1}(r\Xx)^N\phi_\ell) \cdot w_\ell ((r\Xx)^N\Ell\phi_\ell )\\
	=&r^{-q}\Xx (w_\ell^{-1}(r\Xx)^N\phi_\ell) \cdot \left(\Xx(r^2Dw_\ell^2 \Xx (w_\ell^{-1}(r\Xx)^N\phi_\ell))+2MN w_\ell^2 \Xx(w_\ell^{-1}(r\Xx)^N\phi_\ell)+w_\ell\sum_{n=0}^N O_{\infty}(r^{-1})(r\Xx)^n \phi_\ell\right).\nonumber
\end{align}
All the terms in the sum $\sum_{n=0}^N$ can be controlled inductively by Young's inequality.
The first product is exactly as in the uncommuted case; the second product is manifestly positive (and is lower order in $r$ as well). This completes the proof of \cref{eq:inhom:ellipticmain1:commuted}.
\end{proof}
\subsubsection{$r^2\Xx$-commuted elliptic estimates near $\scrip$ for high angular frequencies}
In addition to the $r^{-q}$-weighted elliptic estimates of \cref{prop:inhom:elliptic1}, we also need estimates featuring $r^2\Xx$ commutations in a region where $r\gg M$, which only work for solutions supported on sufficiently large angular frequencies. 
The purpose of these estimates is to later on be able to estimate the $r^p$.energy flux of the $n$th time integral against the $r^p$-energy flux of $(r^2X)^n \psi$. In practice, we will only ever apply this estimate at $\Sigma_{\tau_0}$. Being an estimate concerning the large $r$-region, it is Minkowskian in nature.

For just this section, we define, for $N\in\mathbb{N}$, $\tilde{\psi}^{[N]}:=(r^2\Xx)^N\psi$. 
\newcommand{\psiNN}[1][N]{\tilde{\psi}^{[#1]}}
If $\Box_{g_M}\phi=F$, we have the following commutation relations, which are readily confirmed via an induction argument (see e.g.~\cite{kehrberger_case_2024}[Proposition~9.2 with $s=0$]):
\begin{nalign}\label{eq:inhom:r2Xcommuted}
&	\pu(r^{-2N}D^N\pv \psiNN)=(T-D\Xx)(r^{-2N}D^{N-1}\Xx \psiNN)\\&=\frac{D^{N+1}}{r^{2N+2}}\left((N(N+1)+\Dl)\psiNN+\frac{M}{r}a_N \psiNN+M b_N \psiNN[N-1]\right)+\frac{D^{N+1}}{r^{2N+2}}(r^2\Xx)^N(r^3F).
\end{nalign}
 for some $a_N,b_N\in\mathbb{N}$ (more precisely, $a_N=-2N^3$ and $b_N=-6N(N+1)-2$). We denote
 \begin{equation}
 	r^{-2N}\tilde{\Ell}_N \psiNN:=D\Xx(r^{-2N}D^{N-1}\Xx \psiNN)+\frac{D^{N+1}}{r^{2N+2}}\left((N(N+1)+\Dl)\psiNN+\frac{M}{r}a_N \psiNN+M b_N \psiNN[N-1]\right).
 \end{equation}

\begin{prop}\label{prop:inhom:elliptic2}
Let $N\in\N$,  let $q\in\R$. Then there exists a sufficiently large $\ell_0(N,q)\in\N$, such that for solutions $\phi$ to $\Box_{g_M}\phi=F$ and $\phi=\phi_{\ell\geq\ell_0}$: 
	\begin{nalign}
	&\int_{\C_{\tau}} r^{q-2}\left(\frac{|\sl\Xx \psiNN[N]|^2}{r^{4N}}+\frac{|\Dl\psiNN[N]|^2}{r^{4N+2}} \right)\dd \mu 
	\lesV[\Vc]_{\ell_0,N,q} \sum_{n=0}^N\int_{\C_{\tau}} r^q\frac{|T\Xx \psiNN[n]|^2}{r^{4N}}+r^q \frac{((r^2\Xx)^n(r^3F))^2}{r^{4N+4}}\dd \mu
	\\
	&\qquad\qquad+\sum_{n_1+n_2+n_3\leq N}\int_{\Sigma_{\tau}}r^{-2}\left((\sl^{n_1}\partial_r|_{\Sigma}T^{n_2+1}(r\partial_r|_{\Sigma})^{n_3}\psi)^2+r^{-2}(\sl^{n_1}T^{n_2}(r\partial_r|_{\Sigma})^{n_3} (r^2 F))^2\right)\dd \mu.
	\end{nalign}

\end{prop}
\begin{proof}
	As in the proof of \cref{prop:inhom:elliptic1}, we may assume without loss of generality that $\phi\in C^{\infty}_c(\Sigma_{\tau_0})$.
We consider the square $r^{-4N+q}(D^{-1}\tilde{\Ell}_N\psiNN)^2$, and consider the cross-term
\begin{nalign}
2&	r^{q-2N-2}D^{N}\Xx(r^{-2N}D^{N-1} \Xx \psiNN) (N(N+1)+\Dl)\psiNN=2\Xx(r^{q-4N-2} D^{2N-1}\Xx\psiNN(N(N+1)+\Dl)\psiNN )\\
&	-2\Xx\(\frac{D^N}{r^{2N+2-q}}\) r^{-2N} D^{N-1}\Xx\psiNN (N(N+1)+\Dl)\psiNN + \frac{2D^{2N-1}}{r^{4N+2-q}}\left( |\sl\Xx\psiNN|^2-N(N+1)(\Xx\psiNN)^2 \right)+\mathring{\div}(\dots)
\end{nalign}
The first term on the second line can be absorbed, via Young's inequality, into the second term on the second line as well as the $((N(N+1)+\Dl)\psiNN)^2$ in the square $(D^{-1}\tilde{\Ell}_N\psiNN)^2$, provided that $\ell$ is sufficiently large depending on $N$ and $q$.

On the other hand, the boundary term at $\S_{\tau}$ can be controlled via the control established already in the previous proposition.

Finally, the cross terms featuring $a_N$ and $b_N$ can be inductively absorbed via Young's inequality.

 For additional $r\Xx$-commutations, the argument is analogous to that of \cref{prop:inhom:elliptic1}. 
\end{proof}
\subsection{Energy  decay in time and improved decay for time derivatives}
We now assume that we are given initial data with finite $r^p$-weighted energies for $\phi$ and its suitable higher-order derivatives. We additionally assume that we are given some decay for the inhomogeneity. The aim is to deduce energy decay in time, as well as improved decay for time derivatives. We note that the method of deducing energy decay from \cref{eq:inhom:rp} was introduced in \cite{dafermos_new_2010}, and the method of inferring improved decay for each additional $T$-derivative was introduced in \cite{angelopoulos_vector_2018}.

In the proposition below, we split $F=F_1+F_2$. In the context of the proof of \cref{thm:main}, $F_1$ will be the term that is allowed to depend non-trivially on $\phi$ and its derivatives, whereas $F_2$ will be a fixed inhomogeneous function, arising, for instance, from the subtraction of an approximate solution.

\begin{prop}[Energy decay with inhomogeneity]\label{prop:inhom:energydecay}
	Fix $\tau_0\geq 0$ and let $\tilde{p}\in(0,2]$. Consider initial data for $\Box_{g_M}\phi=F$ along $\Sigma_{\tau_0}$ such that:
	\begin{equation}\label{eq:inhom:prop:energyass}
		E_N[\phi](\tau_0)+E_T[T\phi](\tau_0)+\int_{\C_{\tau_0}} r^{\tilde{p}} (\pv\psi)^2+r^2 (\pv T\psi)^2\dd \mu\lesV[\Vc]1 
	\end{equation}

	Assume that for some $\eta>0$, the inhomogeneity $F$ is given by $F=F_1+F_2$, where, for all $n\in\N$:\footnote{In practice, the $F_2$ we consider will always satisfy the below estimate with (almost) one extra power in $\tau$!}
	\begin{equation}\label{eq:inhom:assumptiononF}
	rF_1\lesV \frac{1}{r^2\tau^{\eta}}\min_{q\in[-1,1]}(\tau^q/r^q)\quad\text{and}\quad \fbulkp{\tau_1}{\tau_2}{p}[T^n F_2]+ \int_{\C_{\tau_1}} r^{2+p}T^n(rF)^2\dd \mu \lesV[\Vc] \tau_1^{p-2-2n}  \text{ for all } \tau_1<\tau_2.
	\end{equation}
	Then we have for any $n\in\N$:
 	\begin{equation}\label{eq:inhom:energydecay}
		E_N[T^n\phi](\tau)\lesV[\Vc] \tau^{-\min(\tilde{p},2\eta+1-)-2n}.
	\end{equation}
	Furthermore, if $n=0$ and $p\in[0,\tilde{p}]$, or if $n>0$ and $p\in[0,2]$, we also have
		\begin{equation}\label{eq:inhom:energydecay:pweighted}
	\int_{\C_{\tau}}r^p (\pv T^n\psi)^2 \dd \mu\lesV[\Vc] \tau^{-\min(\tilde{p},2\eta+1-)+p-2n}.
	\end{equation}
\end{prop} 
\begin{rem}\label{rem:inhom:rp:generalised}
	The assumption on $F_1$ in \cref{eq:inhom:assumptiononF} has been stated to tie in directly with the form of the inhomogeneity encountered in \cref{eq:time:aprioriF}. However, in the proof, we only need the range $q\in[-1,0+]$ (see \cref{eq:inhom:lalala}). Furthermore, if we only have the range $q\in[-1,0]$, i.e.~if $r^3F$ attains a nonzero limit towards $\scrip$, the proof still applies with an arbitrarily small loss (because we can only apply \cref{eq:inhom:rp} estimates with $p<2$). More generally, if $rF_1\lesV r^{-2}\tau^{-\eta}\min_{q\in[-1,-1+\tilde{p}/2+]}$, then the result still holds. 
	The nontrivial part in proving this is to show that each time-derivative still gains one full power in decay (as opposed to $\tilde{p}/2$ powers). See \cite{gajic_late-time_2023}[Section 8] for how this works.
	
	Notice, in any case, that if we have more $\tau$-decay at our disposal, then we require less $r$-decay on the inhomogeneity in view of \cref{rem:inhom:rp:u}, and the above comments are then not relevant.
	\end{rem}
\begin{proof}
We first consider the case $\tilde{p}=2$, where we give a brief sketch following \cite{dafermos_new_2010}. Let $\tau_k$ be a dyadically growing sequence of times, say $\tau_k=\tau_02^k$.	Note that the first inequality of \cref{eq:inhom:assumptiononF} implies that, for $\epsilon \ll 1$,
	\begin{equation}\label{eq:inhom:lalala}
	\fbulk{\tau_k}{\tau_{k+1}}[T^n F_1]\lesssim \tau_k^{-2\eta-1-2n+}, \qquad \fbulkp{\tau_k}{\tau_{k+1}}{p}[T^n F_1]=\int_{\tau_k}^{\tau_{k+1}} \int_{\C_{\tau}}r^{\max(p,\epsilon)+1} (rF)^2\lesssim \tau_k^{-2\eta-1-2n+p+}.
\end{equation}
Thus, if we apply \cref{eq:inhom:rp} with $p=2$, we get
	\begin{equation}
		\int_{\tau_k}^{\tau_{k+1}}\dd \tau \int_{\C_{\tau}} r (\pv\psi)^2\dd \mu \lesssim \int_{\C_{\tau_k}} r^2(\pv\psi)^2\dd \mu +O(\tau_k^{\gamma}),
	\end{equation}
	where $\gamma=\max(-2\eta+1+,0)$. 
By the mean value theorem, we can find another dyadic sequence of times $\tau_k'\in[\tau_k,\tau_{k+1}]$ such that
	\begin{equation}
		\int_{\C_{\tau'_{k}}}r(\pv\psi)^2\dd \mu \lesssim \frac{1}{\tau'_k}\int_{\C_{\tau_{k}}}r^2(\pv\psi)^2\dd \mu+O({\tau'_k}^{\gamma-1}).
	\end{equation}
	Now, we apply \cref{eq:inhom:rp} with $p=1$ to obtain 
	\begin{equation}
		\int_{\tau'_k}^{\tau'_{k+1}} E_N[\phi](\tau) \dd \tau  \lesssim \frac{1}{\tau'_k}\int_{\C_{\tau_{k}}}r^2(\pv\psi)^2\dd \mu + E_N[\Vk^1 \phi](\tau'_k) +O({\tau'_k}^{\gamma-1}).
	\end{equation}
	Applying the MVT and energy boundedness (\cref{prop:inhom:energy}), we may deduce $E_N[\phi](\tau)\lesssim \tau^{\gamma-1}$ for all $\tau$. Plugging this back into the above and repeating the argument gives (after an additional commutation with $\Vk$) (cf.~\cite{dafermos_new_2010}[Section 4])\footnote{Alternatively, one may first apply the $p=1$ estimate to deduce $\tau^{\gamma-1}$-decay for the energy, and then apply the $p=2$ estimate to improve to $\tau^{\gamma-2}$. Note moreover that we have not optimized the loss of derivatives on the right-hand side of \eqref{eq:Entau2decay} as this is not relevant in the present paper.}
	\begin{equation}
	\label{eq:Entau2decay}
		E_N[\phi](\tau)\lesssim \tau^{\gamma-2}\big( E_N[\Vk^2\phi]+\int_{\C_{\tau_0}}r^2 (\pv(\Vk^2\psi))^2\dd \mu.\big)
	\end{equation}
	
		In order to prove faster decay for time derivatives, we simply observe that the $T$-commuted $r^p$-flux for $p=2$ can be bounded by the $(r\pv)$-commuted flux with $p=0$, another flux involving angular derivatives and a term featuring $F$; cf.~\cref{lem:inhom:Tconversion}. Since all the latter terms already decay like $\tau^{\gamma-2}$, we therefore gain two extra powers of decay for each time derivative; this proves \cref{eq:inhom:energydecay}.
	
	We use an interpolation argument to infer \cref{eq:inhom:energydecay:pweighted}: More precisely, we split
	\begin{equation}
\int r^p (\pv T^n \psi)^2\dd \mu  =\tau^{p-\lceil p \rceil}\int_{\{r\geq \tau\}} r^{\lceil p \rceil} (\pv T^n \psi)^2\dd \mu+\tau^{p-\lfloor p \rfloor}\int_{\{r\leq \tau\}} r^{\lfloor p \rfloor}(\pv T^n \psi)^2 \dd \mu
	\end{equation}	
	and then apply our existing estimates for integer $p$.
	
	It is left to extend {to} the case $\tilde{p}<2$. Assume first that $\tilde{p}>1$. We then have an estimate for  the $r^{\tilde{p}}$ and $r^{\tilde{p}-1}$ flux, from which we can get an estimate for $p=1$ by interpolation. Once we act with time-derivatives to prove more decay, we use that the time-derivatives have finite $r^2$ flux.
	If $\tilde{p}<1$, then we directly  appeal to \cref{lem:inhom:Tconversion} to go from the $r^{\tilde{p}}$-flux to the $T$-commuted $r^{1+\tilde{p}}$-flux, and only then interpolate.
\end{proof}

\subsection{Pointwise decay estimates}

In this section, we use the hierarchy of elliptic estimates from \cref{prop:inhom:elliptic1} to convert energy decay into pointwise decay. 
\begin{cor}\label{cor:inhom:decay}
	Under the assumptions of \cref{prop:inhom:energydecay}, assuming in addition that $\tilde{p}>1$ as well as\footnote{In our applications, $F_2$ will typically decay almost half a power faster!} $F_2\lesV r^{-2} \tau^{-3/2}$, we have
	\begin{equation}\label{eq:inhom:cor:all:l}
		\psi \lesV (\tau^{-\frac{\tilde{p}-1}{2}}+\tau^{-\eta+}) \min_{q\in[0,1]}\(\frac{r}{\tau}\)^q.
	\end{equation}
If we assume, in addition, that for  $0\leq\ell_0\in\N$ there is some $ \eta_{\ell_0}\in\R$ (not necessarily positive) such that
	\begin{equation}\label{eq:inhom:cor:high:l}
		F_{1,\ell\geq \ell_0}\lesV\frac{r^{\ell_0}}{\tau^{\ell_0}} \frac{1}{r^2\tau^{\eta_{\ell_0}+1}}\quad F_{2,\ell\geq\ell_0}\lesV\frac{r^{\ell_0}}{\tau^{\ell_0}} \frac{1}{r^2\tau^{3/2}}, \quad \text{then} \quad r^{-\ell_0}\phi_{\ell\geq \ell_0} \lesV \tau^{-\ell_0-1-\min(\frac{\tilde{p}-1}{2},\eta+)}+\tau^{-\ell_0-\eta_{\ell_0}-1}.
	\end{equation}
\end{cor}
\begin{obs}
	The requirement $\tilde{p}>1$ is because $\psi$ does not attain a finite limit at $\scrip$ if only the $r^p$-flux with $p\leq 1$ is finite. However, if we work with $T\psi$ rather than $\psi$, the assumption that $\tilde{p}>1$ is unnecessary because $T\psi$ has a finite $r^p$-flux with $p=2$ by \cref{eq:inhom:prop:energyass}.
\end{obs}
\begin{proof}
	We first prove the large-$r$ estimate, i.e.~\cref{eq:inhom:cor:all:l} with $q=0$.
	Fix $r_1\geq2M$ and $\chi$ be a cut-off localising to $r\geq 2r_1$. 
	Then we use the fundamental theorem of calculus, Cauchy--Schwarz and Hardy's inequality to write
	\begin{equation}
		\int_{\S_{\tau}^{r_1}} \psi^2  =\int_{\Sigma_\tau^{[r_1,\infty)}} \partial_r|_{\Sigma}(\chi \psi^2)
		 \lesssim \int_{\Sigma_\tau^{[r_1,\infty)}} \chi' \psi^2 +\sqrt{\int_{\Sigma_\tau^{[r_1,\infty)}}r^{\tilde{p}}(\partial_r|_{\Sigma} \psi)^2 \cdot \int_{\Sigma_\tau^{[r_1,\infty)}} r^{2-\tilde{p}} (\partial_r|_{\Sigma} \psi)^2}\lesssim \tau^{-\tilde{p}+1}+\tau^{-2\eta+}.
	\end{equation} 
	
	Next, we prove \cref{eq:inhom:cor:high:l} with an $\epsilon$-loss. 
	For this, we apply the elliptic estimate of \cref{prop:inhom:elliptic1} $\ell_0+1$ times, for $\epsilon>0$ arbitrarily small ({the requirement $\epsilon>0$} is necessary since  \cref{prop:inhom:elliptic1} requires $q<2\ell_0+1$):
	\begin{nalign}
	&	\int_{\S_{\tau}^{r_1}} r^{-2\ell_0+\epsilon}\phi_{\ell\geq \ell_0}^2 
		\lesssim \int_{\Sigma_\tau^{[r_1,\infty)}} r^{-2\ell_0-1+\epsilon} \phi_{\ell\geq \ell_0}^2+r^{-2\ell+1+\epsilon} (\partial_r|_{\Sigma}\phi_{\ell\geq \ell_0})^2\\
	&	\lesssim  \int_{\Sigma_\tau^{[r_1,\infty)}}  r^{-2\ell_0+1+\epsilon} \left((\partial_r|_{\Sigma} (rT\phi_{\ell\geq\ell_0}))^2+(\sl\partial_r|_{\Sigma} (r T\phi_{\ell\geq\ell_0}))^2+((\sl+1)rF_{\ell\geq \ell_0})^2\right)\\
	&	\lesssim \sum_{n=0}^{\ell_0+1} \int_{\Sigma_\tau^{[r_1,\infty)}}  r^{1+\epsilon} \left((\sl^n\partial_r|_{\Sigma} (rT^{\ell+1}\phi_{\ell\geq\ell_0}))^2\right)+ r^{2(n-\ell_0)+\epsilon+1}  (\sum_{k=0}^n\sl^k T^n rF_{\ell\geq \ell_0})^2
	\\
	&\lesssim \tau^{-2\ell_0-2+1+\epsilon-\min(\tilde{p}, 2\eta+1-)}+\tau^{-2\ell_0-2-2\eta_{\ell_0}+\epsilon}.
	\end{nalign}
	
	To finally remove this $\epsilon$-loss, we may restrict to $\phi$ supported on $\ell=\ell_0$ (if $\phi$ is supported on $\ell\geq\ell_0+1$, we can simply redo the computations above with $\epsilon=0$). 
	Then, we notice that with \cref{eq:inhom:elliptic:suitablyadd1}, we obtain control over 
	\begin{equation}
		\int_{\Sigma_{\tau}}r^{-q+2}w_{\ell}^2 (\partial_r|_{\Sigma} \philc)^2 \dd \mu_D, 
	\end{equation}
	without requiring an upper bound for $q$, cf.~\cref{footnote:hardy}.
	Thus, we may estimate, for $\ell=\ell_0$ (dropping the $F$-terms):
	\begin{nalign}
\int_{\S_{\tau}} \philc^2 \lesssim \sqrt{\int_{\Sigma_\tau} \philc^2 \cdot \int_{\Sigma_{\tau}} (\partial_r|_{\Sigma} \philc)^2}
&\lesssim \sqrt{\int_{\Sigma_{\tau}} r^{-2\ell+2}(\partial_r|_{\Sigma} T \psi)^2 \int_{\Sigma_{\tau}} r^{-2\ell}(\partial_r|_{\Sigma} T \psi)^2}\\
&\lesssim \sqrt{\int_{\Sigma_{\tau}} r^{-2\ell+2\ell}(\partial_r|_{\Sigma} T^{\ell} \psi)^2 \int_{\Sigma_{\tau}} r^{-2\ell+2\ell}(\partial_r|_{\Sigma} T^{\ell+1} \psi)^2},
	\end{nalign}
	which then recovers \cref{eq:inhom:cor:high:l} without $\epsilon$-loss, as well as \cref{eq:inhom:cor:all:l} with $q=-1$.
%
%
\end{proof}

\section{Construction of approximate solutions}\label{sec:approx}
\newcommand{\pubar}{\partial_{\bar{u}}}
\newcommand{\pvbar}{\partial_{\bar{v}}}

In this section, we construct a class of approximate solutions to $\Box_{g_M}\phi=0$ (these will also be approximate solutions to $\Box_g\phi=0$).
We keep the construction here specific, but we will explain at the end of the section how to generalise it.

First of all, we want to make reference to solutions to the Minkowskian wave equation. We do this with respect to $u,r$-coordinates; that is we define, for $r\geq \rc$: 
\begin{equation}
r^2	\Box_\eta \phi: =\Xx(r^2\Xx \phi)-r\Xx(rT\phi)+\Dl \phi.
\end{equation}
We now define, for $\ell\in\N$, $m=\{-\ell,\dots,\ell\}$,  $\psi^{\eta}_{\ell,m}=r\phi_{\ell,m}^\eta = (r/(\tilde{\tau}(\tilde{\tau}+r)) )^{\ell+1} $, and note the following:
\begin{lemma}
	In the region where $\tilde{\tau}=u$, we have $\Box_\eta(\phi^{\eta}_{\ell,m}\cdot Y_{\ell,m})=0$.
\end{lemma}
\begin{proof}
	We define the coordinates $\bar{u}=u$ and $\bar{v}=r+u$. Then $\Box_\eta\phi=0$ is equivalent to  $\pubar\pvbar\psi=r^{-2}\Dl\psi$.\footnote{Note that this is a different equation from the one we would get by taking $M=0$ in $(u,v)$-coordinates!}
	 Define now $f=1/\bar{u}-1/{\bar{v}}=r/(u(u+r)$. 
	Since $\pubar\pvbar f^{\ell+1}=f^{\ell-1}\ell(\ell+1)\pvbar f \pubar f+f^{\ell}(\ell+1)\pubar\pvbar f$, the result follows from
	\begin{equation}
		\pvbar f \pubar f=-\bar{u}^{-2}\bar{v}^{-2}=-r^{-2}f^2 \quad \text{and}\quad \pubar\pvbar f=0.\qedhere
	\end{equation}
\end{proof}
We now define out of these Minkowskian solutions approximate solutions for $\Box_{g_M}\phi=0$ by matching to the stationary solutions $w_\ell$ in the finite $r$-region:
\begin{prop}\label{prop:approx}
	Let $\chi(\tfrac{u}{r})$ be a cutoff identically one for $u/r\geq 3/4$ and zero for $u/r\leq 1/4$.
	We define
	\begin{equation}
	\appsi=r\apphi,\qquad	\apphi:=  \phi^{\eta}_{\ell,m}(\chi+(1-\chi)\frac{w_{\ell}}{r^{\ell}})=\frac{\chi r^{\ell}+(1-\chi)w_{\ell}}{\tilde{\tau}^{\ell+1}(\tilde{\tau}+r)^{\ell+1}}\lesV \frac{r^{\ell}}{\tau^{\ell+1}(\tau+r)^{\ell+1}}.
	\end{equation}
	Then we have
	\begin{equation}
		\Box_{g_M}(\apphi \cdot Y_{\ell,m})\lesV \frac{1}{(\tau+r)r^2 } \frac{r^{\ell}}{\tau^{\ell+1}(\tau+r)^{\ell+1}}.
	\end{equation}
\end{prop}
\begin{rem}
	In general, the weights of the wave operator $\Box_{g_M}$ are $1/(r \min(u,r))$. Thus, the approximate solutions gain two weights near null infinity, and one weight towards future timelike infinity compared to a generic function.
\end{rem}
\begin{proof}
The proof is a direct computation. Let us however also give a more conceptual proof by working in three different spacetime regions. Recall that $\tilde{\tau}=u$ for $r\geq \rc+1$ and $\tilde{\tau}=v$ for $r\leq \rc-1$. Throughout the proof, we will simply write $\apphi\cdot Y_{\ell,m}=\phi$.

\textit{Near $\scrip$: $r\geq 4|u|/5$.}
	\begin{nalign}
		r^2\Box_{g_{M}}\phi=\Xx(r^2D\Xx\phi)-rT\Xx(r\phi)+\Dl\phi
	=r^2\Box_{\eta}\phi+2M(-\Xx^2(r\phi)+\Xx\phi)=r^2\Box_\eta\phi +O(r^{-1})|\V^2\phi|
\end{nalign}
We thus see the gain in three weights, as desired.

\textit{Intermediate region: $|u|/5\leq r \leq 4|u|/5$.}
Observe the following:
	If we consider $\phi\cdot f$ instead of $\phi$, with $f$ a function satisfying $f=1+O_{\V}(r^{-1})$, then we no longer gain three powers in $r$: In the term $rT\Xx (r\phi\cdot f)$, we will get
\begin{equation}\label{eq:approx:cutoffrelevance}
	r \Xx f\cdot T(r\phi)= O(r^{-1}) T(r\phi) =O(\tau^{-1})|\V^1\phi|.
\end{equation}
However, restricted to the region $r\leq 4|u|/5$, we still gain $r^{-3}$ decay if we take $f=r^{-\ell}(\chi r^{\ell}+(1-\chi)w_{\ell})$.

\textit{Finite $r$-region: $2M\leq r\leq |u|/5$.}
Notice that in this region, $\philc=\tilde{\tau}^{-\ell-1}(\tilde{\tau}+r)^{-\ell-1}=r^{-\ell}\phi^{\eta}_\ell$. 
Then, in the subregion $r\geq \rc+1$, where $\tilde{\tau}=u$, we use that the wave operator can be written as
	\begin{nalign}
	r^2w_{\ell}\Box_{g_M}\phi=&	\Xx(r^2 Dw_{\ell}^2 \Xx\philc)-w_{\ell}2r\Xx(rT\phi_{\ell})
	=	\Xx(r^2 r^{2\ell} \frac{Dw_{\ell}^2}{r^{2\ell}} \Xx\philc)-w_{\ell}r\Xx(rT\frac{w_{\ell}}{r^{\ell}} r^{\ell}\philc)\\
	=& \frac{D w_{\ell}^2}{r^{2\ell}}\Xx(r^2 r^{2\ell}\Xx\philc)-\frac{w_{\ell}^2}{r^{\ell}}r\Xx(rT\philc r^{\ell})+\Xx(D\frac{w_{\ell}^2}{r^{2\ell}})r^{2\ell+2}\Xx \philc -w_{\ell}r\Xx(w_\ell/r^{\ell})T(r^{\ell+1}\philc)\\
	=r^{\ell+2}\Box_\eta \phi &+\left(\frac{Dw_\ell^2}{r^{2\ell}}-1\right) \Xx(r^2 r^{2\ell}\Xx \philc) -\left(\frac{w_\ell^2}{r^{\ell}}-r^{\ell}\right)+\Xx(D\frac{w_{\ell}^2}{r^{2\ell}})r^{2\ell+2}\Xx \philc -w_{\ell}r\Xx(w_\ell/r^{\ell})T(r^{\ell+1}\philc)
\end{nalign}
The result then follows from the fact that $\Xx\philc$ gains an extra power in $\tau$; that is $\tau\Xx\philc \lesV \tau^{-2\ell-3}$.
Lastly, in the region $r\leq \rc+1$, where $\tilde{\tau}$ interpolates from $u$ to $v$, we again use that both $\Xb$ (or $\Xx$, respectively) and $T$ produce an extra $\tau$-weight when acting on $\philc$.
 \end{proof}

\begin{rem}[Comments on generalisations]\label{rem:approx}
	\begin{itemize}
		\item The cut-off is $\chi$ is present to improve the decay in $r$ as we approach $\scrip$, as is clear from \cref{eq:approx:cutoffrelevance}. However, for all our asymptotic statements of \cref{thm:main}, we could also omit $\chi$ and simply work with $r^{-\ell}w_\ell \phi^{\eta}_{\ell,m}$ instead.
		\item We only used that $\phi^{\eta}_{\ell,m}$ solves $\Box_\eta(\phi^{\eta}_{\ell,m}Y_{\ell,m})=0$ with respect to $u,r$-coordinates, rather than using the precise functional form of $\phi^\eta_{\ell,m}$. In particular, we can let $\phi^\eta$ be a Minkowskian solution arising from more complicated initial data! This is relevant for deducing higher-order asymptotics, see \cref{rem:time:secondorder}.
		\item Entirely analogous arguments apply to $\Box_{g_M}\phi=V(r)\phi$ with $V(r)$ a potential satisfying the asymptotics $V(r)\sim \alpha r^{-2}$, with $\alpha>-\frac{1}{4}$ (and assuming that smooth $w_{\ell}$ do not have finite energies, i.e.\ the absence of \emph{zero-energy resonances}). See \cite{gajic_late-time_2023}[Lemma 10.1].
	
	\end{itemize}
\end{rem}

Up to multiplication by a constant factor, \cref{prop:approx} provides a single approximate solution for each $\ell$. We can easily construct further approximate solutions by differentiating or integrating in $T$.
For later reference, the following observation will be important:
\begin{lemma}\label{lem:approx}
	Let $\apphi$ be as in \cref{prop:approx}. Then we have, for large $r\gg \rc$:
	\begin{equation}\label{eq:approx:time}
		rT^{-1}\apphi(\tau_0,r):=-r\int_{\tau_0}^{\infty} \apphi(\tau',r)\dd \tau'=\sum_{i=0}^{\ell-1}\frac{c^{(\ell)}_i}{r^i}+(-1)^{\ell} \binom{2\ell}{\ell}\frac{\log r}{r^{\ell}}+O(r^{-\ell})
	\end{equation}
	for some coefficients $c^{(\ell)}_i$ whose precise values do not matter.
\end{lemma}
\begin{proof}
	First of all, we can replace $\apphi$ simply by $\phi^{\eta}_{\ell,m}$, as the difference between the two is subleading in $r$ and therefore will contribute a logarithm at later order.
	We then note that $(r^2\Xx)^{\ell+1}\frac{r^{\ell+1}}{\tau^{\ell+1}(\tau+r)^{\ell+1}}=\frac{(2\ell+1)!}{\ell!}\left(\frac{r}{\tau+r}\right)^{2\ell+2}$. We may now take the $\int_{\tau_0}^{\infty}$-integral; this gives 
	\begin{equation}
		\int_{\tau_0}^{\infty}(r^2\Xx)^{\ell+1}\frac{r^{\ell+1}}{\tau^{\ell+1}(\tau+r)^{\ell+1}}\dd\tau= \frac{1}{2\ell+1}\frac{(2\ell+1)!}{\ell!}r+\dots, 
	\end{equation}
	from which the result follows after $\ell+1$ integrations in $1/r$, giving $\frac{(-1)^{\ell+1}}{\ell!}   (1/r)^{\ell} \log r \frac{(2\ell)!}{\ell!}$.
\end{proof}

\section{Proof of \cref{thm:main}}\label{sec:time}

In this section, we prove the main result of the paper, namely \cref{thm:main}. In particular, we will assume throughout this entire section that \cref{ass:main} holds. Furthermore, for ease of notation, we will assume in addition that $\beta\geq 2+\delta$. To remove this assumption, the reader simply has to replace every appearance of $2+\delta$ with $\min(\beta,2+\delta)$ below. (For the readers convenience, we keep track of $\beta$ in \cref{sec:time:first,sec:time:subtract,sec:time:second,sec:time:inbetween}.)
We begin with some preliminaries in \cref{sec:time:prelim}, and give an overview of the proof in \cref{sec:time:overview}.
\subsection{Preliminaries}\label{sec:time:prelim}
We shall prove \cref{thm:main} by iteratively improving the following assumption:
\begin{iterate}\label{ass:time}
Assume that $g$ and $\phi$ are as in \cref{ass:main}, and assume that for some $\alpha\in\R$:
	\begin{equation}\label{eq:time:apriori}\tag{$\mathfrak{I}1(\alpha)$}
		\phi\lesV\tau^{-3/2-\alpha} \min_{q\in[0,1]}\(\frac{\tau}{r}\)^q \quad \(\iff r\phi\lesV \tau^{-1/2-\alpha}\min_{q\in[0,1]} \left(\frac{r}{\tau}\right)^q\).
	\end{equation}
\end{iterate}
By \cref{ass:main}, this holds with $\alpha=-1/2$. We now make a few preliminary remarks on this assumption:
\begin{rem}
Disregarding $\tau$-decay,	\cref{eq:time:apriori} implies that $\pu\phi \lesssim D$ towards $\hplus$. We therefore have that $G^{v\beta}\partial_\beta\phi\lesssim D$ towards $\hplus$ (since $G^{vA}$ vanishes identically near $\hplus$.)
\end{rem}
\begin{rem}
	In view of \cref{eq:dyn:Guv}--\cref{eq:dyn:GAB}, \cref{eq:time:apriori} implies the following bound for the inhomogeneity in $\Box_{g_M}\phi=F$:
	\begin{equation}\label{eq:time:aprioriF}
		-\detgm F=	\partial_\mu (G^{\mu\nu}\partial_\nu \phi)\lesssim r^{-1}\tau^{-3/2-\alpha-\delta}\min_{q\in[-1,1]}\left(\frac{\tau}{r}\right)^q.
	\end{equation}
	We note that if we {did not} assume the improved $\tau$-decay for $g$ away from $\scrip$ (cf.~\cref{eq:dyn:improved1} and \cref{eq:dyn:pol:Omega}--\cref{eq:dyn:pol:trx} with $n=-1$), we would only be able to take $q\in[0,1]$ in \cref{eq:time:aprioriF}.
		Moreover, in addition to $\V$-regularity, \cref{eq:time:apriori} implies, for all $k\in\N$:
	\begin{equation}\label{eq:time:apriori:r2X}
		(r^2\Xx)^{k}\psi \lesV_k \tau^{-1/2-\alpha+k+}\min_{q\in[0,1]}\left(\frac{r}{\tau}\right)^q, \qquad 	(r^2\Xx)^{k}\phi \lesV_k \tau^{-3/2-\alpha+k+}\min_{q\in[0,1]}\left(\frac{r}{\tau}\right)^q;
	\end{equation}
	as follows directly from integrating in $u$ the $(r^2\Xx)$-commuted wave equation \cref{eq:inhom:r2Xcommuted}. 
\end{rem}
Notice that we can directly upgrade \cref{eq:time:apriori}:
\begin{lemma}\label{lemma:time:zeroth}
	Under \cref{ass:main}, \cref{eq:time:apriori} holds with $\alpha=0$.
\end{lemma}
\begin{proof}
	Since $\phi$ is compactly supported initially, and in view of \cref{eq:time:aprioriF}, we may apply \cref{prop:inhom:main} with $\eta=\delta$ and $\tilde{p}=2$, yielding that  \cref{eq:time:apriori} holds with $\alpha=-1/2+\delta-$. 
	We may iteratively apply \cref{prop:inhom:main} with an improved $\eta$ to eventually recover $\alpha=0$. 
\end{proof}

\paragraph{Preliminaries on time integrals} To increase the $\alpha$ in \cref{eq:time:apriori}, time integrals will play a crucial role. Here we provide some preliminaries on time integrals. Given a sufficiently fast decaying function $\phi$, we define the time integral via:
\begin{equation}\label{eq:time:inverse}
	T^{-1}\phi(u,v,\theta^A) =-\int_u^{\infty} \phi(u',v+u'-u,\theta^A) \dd u'=-\int_{v}^\infty \phi(u+v'-v,v')\dd v'=-\int_{\tau}^{\infty} \phi(\tau',r,\theta^A)\dd \tau'.
\end{equation}
For $n\geq 1$, we inductively define $T^{-n}\phi:=T^{-1}T^{-n+1}\phi$ whenever this is well-defined, i.e.~whenever $\phi\lesssim \tau^{-n-}$. Under this assumption, we then also have, via integration by parts: 
\begin{equation}
	T^{-n}\phi(\tau,r,\theta^A)=\frac{-1}{(n-1)!}\int_{\tau}^{\infty}(\tau-\tau')^{n-1} \phi(\tau',r,\theta^A)\dd \tau'.
\end{equation}

If $\Box_{g_M}\phi=F$, then we use that $T$ commutes with the Schwarzschild wave operator $\Box_{g_M}$ to obtain:
\begin{equation}\label{eq:time:inverseF}
	\Box_{g_M}T^{-1}\phi =-\int_{u}^{\infty}  F(u',v+u'-u)\dd u'=-\int_{v}^{\infty} F(u+v'-v,v')\dd v'.
\end{equation}

\subsection{Overview of the proof of \cref{thm:main}}\label{sec:time:overview}
As outlined in \cref{sec:sketchproof}, we will prove \cref{thm:main} by writing \cref{eq:main:wave} as an inhomogeneous wave equation on Schwarzschild. 
The general approach is as follows:
Our initial assumption (in this case  \cref{eq:time:apriori} with $\alpha=0$ in view of \cref{lemma:time:zeroth}) on $\phi$ implies a decay estimate for the inhomogeneity as in \cref{eq:time:aprioriF}. 
We then define the first time integral of $\phi$, and prove that along the initial hypersurface $\Sigma_{\tau_0}$, $T^{-1}\phi$ satisfies the assumptions of \cref{prop:inhom:main}, allowing us to deduce decay for $T^{-1}\phi$, and thus inferring improved decay for $\phi=TT^{-1}\phi$. 

Let us now be more precise. Purely for simplicity, we will first make the additional assumption that $\delta>1/2$ in \cref{ass:main}. 
We then explicitly compute the asymptotic behaviour of the initial data of the \textbf{first time integral} for fixed $\ell$-modes to infer that $r^2\Xx(T^{-1}\psi_{\ell=0})$ and $\frac{r^2}{\log r}\Xx(T^{-1}\psi_{\ell=1})$ attain finite limits, and that $r^2\Xx(rT^{-1}\psi_\ell)\lesssim 1$ for $\ell>1$.
 Together with the elliptic estimates in \cref{prop:inhom:elliptic2}, which we may apply to $\phi_{\ell\geq\ell_0}$ for a suitably large $\ell_0$, this suffices to conclude the finiteness of the $r^p$-flux of $T^{-1}\psi$ with $p=2$. 
We may then appeal to \cref{cor:inhom:decay} to show that $T^{-1}\phi$ satisfies \cref{eq:time:apriori} with $\alpha=0$, which also implies that $\phi$ satisfies \cref{eq:time:apriori} with $\alpha=1$.

In order to increase the $\alpha$ in  \cref{eq:time:apriori} further, we wish to take the \textbf{second time integral}. It will turn out, however, that for $\ell=0,1$, $r\Xx(T^{-2}\psi_\ell)$ attains a finite limit, while $\frac{r^2}{\log r}\Xx(T^{-2}\psi_{\ell})$ for $\ell\geq2$ also attains a finite limit. 
Thus, the $\ell=0$ and $\ell=1$-modes only have a finite $r^p$-flux with $p<1$, which by itself would only suffice to derive \cref{eq:time:apriori} with $\alpha=3/2-$. 

We can, however, subtract an approximate solution $\appsi[\ell\leq 1]$ (which itself also satisfies \cref{eq:time:apriori} with $\alpha=3/2$) such that the limits of $r^2\Xx(T^{-1}\widehat{\psi}^{[1]}_{\ell=0})$ and $\frac{r^2}{\log r}\Xx(T^{-1}\widehat{\psi}^{[1]}_{\ell=1})$ vanish, where $\widehat{\psi}^{[1]}=\psi-\appsi[\ell\leq1]$.
Then, the corresponding time integral of the difference will admit finite $r^p$-fluxes with $p=2$, and we can show that \cref{eq:time:apriori} holds with $\alpha=2$ for $\widehat{\psi}^{[1]}$, which, in fact, already proves the leading-order asymptotics of $\psi$ in \cref{thm:main}.

After the above procedure, we moreover obtain the validity of the following assumption with $\mathfrak{L}=2$:
\begin{iterate}\label{ass:time:n32}
	Let $\phi$ satisfy \cref{eq:time:apriori} with $\alpha=3/2$. Assume, in addition, that for $\mathfrak{L}\in\R$ satisfying $\mathfrak{L}<2+\delta$, \cref{eq:time:apriori} holds with $\alpha= \mathfrak{L} $ when $\phi$ is replaced by $\phi_{\ell\geq \lfloor \mathfrak{L}\rfloor}$, i.e.~assume
	\begin{equation}\label{eq:time:ass2}\tag{$\mathfrak{I}2(\mathfrak{L})$}
\phi\lesV \tau^{-3} \min_{q\in[0,1]}\left(\frac{\tau}{r}\right)^q\quad \text{and}	\quad	\phi_{\ell\geq \lfloor \mathfrak{L} \rfloor} \lesV \tau^{-3/2-\mathfrak{L}}\min_{q\in[0,1]}\left(\frac{\tau}{r}\right)^q
	\end{equation}
\end{iterate}
 \begin{figure}[h!t!]
	\includegraphics[width=0.30\textwidth]{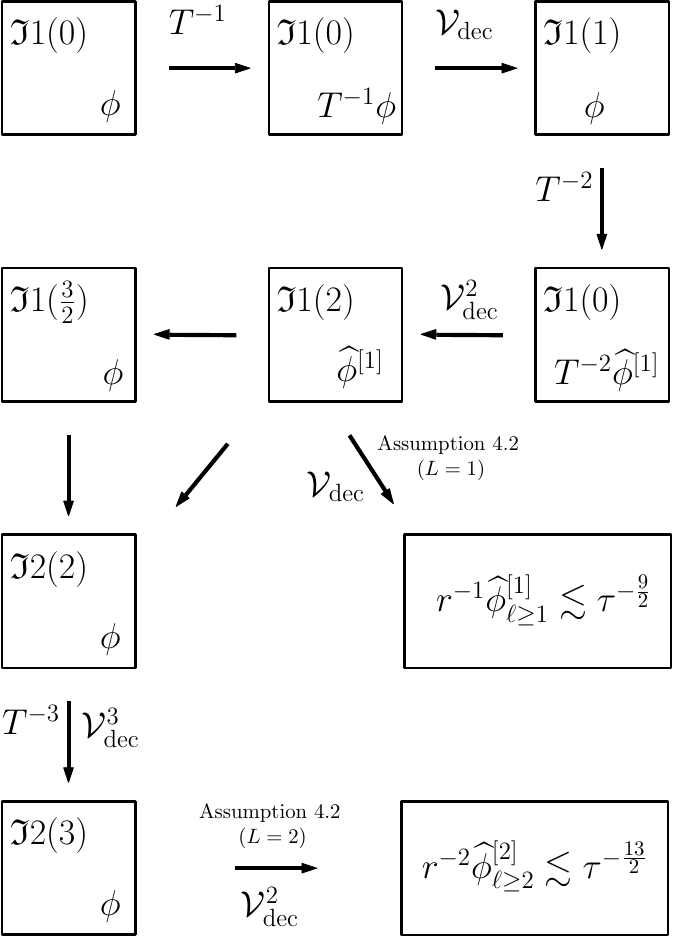}
	\caption{An overview of iteration argument for improving $\alpha$ in $\mathfrak{I}1(\alpha)$ and deriving $\mathfrak{I}2(3)$ (in the case $\delta>1$).}\label{fig:diagram1}
\end{figure}

In order to also establish the leading-order asymptotics for higher $\ell$-modes, we proceed similarly as before, but we now have to iteratively increase the $\mathfrak{L}$ in \cref{eq:time:ass2}, rather than the $\alpha$ in \cref{eq:time:apriori}.
As we will see, the maximum $\mathfrak{L}$ is restricted by the value of $\delta$; see also \cref{rem:intro:modecoupling}. See \cref{fig:diagram1} for a diagram of the logic of the proof.

We shall make the additional assumption that $\delta>1/2$ in  \cref{sec:time:first,sec:time:subtract,sec:time:second}. This assumption will be removed in \cref{sec:time:inbetween}.

\subsection{First time integration}\label{sec:time:first}
In this section, we determine the behaviour of $T^{-1}\phi$ along $\Sigma_{\tau_0}$ and show that \cref{eq:time:apriori} holds with $\alpha=1$.
\subsubsection{Initial data for first time integral, $\ell=0$}
We first consider $T^{-1}\psi_{\ell}$ with $\ell=0$.
\begin{prop}\label{prop:time:l=0:1}
Let $\delta>1/2$.	The $\ell=0$ time integral $T^{-1}\psi_{\ell=0}(\tau_0,r,\theta^A)$ along $\Sigma_{\tau_0}$ satisfies:
	
	a) 	 $r^2\pv( T^{-1}\psi_{\ell=0})(\tau_0,r)\lesV 1$ if  \cref{eq:time:apriori} holds with $\alpha=0$, and we have
		\begin{equation}\label{eq:time:l=0:firststatement}
	\lim_{r\to\infty}r^2\Xx(T^{-1}\psi_{\ell=0})|_{\Sigma_{\tau_0}}=\frac{1}{4\pi}\int_{\scrip}( \mathfrak{m}-M)\psi+\frac{\mathfrak{F}^{(4)}}{2}\dd u \dd \sigma +MC_0=:{I}_0[\phi], \quad \text{where}
	\end{equation}
 \begin{equation}\label{eq:time:C0}
		4\pi C_0=\int_{\S_{\tau_0}^{\rc}} r^2 \phi\dd \sigma+\int_{\C_{\tau_0}} r \Xx \psi -\frac{G^{u\beta}\partial_\beta\phi}{2D\sin \vartheta}\dd\mu_D+\int_{\Cbar_{\tau_0}}-r\Xb\psi-\frac{G^{v\beta}\partial_\beta\phi}{2D\sin \vartheta}\dd\mu_D+\int_{\hplus} \frac{G^{u\beta}\partial_\beta\phi}{2} \dd \tilde{\sigma} \dd v 	-\int_{\D^{2M,\infty}_{\tau_0,\infty}} r^2 \mathfrak{F}\dd \mu_D.
	\end{equation}
	
	b) If  \cref{eq:time:apriori} holds with $\alpha=1$, then we moreover have:
	\begin{equation}\label{eq:time:l=0NPconstant}
	 r^2\pv(T^{-1}\psi_{\ell=0})(\tau_0,r) -\frac{1}{4\pi}\int_{\scrip}  \mathfrak{m} \psi+ \frac{\mathfrak{F^{(4)}}}{2} \dd u \dd \sigma \lesV r^{-\max(1,\delta+1/2-)}.
	\end{equation}

\end{prop}
\begin{rem}
	Notice that \cref{eq:time:l=0NPconstant} provides more detailed asymptotics than \cref{eq:time:l=0:firststatement} and also gives an alternative form for ${I}_0[\phi]$ (which requires an improved $\alpha$ in \cref{eq:time:apriori}).
\end{rem}

\begin{proof}
\textit{Proof of a):} We first give the proof with $\mathfrak{F}=0$. Recall from \cref{eq:dyn:wave:X}, \cref{eq:time:inverse,eq:time:inverseF} that:
\begin{equation}\label{eq:time:l=0:waveeq}
	\Xb(r^2D\Xb T^{-1}\phi)+r\Xb(r\phi)+\Dl T^{-1}\phi=\int_{\tau}^{\infty}\frac{\partial_\alpha\left(G^{\alpha\beta}\partial_\beta\phi\right)}{2D\sin \vartheta}\dd \tau'=	\Xx   (r^2D\Xx   T^{-1}\phi)-r\Xx   (r\phi)+\Dl T^{-1}\phi.
\end{equation}
In order to find the behaviour of $T^{-1}\phi_{\ell=0}$, we integrate over $S^2$ and along $\partial_r|_{\Sigma}$ up until some $r_\infty>\rc$; this gives
\begin{nalign}\label{eq:time:l=0:proof1}
	\int_{\S_{\tau_0}^{r_{\infty}}} r^2 D\Xx T^{-1}\phi \dd \sigma=
	&\int_{\S_{\tau_0}^{\rc}} r^2 \phi\dd \sigma
	+\int_{\C_{\tau_0}^{\rc, r_{\infty}}} r \Xx (r\phi) \dd r \dd \sigma
	-\int_{\Cbar_{\tau_0}}
r\Xb(r\phi)\dd r \dd \sigma
	+\int_{\D_{\tau_0,\infty}^{2M, r_{\infty}}}  \frac{\partial_\alpha\left(G^{\alpha\beta}\partial_\beta\phi\right)}{2D\sin \vartheta}\dd r\dd \tau\dd \sigma\\
	=&\int_{\S_{\tau_0}^{\rc}} r^2 \phi\dd \sigma+\int_{\C_{\tau_0}^{\rc, r_{\infty}}} r \Xx (r\phi) -\frac{G^{u\beta}\partial_\beta\phi}{2D\sin \vartheta}\dd r \dd \sigma+\int_{\Cbar_{\tau_0}}-r\Xb(r\phi)-\frac{G^{v\beta}\partial_\beta\phi}{2D\sin \vartheta} \dd r \dd \sigma\\
	&\qquad\qquad \qquad +\frac12 \int_{\hplus} G^{u\beta}\partial_\beta\phi \dd \tilde{\sigma} \dd v+\frac12\int_{\Gamma_{r_{\infty}}} (G^{v\beta}-G^{u\beta})\partial_\beta\phi \dd \tilde{\sigma} \dd u,
\end{nalign}
where we have used that $G^{v\beta}\partial_\beta\phi\lesssim D$, that the $\partial_A$ derivatives lead to a total angular derivative and therefore do not contribute, that the boundary terms as $\tau\to\infty$ vanish by our a priori decay assumption, and finally, that 
\begin{equation}
		\int_{\S_{\tau_0}^{\rc}}r^2 D \Xx   T^{-1}\phi \dd \sigma =\int_{\S_{\tau_0}^{\rc}} r^2 D\Xb T^{-1}\phi+r^2\phi \dd \sigma
\end{equation}
along with the fact that $D\Xb T^{-1}\phi$ vanishes at the horizon.
We now convert the integral along $\Gamma_{r_{\infty}}$ into an integral along $\scrip$: Observe
\begin{align}\label{eq:time:l=0:proof2}\nonumber
&	\int_{\Gamma_{r_{\infty}}} (G^{v\beta}-G^{u\beta})\partial_\beta\phi \dd \tilde{\sigma} \dd u =	\int_{\Gamma_{r_{\infty}}} G^{uv}(T\phi-2\pv\phi)+(G^{vA}-G^{uA})\partial_A\phi  \dd \tilde{\sigma} \dd u\\
	&=\int_{\S_{\tau_0}^{r_{\infty}}} -G^{uv}\phi \dd \tilde{\sigma}+\int_{\Gamma_{r_{\infty}}}-TG^{uv}\phi-2G^{uv}\pv\phi+(G^{vA}-G^{uA})\partial_A\phi \dd u \dd\tilde{\sigma}\\
	=&2\cdot \frac{4\pi C_1}{r_{\infty}}
	+\frac1{r_{\infty}}\int_{\C_{\tau_0}^{r_{\infty},\infty}} 2\pv (G^{uv}r \phi )\dd v \dd \tilde{\sigma}-\frac1{r_{\infty}}\int_{r_{\infty}}^{\infty} \int_{\Gamma_{r_{\infty}}} \Xx \(-TG^{uv}r\phi -2G^{uv}r\pv\phi+(G^{vA}-G^{uA})\partial_A(r\phi)\) \dd u \dd \tilde{\sigma},\nonumber
\end{align}
where
\begin{equation}
2\cdot 4\pi 	C_1=\int_{\S_{\tau_0}^{\infty}}-G^{uv}r\phi\dd\tilde{\sigma}+\int_{\scrip} -TG^{uv}r\phi+G^{vA}\partial_A(r\phi) \dd u \dd \tilde{\sigma}.
\end{equation}
The last two terms in \cref{eq:time:l=0:proof2} decay like $O(r^{-1/2-\delta+}+r^{-2})$ by  \cref{eq:time:apriori} with $\alpha=0$.\footnote{For instance, when integrating $\Xx (G^{vA}\partial_A(r\phi))\lesssim r^{-1} u^{-\delta-1/2}\min_{q\in[-1,1]}(u/r)^q$, we take $q=\min(1/2-\delta+,1)$ so that the double integral is bounded by $r^{1/2-\delta+}+r^{-1}$.} Further, notice that under \cref{ass:dyn:pol}, $G^{uA}=0$ near $\scrip$.
Thus, combining \cref{eq:time:l=0:proof1} and \cref{eq:time:l=0:proof2}, we have now established that along $\Sigma_{\tau_0}$, for $C_0$ as in \cref{eq:time:C0}, and for $\gamma=\min(2,1/2+\delta-)$:
\begin{nalign}\label{eq:time:l=0:NP}
	r^2D\Xx   T^{-1}\phi_{\ell=0}&=C_0+\frac{C_1}{r}+O(r^{-\gamma})\implies \Xx   T^{-1}\phi_{\ell=0}= \frac{C_0}{r^2}+\frac{C_1+2MC_0}{r^3}+O(r^{-\gamma-2})\\
	\implies -T^{-1}\phi_{\ell=0}& =\frac{C_0}{r}+\frac{C_1+2MC_0}{2r^2}+O(r^{-\gamma-1})\implies r^2\Xx  (rT^{-1}\phi_{\ell=0})=\frac{(C_1+2MC_0)}{2}+O(r^{-\gamma+1}),
\end{nalign}
where we have used in the second line that the limit of $T^{-1}\phi_{\ell=0}$ towards $\scrip$ vanishes. This already proves \cref{eq:time:l=0:firststatement}.

\textit{Proof of b):} In order to also prove \cref{eq:time:l=0NPconstant} under the stronger decay assumption \cref{eq:time:apriori} with $\alpha=1$, we convert the terms in \cref{eq:time:l=0:proof1} (i.e.~the $C_0$-term in \cref{eq:time:C0}) into an integral along $\scrip$ by applying the divergence theorem to the identity:
\begin{equation}\label{eq:time:l=0:rewrite0}
0=	\sqrt{-\det g}\Box_{g}\phi=\sqrt{-\det g_M}\Box_{g_M}\phi+\partial_\alpha(G^{\alpha\beta}\partial_\beta\phi).
\end{equation}

On the one hand, we have
\begin{nalign}
	&\frac12\int_{u_0}^{\infty}\int_{v_0}^{\infty}\int_{S^2} \sqrt{-\det g_M} \Box_{g_M}\phi \dd \vartheta\dd \varphi \dd u \dd v\\
	=&\int_{v_0}^{\infty}\int_{2M}^{\rc}\int_{S^2} \Xb (Dr^2\Xb \phi)+Tr\Xb (r\phi)+\Dl\phi \dd \sigma \dd r \dd v+2\int_{u_0}^{\infty}\int_{{\rc}}^\infty \int_{S^2}\Xx   (Dr^2\Xx   \phi)-Tr\Xx   (r\phi)+\Dl\phi \dd \sigma \dd r \dd u \\
	=&\int_{\Gamma_{\rc}} Dr^2(\Xb-\Xx) \phi \dd u \dd \sigma
	+\int_{\scrip}Dr^2\Xx   \phi\dd u \dd \sigma  -\int_{\Cbar_{\tau_0}} r\Xb(r\phi) \dd r \dd \sigma +\int_{\C_{\tau_0}} r\Xx  ( r\phi) \dd r' \dd \sigma	\\
	&{+\lim_{\tau_\infty\to\infty}\int_{\Cbar_{\tau_{\infty}}} r\Xb (r\phi)  \dd r \dd \sigma-\lim_{\tau_\infty \to \infty}\int_{\C_{\tau_{\infty}}} r\Xx  (r\phi) \dd r \dd \sigma}
\end{nalign}
The first limit vanishes directly by our decay assumptions; for the second limit, we integrate  $\pu\pv(r\phi_{\ell=0})=-2MDr^{-3} r\phi_{\ell=0}+O(r^{-3}\tau^{-3/2-\delta})$ to see that it also vanishes.
Using also that $D\Xb-D\Xx   =-T$ and that $r^2\Xx\phi=-r\phi+O(1/r)$, we further simplify the $\Gamma_{\rc}$ and the $\scrip$ integral as follows:

\begin{equation}\label{eq:time:l=0:rewrite4}
	\int_{\Gamma_{\rc}}-T(r^2\phi)\dd u \dd \sigma=\int_{\S_{\tau_0}^{\rc}}r^2\phi \dd \sigma,
	\qquad \int_{\scrip}Dr^2\Xx   \phi\dd u \dd \sigma=-\int_{u_0, r=\infty}^{\infty} r\phi \dd u \dd \sigma.
\end{equation}

On the other hand, we have
\begin{multline}\label{eq:time:l=0:rewrite5}
\frac12	\int_{u_0}^{\infty}\int_{v_0}^{\infty}\int_{S^2} \partial_\alpha(G^{\alpha\beta}\partial_\beta\phi) \dd \vartheta\dd \varphi \dd u \dd v\\=\frac12
	\int_{\mathcal{H}^+} G^{u\beta}\partial_\beta \phi\dd v \dd\tilde{\sigma}- \frac12\int_{\C_{\tau_0}} G^{u\beta}\partial_\beta \phi\dd v \dd\tilde{\sigma}
	+\frac12\int_{\scrip}G^{v\beta}\partial_\beta\phi \dd u \dd\tilde{\sigma}-\frac12\int_{\Cbar_{\tau_0}}G^{v\beta}\partial_\beta\phi \dd u \dd\tilde{\sigma}.
\end{multline}

Combining \cref{eq:time:l=0:rewrite0}--\cref{eq:time:l=0:rewrite5} allows us to re-write the horizon term as
\begin{multline}\label{eq:time:rewrite:5}
\frac12	\int_{\mathcal{H}^+}G^{u\beta}\partial_\beta \phi\dd u \dd\tilde{\sigma}=\frac12\int_{\C_{\tau_0}} G^{u\beta}\partial_\beta \phi\dd v \dd\tilde{\sigma}
	-\frac12\int_{\scrip}G^{v\beta}\partial_\beta\phi \dd u \dd\tilde{\sigma}+\frac12\int_{\Cbar_{\tau_0}}G^{v\beta}\partial_\beta\phi \dd u \dd\tilde{\sigma}\\
	-\left(-\int_{\Cbar_{\tau_0}} r\Xb (r\phi) \dd r \dd \sigma +\int_{\C_{\tau_0}} r\Xx   r\phi \dd r \dd \sigma -\int_{\scrip} r\phi \dd u \dd \sigma+\int_{\S_{\tau_0}^{\rc}}r^2\phi \dd \sigma\right)
\end{multline}

Combining \cref{eq:time:rewrite:5} with \cref{eq:time:l=0:proof1,eq:time:l=0:proof2}, we finally have
\begin{nalign}
	\int_{\S_{\tau_0}^{r_{\infty}}}r^2 D\Xx   T^{-1}\phi (u_0,r)\dd \sigma\
	&=
\int_{\scrip} r\phi \dd u \dd \sigma	+\frac1{2r_{\infty}}\(\int_{\S_{\tau_0}^{\infty}}-G^{uv}r\phi\dd\tilde{\sigma}+\int_{\scrip} -TG^{uv}r\phi+G^{vA}\partial_A(r\phi) \dd u \dd \tilde{\sigma} \)\\
&	- \int_{\C_{\tau_0}^{r_{\infty},\infty}} r \Xx  (r\phi)-\frac{G^{u\beta}\partial_\beta\phi}{2D\sin \vartheta} \dd \sigma \dd r+O(r^{-2}).
\end{nalign}
Thus, as in \cref{eq:time:l=0:NP} (where we may now take $\gamma=\min(2,3/2+\delta-)$ as we can replace $\delta+0$ by $\delta+1$), we  conclude that
\begin{equation}
4\pi	r^2\Xx  (T^{-1}\psi_0)(u_0,r)=\frac12\int_{\scrip} (\frac{-TG^{uv}}{2\sin \vartheta}+2M)r\phi+\frac12 \frac{G^{vA}\partial_A(r\phi)}{\sin \vartheta} \dd u \dd \sigma+O(r^{-\gamma})
	=\int_{\scrip}  \mathfrak{m} r\phi (u,\theta^A) \dd u \dd\sigma+O(r^{-\gamma}),
\end{equation}
where we have integrated by parts over the sphere and used \cref{lem:dyn} in the final step.

\textit{The case $\mathfrak{F}\neq0$:} If we now also allow for $\mathfrak{F}\neq 0$, then we get in \cref{eq:time:l=0:proof1} an additional term of the form
\begin{equation}
-	\int_{\D^{2M,r_{\infty}}_{\tau_0,\infty}} r^2 F\dd \mu_D=\frac{\int_{\scrip} \mathfrak{F^{(4)}}\dd u \dd \sigma}{r_{\infty}}-	\int_{\D^{2M,\infty}_{\tau_0,\infty}} r^2 \mathfrak{F}\dd \mu_D.
\end{equation}
When applying the divergence theorem to \cref{eq:time:l=0:rewrite0} (with LHS=$\detgm \mathfrak{F}$) {in the proof of \textit{b)}}, the spacetime integral over $\mathfrak{F}$ will then cancel out.
\end{proof}
\begin{lemma}\label{prop:time:rewritingl=0}
	Under the assumption of \cref{prop:time:l=0:1} \textit{a)}, we have (denoting $n_{\Cbar}^{\mu}=g^{\mu \nu} (\dd v)_\nu$, $n_{\C}^{\mu}=g^{\mu \nu} (\dd u)_\nu$)
	\begin{multline}
		4\pi {\integral}_0[\phi]=\int_{\scrip} ( \mathfrak{m}-M)\psi +\frac12\mathfrak{F}^{(4)}\dd u \dd \sigma +\frac12\int_{\hplus} \otrx \phi \detgs \dd\tilde{\sigma}\dd v\\
		+\frac{M}2\int_{\M} r^2\mathfrak{F} \dd \mu_D+\frac{M}2 \int_{\S_{\tau_0}^{2M}}  \phi \detgs \dd\tilde{\sigma} +\frac{M}2\int_{\Cbar_{\tau_0}} n_{\Cbar}^\mu \partial_\mu \phi \detgs \dd\tilde{\sigma} \dd u+\frac{M}2\int_{\C_{\tau_0}} n_{\C}^\mu \partial_\mu \phi \detgs \dd\tilde{\sigma} \dd v.
	\end{multline}
\end{lemma}
\begin{proof}
	We re-write the terms appearing in \cref{eq:time:C0} as follows:
For the integral along $\hplus$, we write
\begin{nalign}
	&\frac12 \int_{\hplus} G^{u\beta}\partial_\beta \phi \dd \tilde{\sigma}\dd u=\frac12 \int_{\hplus} \detgsm \pv\phi -\detgs (\pv\phi +b^A\sl^g_A \phi) \dd \tilde{\sigma}\dd v\\
	=&\frac12 \int_{\hplus} \pv\big((\detgsm -\detgs)\phi\big) +\pv(\detgs) \phi +\detgs \,\div^g (b) \phi) \dd \tilde{\sigma}\dd v\\
	=&\frac12 \int_{\S_{\tau_0}^{2M}} (\detgsm -\detgs)\phi\dd \tilde{\sigma}+\frac12 \int_{\hplus} \otrx \detgs \dd v \dd\tilde{\sigma},
\end{nalign}
where we inserted the expressions for $G^{u\beta}$ from \cref{eq:dyn:Gspelledout} and  used \cref{eq:dyn:DlogG} (with $u$ and $v$ interchanged).

For the integral along $\S_{\tau_0}^{\rc}$ appearing in \cref{eq:time:C0}, we use that $\lim_{r\to\infty}\int_{\S_{\tau_0}^r} r^2\phi \dd \tilde{\sigma}=0$ to write
\begin{equation}
	\int_{\S_{\tau_0}^{\rc}} r^2\phi \dd \sigma = \frac12\Big( \Big(\int_{\S_{\tau_0}^{2M}} r^2\phi \dd \sigma +\int \Xb(r^2\phi) \dd \mu_D \Big) -\int_{\C_{\tau_0}} \Xx(r^2\phi) \dd \mu_D\Big),
\end{equation}

The result then follows by  inserting the two identities above, along with the expressions \cref{eq:dyn:Gspelledout}, into \cref{eq:time:C0}. 
\end{proof}

\subsubsection{Initial data for first time integral, $\ell>0$}
We now consider $T^{-1}\psi_{\ell}$ with $\ell>0$. Recall the definition of $H_{\ell,m}^{(\ell+1)}$ in \cref{eq:main:thm:integralexpression} (see also \cref{eq:time:proof:loggenerator}):
\begin{prop}\label{prop:time:l>0:1}
Let $\delta>1/2$ and let $\ell\geq 1$. Then the time integral $T^{-1}\phi_{\ell}(\tau_0,r,\vartheta,\varphi)$ satisfies along $\Sigma_{\tau_0}$:
	
	a) If  \cref{eq:time:apriori} holds with $\alpha=0$, then $r^2\pv(r T^{-1}\psi_{\ell})|_{\Sigma_{0}}\lesV 1$ for $\ell\geq 2$, while, for $\ell=1$: 
		\begin{equation}\label{eq:time:prop:first:l=1}
	T^{-1}\psi_{\ell=1,m}(\tau_0,r)+\frac{\int_{\scrip} H_{\ell=1,m}^{(2)} \dd u \dd \tilde{\sigma}}{2\ell+1}\frac{\log r}{r}\lesV r^{-1}.
	\end{equation}

	b) If  \cref{eq:time:apriori} holds for some $\alpha\in\R_{\geq0}$ with $\alpha<\beta$, then we have, for $ \ell\geq \alpha+2$: $(r^2\Xx)^{\lfloor \alpha\rfloor+1}( T^{-1}\psi_{\ell})|_{\Sigma_0} \lesV 1$. Furthermore,  for $1\leq \ell <\max( \alpha+1/2+\delta,\beta)$, we have
	\begin{equation}\label{eq:time:proplog}
	T^{-1}\psi_{\ell,m}(\tau_0,r)-\sum_{j=1}^{\ell-1} \frac{c^{(\ell)}_i}{r^i} +\frac{\int_{\scrip} H_{\ell,m}^{(\ell+1)} \dd u \dd \tilde{\sigma}}{2\ell+1}\frac{\log r}{r^{\ell}}\lesV r^{-\ell}.
	\end{equation}
		
	c) In the case where  \cref{eq:time:apriori} holds with $\alpha=3/2$,  and if  $\min(2+\delta,\beta)>\mathfrak{L}\in \N$, then we have, for $\ell\geq \mathfrak{L}+2$: $(r^2\Xx)^{\mathfrak{L}+1}(T^{-1}\psi_\ell)|_{\Sigma_{\tau_0}}\lesV 1$. On the other hand, for $\ell<\min(2+\delta,\beta)$, we have $(r^2\Xx)^{\ell}(T^{-1}\psi_\ell)|_{\Sigma_{\tau_0}}\lesV \log r$.
\end{prop}
\begin{proof}
To also cover the computations for higher time integrals later on, we provide formulae for $T^{-N}\phi$ for a general $N\in\N$.  The computations below are then justified if  \cref{eq:time:apriori} holds for $\alpha=N$; we will  make this assumption for now. 

\textit{Proof of a):}  As before, we will first give the proof with $\mathfrak{F}=0$. Analogously to \cref{eq:time:l=0:waveeq}, we have
\begin{multline}
	\Xx   (r^2 Dw_\ell^2\Xx   T^{-N-1}\check{\phi}_{\ell,m})-rw_{\ell}\Xx   (T^{-N}r\phi_{\ell,m})	=\Xb(r^2D w_{\ell}^2 \Xb T^{-N-1}\philcm)+r w_{\ell} \Xb(T^{-N}r\phi_{\ell,m})\\
	=-\int_{\tau_0}^{\infty}\int_{S^2}r^2 w_{\ell} \frac{(\tau_0-\tau')^{N}}{N!}F[\phi](\tau',r,\vartheta',\varphi')Y_{\ell,m}\dd \tau' \dd \vartheta' \dd \varphi'.
\end{multline}
As in the $\ell=0$-case, we then have, for any $r_{\infty}>\rc$:
\begin{align}\label{eq:time:higherell:wave}\nonumber
r^2 Dw_{\ell}^2 \Xx T^{-1-N}\philcm(\tau_0, r_{\infty})&=r^2 w_{\ell} T^{-N}\phi_{\ell,m}(u_0,\rc) +\int_{\rc}^{r_{\infty}} r w_\ell \Xx (T^{-N}r\phi_{\ell,m}) \dd r+\int_{2M}^{\rc}-rw_{\ell}\Xb(T^{-N}r\phi_{\ell,m})\dd r \dd \sigma\\
&\qquad\qquad +\int_{\D_{\tau_0,\infty}^{2M, r_{\infty}}}  w_{\ell}\frac{(\tau_0-\tau)^N}{N!}\frac{\partial_\alpha\left(G^{\alpha\beta}\partial_\beta\phi\right)}{2D\sin \vartheta} Y_{\ell,m}\dd r\dd \tau\dd \sigma.
\end{align}
We now compute the spacetime integral. Compared to the $\ell=0$ case, we now do not only get boundary terms, but also bulk terms coming from when derivatives hit $w_\ell$, $(\tau-\tau_0)^N$ or $Y_{\ell,m}$, respectively.
\begin{nalign}\label{eq:time:higherl:proof2}
&2\int_{\D_{\tau_0,\infty}^{2M, r_{\infty}}}  w_{\ell}\frac{(\tau_0-\tau)^N}{N!}\frac{\partial_\alpha\left(G^{\alpha\beta}\partial_\beta\phi\right)}{2D\sin \vartheta} Y_{\ell,m}\dd r\dd \tau\dd \sigma\\
=-&\int_{\D_{\tau_0,\infty}^{2M, r_{\infty}}}  \frac{Dw_\ell'}{w_{\ell}}\left(w_\ell\frac{(\tau_0-\tau)^N}{N!} (G^{v\beta}-G^{u\beta})\partial_\beta\phi \cdot Y_{\ell,m} \right)
+w_{\ell} \frac{(\tau_0-\tau)^N}{N!} G^{A\beta}\partial_\beta\phi \cdot \partial_A Y_{\ell,m}\\
&\qquad\qquad\qquad  -w_{\ell} \frac{(\tau_0-\tau)^{N-1}}{(N-1)!} \(\mathbf{1}_{r\leq \rc} G^{v\beta}+\mathbf{1}_{r\geq \rc}G^{u\beta}\)\partial_\beta\phi \cdot Y_{\ell,m} \dd u \dd v \dd \tilde{\sigma}\\
+&\int_{\hplus} w_{\ell} \frac{(\tau_0-\tau)^N}{N!} G^{u\beta}\partial_\beta\phi \cdot Y_{\ell,m} \dd v \dd \tilde{\sigma} 
+\int_{\Gamma_{r_{\infty}}} w_{\ell}\frac{(\tau_0-\tau)^N}{N!}(G^{v\beta}-G^{u\beta})\partial_\beta\phi\cdot Y_{\ell,m} \dd u \dd \tilde{\sigma}\\
-&\int_{\C_{\tau_0}} w_{\ell} \frac{(\tau_0-\tau)^N}{N!} G^{u\beta}\partial_\beta\phi \cdot Y_{\ell,m} \dd v \dd \tilde{\sigma}
-\int_{\Cbar_{\tau_0}} w_\ell \frac{(\tau_0-\tau)^N}{N!} G^{v\beta}\partial_\beta\phi \cdot Y_{\ell,m} \dd u \dd \tilde{\sigma},
\end{nalign}
where the limiting boundary terms as $\tau\to\infty$ vanish by  \cref{eq:time:apriori} (with $\alpha=N$), and we used that $\pu w_{\ell}=-Dw_\ell'$.
Now, if  \cref{eq:time:apriori} holds for $\alpha=N$, then the integrals along $\hplus$, $\Cbar_{\tau_0}$ and $\C_{\tau_0}$ in \cref{eq:time:higherl:proof2} will each give a finite contribution as $r_\infty\to\infty$. 

For the boundary term along $\Gamma_{r_{\infty}}$, we argue analogously to \cref{eq:time:l=0:proof2} to infer:
\begin{nalign}
		\int_{\Gamma_{r_{\infty}}} w_{\ell} \frac{(\tau_0-\tau)^N}{N!} (G^{v\beta}-G^{u\beta})\partial_\beta\phi Y_{\ell,m} \dd u \dd \tilde{\sigma}
	= r_{\infty}^{\ell-1}\int_{\scrip} \frac{(\tau_0-\tau)^N}{N!} G^{v\beta}\partial_\beta(r\phi) Y_{\ell,m}\dd u \dd \tilde{\sigma}+O(r_{\infty}^{\ell-2}+r_{\infty}^{\ell-1/2-\delta+}).
\end{nalign}

Moving on, we analyse the bulk terms in \cref{eq:time:higherl:proof2}. Since the term restricted to $r\leq \rc$ gives a finite contribution, we need only consider the $r\geq \rc$ part of the bulk term:
\begin{nalign}\label{eq:time:higherl:bulklimit1}
&\int_{\D_{\tau_0,\infty}^{\rc, r_{\infty}}} Dw_\ell'{\frac{(\tau_0-\tau)^N}{N!}} (G^{v\beta}-G^{u\beta})\partial_\beta\phi \cdot Y_{\ell,m} 
+w_{\ell}{ \frac{(\tau_0-\tau)^N}{N!}} G^{A\beta}\partial_\beta\phi \cdot \partial_A Y_{\ell,m} 
-w_{\ell} {\frac{N(\tau_0-\tau)^{N-1}}{N!}} G^{u\beta}\partial_\beta \phi \cdot Y_{\ell,m} \dd \tilde{\mu}\\
=&\int_{\rc}^{r_{\infty}} \frac{w_\ell'}{r} \dd v \int_{\scrip} {\frac{(\tau_0-\tau)^N}{N!}} G^{v\beta}\partial_\beta(r\phi) Y_{\ell,m}\dd u \dd \tilde{\sigma}+\int_{\rc}^{r_{\infty}} \frac{w_\ell}{r^2}\dd v\int_{\scrip} {\frac{(\tau_0-\tau)^N}{N!}} r^2 G^{A\beta}\partial_\beta\phi \cdot \partial_A Y_{\ell,m} \dd u \dd \tilde{\sigma}\\
&\qquad\qquad\qquad\qquad-\int_{\rc}^{r_{\infty}} \frac{w_\ell}{r^2}\dd v \int_{\scrip}  {\frac{N(\tau_0-\tau)^{N-1}}{N!}} r^2  G^{u\beta}\partial_\beta \phi Y_{\ell,m} \dd u \dd \tilde{\sigma}+O(r_{\infty}^{\ell-2}(\delta_{\ell,2}\log r_{\infty}+1)+r_{\infty}^{\ell-1/2-\delta+}),
\end{nalign}
where the logarithmic term for $\ell=2$ is due to integration of higher-order terms of the form $r^{-3}w_\ell$ etc. 
From now on, we set $\alpha=N=0$ and work with $T^{-1}\phi$.
Note that
\begin{equation}
		\int_{\rc}^{r_{\infty}}\frac{ \ell^{-1} w_\ell'}{r}\dd v\,,		\int_{\rc}^{r_{\infty}}\frac{  w_\ell}{r^2}\dd  v =\int_{\rc}^{r_{\infty}}  r^{\ell-2}+O(r^{-\ell-3})\dd v=\begin{cases}
		\frac{1}{\ell-1}r_{\infty}^{\ell-1}+O(r_{\infty}^{\ell-2}(\delta_{\ell,2}\log r_{\infty}+1)),&\ell >1,\\
		\log {r_{\infty}} +O(1), &\ell=1.
	\end{cases}
\end{equation}
Now, for $\ell\geq2$, the computations above imply that, for constants $C_{\ell,m}$,
\begin{equation}\label{eq:time:higherell:proof:r2Xx}
	r^2 D w_{\ell}^2 T^{-1}\philcm(\tau_0,r_\infty)=C_{\ell,m} r_{\infty}^{\ell-1}+O(r_{\infty}^{\ell-2}(\delta_{\ell,2}\log r_{\infty}+1)+r_{\infty}^{\ell-1/2-\delta+}),
\end{equation}
which implies that $r T^{-1}\phi_\ell\lesV 1/r$, and that $r^2\pv(rT^{-1}\phi_\ell)\lesV 1$. This proves the $\ell\geq2$ statement of \textit{a)}.

On the other hand, for $\ell=1$, we have, using \cref{lem:dyn}  \todo{signs}
\begin{align}\nonumber
-\left.	r^2Dw_{\ell}^2 \Xx T^{-1}\check{\phi}_{\ell=1,m}
\right|_{\Sigma_{\tau_0}}	=&\frac{\log r}{2}\Big(\int_{\scrip}( G^{vB}\partial_B(r\phi) -TG^{uv} r\phi)Y_{1,m}\dd u \dd \tilde{\sigma}+\int_{\S_{\tau_0}^{\infty}} G^{uv}r\phi Y_{1,m} \dd \tilde{\sigma}\Big)\\\label{eq:time:timeinversionexpressionl=1}
	&+\frac{\log r}{2} \int_{\scrip} (r G^{AB}\partial_B(r\phi)-G^{vA}r\phi) \cdot \partial_A Y_{1,m} \dd u \dd \tilde{\sigma} +O_{\V}(1)\\
	=\frac{\log r}{2} \int_{\scrip}  (4 \mathfrak{m}-4M) \psi Y_{1,m}&-2\Big((\gsh^{(1)})^{AB}\sl_B\psi-2(b^{(1)})^{A} \psi )\Big)\sl_A Y_{1,m} \dd u \dd \sigma+O_{\V}(1)=-\tilde{C}_{1,m}\log r+O_{\V}(1).\nonumber
\end{align}
The result \cref{eq:time:prop:first:l=1} then follows from  computing (we drop all lower order terms):
\begin{equation}
	\Xx   T^{-1}\philcm=\tilde{C}_{1,m} r^{-2-2\ell}\log r+\dots\implies -T^{-1}\phi_{\ell,m} =\frac{\tilde{C}_{1,m}}{1+2\ell}r^{-1-\ell}\log r+\dots \implies r^2 \Xx   T^{-1}r\phi_{\ell,m}= \frac{\tilde{C}_{1,m}\ell }{(2\ell+1)}\log r+\dots
\end{equation}

\textit{Proof of b):} We assume that  \cref{eq:time:apriori} holds for some real $\alpha\geq 0$. We may then write
\begin{nalign}\label{eq:time:proof:loggenerator}
&-\frac12	\int_{\D_{\tau_0,\infty}^{\rc, r_{\infty}}} Dw_{\ell}' (G^{v\beta}-G^{u\beta})\partial_\beta\phi Y_{\ell,m}+w_{\ell} G^{A\beta}\partial_\beta \phi\partial_A Y_{\ell,m} \dd \tilde{\mu}\\
	=&\sum_{i=2}^{\lfloor \alpha+3/2+\delta-\rfloor} \int_{\rc}^{r_{\infty}} r^{\ell-i} \dd r \int_{\scrip} \underbrace{\left[-(2D)^{-1}r^{-\ell}\left(D w_{\ell}' (G^{v\beta}-G^{u\beta})\partial_\beta\phi Y_{\ell,m}+w_{\ell} G^{A\beta}\partial_\beta \phi\partial_A Y_{\ell,m}\right)\right]^{(i)}}_{H^{(i)}_{\ell,m}(u,\vartheta,\varphi)}   \dd u \dd \tilde{\sigma}\\
	&+O(r^{\ell-\alpha-1/2-\delta+}+r^{\ell-\alpha-2}(\delta_{\ell,\alpha+2}\log r+1)),
\end{nalign}
since each $H_{\ell,m}^{(i)}$ satisfies $H_{\ell,m}^{(i)}\lesssim u^{-1/2-\delta-\alpha+i-2}$. Thus, for $\ell <\alpha+1/2+\delta$,  the $i=\ell+1$-term generates a logarithmic term
\begin{equation}
	r^2Dw_{\ell}^2\Xx T^{-1} \philcm=\dots +\log r \int_{\scrip} H^{(\ell+1)}_{\ell,m} \dd u \dd \tilde{\sigma}+\dots
\implies 
r	T^{-1}\phi_{\ell,m} =\dots-\frac{1}{1+2\ell} r^{-\ell}\log r   \int_{\scrip} H_{\ell,m}^{(\ell+1)} \dd u\dd \tilde{\sigma} +\dots;
\end{equation}
this proves \cref{eq:time:proplog}.
If $\ell \geq \alpha+2$,  there are no logarithms for $T^{-1}\psi_\ell$ at order $\lfloor \alpha\rfloor+1$, completing the proof of~\textit{b)}.

\textit{Proof of \textit{c)}:} We apply the arguments above with $\alpha=3/2$: For $\ell<\alpha+1/2+\delta=2+\delta$, the statement is the same as in \textit{b)}. For $\ell\geq \mathfrak{L}+2$, we have that $H_{\ell,m}^{(\mathfrak{L}+1)}\lesssim u^{-1-}$. If $\ell\geq \mathfrak{L}+2$, this then gives a $1/r$-expansion for $T^{-1}\psi_{\ell}$ up to order $r^{-\mathfrak{L}-1}$.

\textit{The case $\mathfrak{F}\neq 0$:} We simply pick up an additional spacetime integral $-\int_{\D_{\tau_0,\infty}^{2M,r_{\infty}}} r^2w_{\ell} \mathfrak{F}\dd \mu_D$ in \cref{eq:time:higherell:wave}, resulting in $H^{(\ell+1)}_{\ell,m}$ as given in \cref{eq:main:thm:integralexpression}. 
\end{proof}
\begin{lemma}\label{cor:time:l=1alternative}
For $\ell=1$, if $g$ solves the Einstein vacuum equations, we can simplify the integral in \cref{eq:time:prop:first:l=1} as follows:
		\begin{equation}
		\int_{\scrip} H_{\ell=1,m}^{(2)} \dd u \dd \tilde{\sigma}=-\int_{\scrip}\psi\Big( (2 \mathfrak{m}-2M) Y_{\ell=1,m}  +(\gsh^{(1)})^{AB} \sl_A \sl_B Y_{\ell=1,m} +(r w_\ell\mathfrak{F})^{(4)} Y_{\ell,m}\Big) \dd u \dd {\sigma}.
	\end{equation}
\end{lemma}
\begin{proof}
		The result follows from using the relation $b^{(1)}=\mathring{\div} \gsh^{(1)}$ derived in \cref{app:motivation}.
%
%
\end{proof}

\subsubsection{Summing in $\ell$ and improving $\alpha$ in \cref{eq:time:apriori}}
Having already obtained estimates for the time inverse for fixed $\ell$-modes, we now prove 
\begin{prop}\label{prop:time:firstcomplete}
Let $\delta>1/2$.	If \cref{eq:time:apriori} holds with $\alpha=0$, we in fact have that
	\begin{equation}
E_N[T^{-1}\phi](\tau_0)+\int_{\C_{\tau_0}} r^2 (\pv (rT^{-1}\phi))^2 \dd \mu \lesV[\Vc]1.
	\end{equation}
\end{prop}
\begin{proof}
	We write $\phi=\phi_{\ell\leq \ell_0}+\phi_{\ell>\ell_0}$ for $\ell_0$ sufficiently large. 
	For the lower angular modes, the result follows from \cref{prop:time:l=0:1,prop:time:l>0:1}.
	For the remainder, we apply the elliptic estimate from \cref{prop:inhom:elliptic1} with $q=-2$, or, alternatively, \cref{prop:inhom:elliptic2} with $q=2$ together with Hardy's inequality.
\end{proof}
We can now directly infer:
\begin{cor}\label{cor:time:firstiterate}
	Let $\delta>1/2$. If \cref{eq:time:apriori} holds for $\alpha=0$, then it also holds for $\alpha=1$.
\end{cor}
\begin{proof}
	By \cref{eq:main:assonF}, \cref{eq:time:apriori} and \cref{eq:time:aprioriF} with $\alpha=0$, we have  that $T^{-1}\phi$ satisfies
	\begin{equation}\label{eq:time:completfirstinhom}
		\Box_{g_M}T^{-1} \phi =T^{-1}F+T^{-1}\mathfrak{F}, \quad \text{with}\quad T^{-1}F \lesV r^{-3} (\tau^{-1/2-\alpha-\delta}+\tau^{-\beta})\min_{q\in[-1,1]}\left(\frac{\tau}{r}\right)^q
	\end{equation} 
	with $\alpha=0$.
	Since we also have \cref{prop:time:firstcomplete}, we are therefore in the setting of \cref{prop:inhom:energydecay} and \cref{cor:inhom:decay} with $F_1=T^{-1}F$, $F_2=T^{-1}\mathfrak{F}$,  $\eta=\min(1/2+\delta,\beta)$ and $\tilde{p}=2$ (and with $\phi$ replaced by $T^{-1}\phi$). 
	We may therefore conclude that $T^{-1}\psi\lesV \tau^{-\delta+}+\tau^{-1/2}$ and $T^{-1}\phi\lesV \tau^{-1-\delta+}+\tau^{-3/2}$; this proves the result.

Note that if we merely had $\delta>0$, we would still be able to conclude that \cref{eq:time:apriori} holds with $\alpha=\min(1/2+\delta-,1)$. We may then plug this new $\alpha$ into \cref{eq:time:completfirstinhom} and repeat the argument to recover $n=1$.
\end{proof}
For later purposes, we also need the following:
\begin{prop}\label{prop:time:first:elliptic}
Let $\delta>1/2$.	If \cref{eq:time:apriori} holds for $\alpha=n\in\N$, or if \cref{eq:time:ass2} holds for $\min(2+\delta,\beta)>\mathfrak{L}=n\in\N$, then there exists $\ell_0$ sufficiently large such that
	\begin{equation}
		\sum_{k=0}^n \int_{\C_{\tau_0}^{\rc,\infty}} r^2 (\pv (r^2\Xx)^k(rT^{-1}\phi_{\ell\geq \ell_0}))^2 \dd \mu \lesV[\Vc]1.
	\end{equation}
	
\end{prop}
\begin{proof}
To ease the notion, we set $\mathfrak{F}=0$. We apply \cref{prop:inhom:elliptic2} to the equation \cref{eq:time:completfirstinhom} with $q=2+4N$ and $N=n+1$. 
This schematically gives, for $\phi=\phi_{\ell\geq \ell_0}$ sufficiently large:
\begin{equation}\label{eq:time:Hardyto}
	\int_{\C_{\tau_0}}( \Xx (r^2\Xx)^{n+1}T^{-1} \psi)^2 \dd \mu \lesssim \int_{\C_{\tau_0}} r^2( \Xx (r^2\Xx)^{n+1} \psi)^2 +r^{-2}( (r^2\Xx)^{n+1}(r^3 T^{-1}F))^2\dd \mu+\dots,
\end{equation}
where the RHS is finite since $\psi$ is compactly supported along $\Sigma_{\tau_0}$ and because the inhomogeneity satisfies, in view of \cref{eq:time:apriori:r2X}:
	\begin{equation}
	(r^2\Xx)^{n} (r^3F)\lesV \tau^{-3/2-\delta+}\implies (r^2\Xx)^{n+1}(r^3 F)=\tau^{-1/2-\delta+}\min_{q\in[0,1]}\left(\frac{r}{\tau}\right)^q\implies T^{-1}(r^2\Xx)^{n+1}(r^3F)\lesV \tau^{-}r^{1/2-},
\end{equation}
which is integrable when divided by $r^2$ (this argument also works for $\delta\leq1/2$, in which case we choose $q=1/2-\delta+$). 
The result then follows by applying  Hardy's inequality to \cref{eq:time:Hardyto}.
\end{proof}

\subsection{Subtraction of an approximate solution}\label{sec:time:subtract}
Henceforth, in view of \cref{cor:time:firstiterate}, we may assume that \cref{eq:time:apriori} holds with $\alpha=1$. As already explained in \cref{sec:time:overview}, we need to subtract an approximate solution to find the leading-order asymptotics:
	We define, for $\ell_0\in \N$ (whenever the integrals in \cref{eq:time:subtract:Ilm} below are well-defined): 
\begin{equation}
	\widehat{\phi}^{[\ell_0]}:= \phi-T\apphi[\ell=0] {\integral}_0[\phi]-\sum_{\ell=1}^{\ell_0}\sum_{m=-\ell}^{\ell} \apphi[\ell,m] {\integral}_{\ell,m}[\phi]Y_{\ell,m},
\end{equation}
where our choice of the coefficients, in view of \cref{prop:time:l=0:1} and \cref{prop:time:l>0:1}, is:
\begin{align}\label{eq:time:subtract:Ilm} 
{\integral}_0[\phi]=\frac{1}{4\pi}\int_{\scrip}  \mathfrak{m}\psi -\frac12 \mathfrak{F}^{(4)}\dd u \dd {\sigma}&&	{\integral}_{\ell,m}[\phi]: \frac{(-1)^{\ell}}{(2\ell+1)\binom{2\ell}{\ell}}\int_{\scrip}- H_{\ell,m}^{(\ell+1)} \dd u \dd \tilde{\sigma} \, \text{ if }\, \ell>0.
\end{align}
Indeed, this choice guarantees, in view of \cref{eq:time:l=0NPconstant} for $\ell=0$, and in view of  \cref{eq:time:proplog} and  \cref{eq:approx:time} for $\ell>0$:
\begin{lemma}
	If \cref{eq:time:apriori} holds with $\alpha=1$, then 
	\begin{equation}\label{eq:time:subtr:l=0}
		r^2\pv(T^{-1}\widehat{\phi}_{\ell=0}^{[0]})|_{\Sigma_{\tau_0}}\lesV r^{-\max(1,\delta+1/2-)}.
	\end{equation}
	Similarly, if  \cref{eq:time:apriori} holds for $\alpha\geq 1$, then, for $0<\ell <\max(\alpha+1/2+\delta,\beta)$ :
	\begin{equation}\label{eq:time:subtr:l>0}
		T^{-1}\widehat{\phi}^{[\ell]}_{\ell}|_{\Sigma_{\tau_0}}=\sum_{j=0}^{\ell-1}\frac{\bar{c}_i^{(\ell)}}{r^i}+O_{\V}(r^{-\ell}), \quad \text{or, equivalently,} \quad  (r^2\Xx)^{\ell} (T^{-1}\widehat{\phi}^{[\ell]}_{\ell})|_{\Sigma_{\tau_0}}\lesV 1.
	\end{equation}
\end{lemma}

\subsection{Second time integration}\label{sec:time:second}
We now analyse the asymptotics of the initial data for the second time integral for the difference quantity~$\widehat{\phi}^{[1]}$.
\subsubsection{Initial data for second time integral, $\ell=0$}
We first consider $\widehat{\phi}^{[0]}_{\ell}$ with $\ell=0$.
\begin{prop}\label{prop:time:second:l=0}
	Assuming \cref{ass:time} with $\alpha=1$ and $\delta>1/2$, we have $r^2\pv(T^{-2}\widehat{\phi}^{[0]}_{\ell=0})|_{\Sigma_{\tau_0}}\lesV \log r$ along $\Sigma_{\tau_0}$.
\end{prop}
\begin{rem}\label{rem:time:secondorder}
	We can, in principle, also prove the more precise statement that $r^2\pv(rT^{-2}\widehat{\phi}^{[0]})|_{\Sigma}-{\integral}_{0}^{(1)}[\phi] \log r\lesV 1$ for an appropriate constant ${\integral}_{0}^{(1)}[\phi]$.
	This would then allow us to define a new approximate solution that we could subtract from $\widehat{\phi}^{[0]}$ (cf.~\cref{rem:approx}), i.e.~it would allow us to deduce \emph{higher-order} late-time asymptotics as well. We shall not do this here.
\end{rem}
\begin{proof}
By \cref{eq:time:subtr:l=0}, we have $\Xx(rT^{-1}\widehat{\phi}^{[0]}_{\ell=0})\lesV r^{-3}$ along $\C_{\tau_0}$. 
We can therefore repeat the computation in \cref{eq:time:l=0:proof1} with $\phi_{\ell=0}$ replaced by $T^{-1}\widehat{\phi}^{[0]}_{\ell=0}$. (Notice that $T^{-1}\widehat{\phi}^{[0]}$ satisfies \cref{eq:time:apriori} with $\alpha=0$.)
The only difference will then be that, on the RHS of \cref{eq:time:l=0:proof1}, we have an extra inhomogeneous term of the form 
\begin{equation}
	\int_{\D_{\tau_0,\infty}^{r_{\infty},\infty}} r^2 \Box_{g_M}(\apphi[0,0]\cdot Y_{0,0})\dd \mu \lesssim 	\int_{\D_{\tau_0,\infty}^{r_{\infty},\infty}} \frac{1}{\tau(\tau+r)^2} \dd \mu \lesssim \frac{\log r_{\infty}}{r_{\infty}},
\end{equation}
where we appealed to \cref{prop:approx}.
\end{proof}

\subsubsection{Initial data for the second time integral, $\ell>0$}
We next consider $\widehat{\phi}^{[1]}_{\ell}$ with $\ell>0$.
\begin{prop}\label{prop:time:second:highl}
	Let $\delta>1/2$, and let $\ell\geq 1$. Then $T^{-2}\widehat{\phi}^{[1]}_{\ell}$ satisfies along $\Sigma_{\tau_0}$:
	
	a) If  \cref{eq:time:apriori} holds with $\alpha=1$, then $r^2\pv(T^{-2}\widehat{\phi}^{[1]}_{\ell=1})\lesV \log r$, $r^2\pv(T^{-2}\psi_{\ell=2})\lesV \log r$ and $r^2\pv(T^{-2}\psi_{\ell})\lesV 1$  for $\ell > 2$.\footnote{Without the subtraction, we would merely have
	$r^2\pv(T^{-2}\psi_{\ell=1})\lesV r$. }
	
		b) If  \cref{eq:time:apriori} holds with $\alpha=3/2$, and if  $\min(2+\delta,\beta)>\mathfrak{L}\in\N$,  then we have, for $\ell\geq \mathfrak{L}+2$: $(r^2\Xx)^{\mathfrak{L}}(T^{-2}\psi_\ell)|_{\Sigma_{\tau_0}}\lesV 1$. On the other hand, for $1\leq \ell<\min(2+\delta,\beta)$, we have $(r^2\Xx)^{\ell-1}(T^{-2}\psi_\ell)|_{\Sigma_{\tau_0}}\lesV \log r$.

\end{prop}
\begin{proof}
	For simpler reading, we again set $\mathfrak{F}=0$, the case $\mathfrak{F}\neq0$ can be obtained in a straightforward manner.
	
	\textit{$\ell\geq2$:}	We first prove the $\ell\geq2$-part of the proposition: We revisit the proof of \cref{prop:time:l>0:1} with $N=1$. We then see that the bulk term (the integral over $\D_{\tau_0,\infty}^{2M,r_{\infty}}$), as well as the integrals along $\hplus$ and $\Gamma_{r_{\infty}}$in \cref{eq:time:higherl:proof2} can be treated exactly as for the first time inversion $N=0$, since we now have exactly 1 power more decay to counter the additional $(\tau_0-\tau)$-factor (since we have upgraded  \cref{eq:time:apriori} from $\alpha=0$ to $\alpha=1$). 
	Furthermore, the terms along $\C_{\tau_0}$ and $\Cbar_{\tau_0}$ in the last line of \cref{eq:time:higherl:proof2} vanish identically for $N>0$.
	
	Thus, the only place that produces different asymptotic behaviour compared to the first time inversion is the first line of \cref{eq:time:higherell:wave}: Indeed, if we plug \cref{eq:time:proplog} into the $\C_{\tau_0}$-integral of \cref{eq:time:higherell:wave} we obtain schematically that, for $1\leq \ell <n+1/2+\delta$:
	\begin{equation}\label{eq:time:proof:highellsecondtimeinverse}
	r^{2}D w_{\ell^2} \Xx T^{-2} \philcm\sim 	\int_{\rc}^{r_{\infty}} rw_\ell \Xx(\dots+ \frac{\log r}{r^{\ell}}+\dots) \dd r \dd \sigma =\dots+r_{\infty}\log r_{\infty}+\dots,
	\end{equation}
	which then gives that $rT^{-2}\phi_{\ell}\sim \dots +\frac{\log r}{r^{\ell-1}}$, i.e.~the log-term is pushed one order forwards.\footnote{There will, however, still be a $\log r$-term in \cref{eq:time:proof:highellsecondtimeinverse}, i.e.~a $r^{-\ell}\log r$-term in $T^{-2}\psi_\ell$!} This is enough to prove the $\ell\geq 2$-statement of \textit{a)} as well as all of \textit{b)}.
	
\textit{$\ell= 1$:}	We now set $\ell=1$ and prove the remaining part of \textit{a)}: By construction of $\widehat{\psi}^{[1]}$, the leading-order log-term ($T^{-2}\psi_{\ell=1}\sim \log r$) generated in \cref{eq:time:proof:highellsecondtimeinverse} disappears if we consider $\widehat{\psi}^{[1]}_{\ell=1}$, cf.~\cref{eq:time:subtr:l>0}.
	
As in the $\ell=0$-case, the price to pay for the subtraction is a new inhomogeneity added to  \cref{eq:time:higherell:wave}, namely:
\begin{equation}
	\int_{\D_{\tau_0,\infty}^{2M,r_{\infty}}} w_\ell \frac{(\tau_0-\tau)^N}{N!} r^2\Box_{g_M} (\apphi\cdot Y_{\ell,m}) \dd u \dd v
	\end{equation}
with $N=\ell=1$. 
Then, we can estimate, for any $N=\ell\in\N$ (we keep this general for later reference):
\begin{equation}\label{eq:time:inhomogeneityintegral}
		\int_{\D_{\tau_0,\infty}^{2M,r_{\infty}}} w_\ell \frac{(\tau_0-\tau)^N}{N!} r^2\Box_{g_M} (\apphi\cdot Y_{\ell,m}) \dd u \dd v\lesssim 	\int_{\D_{\tau_0,\infty}^{2M,r_{\infty}}} w_\ell \frac{(\tau_0-\tau)^N}{N!} \frac{r^{\ell}}{(\tau+r)^{\ell+2}\tau^{\ell+1}} \dd u \dd v \lesssim 
r_{\infty}^{\ell-2 }\log r_{\infty}+C
\end{equation}
for some constant $C$. 
This logarithmic term will contribute to $T^{-\ell-1}\widehat{\psi}^{[\ell]}_{\ell}$ only at subleading order $r^{-2}\log r$.\footnote{ That is, in contrast to the $\ell=0$ case, the subtraction of an approximate solution induces a logarithmic term that appears \textit{two powers later} than the logarithmic term that we eliminated via the subtraction! In particular, if we wanted to find the next-to-leading order asymptotics for $\ell=1$, they would not depend come from the subtracted approximate solution.} 
We thus get that $r^2\Xx T^{-\ell-1}\widehat{\psi}^{\ell}_{\ell}\lesV \log r$ for $\ell=1$, as required.
\end{proof}

\subsubsection{Summing in $\ell$ and proving \cref{eq:time:ass2} with $\mathfrak{L}=2$}
\begin{prop}\label{prop:time:secondcomplete}
Let $\delta>1/2$. 	Under  \cref{eq:time:apriori} with $\alpha=1$, we in fact have that
	\begin{equation}\label{eq:time:second:propsum1}
		E_N[T^{-2}\widehat{\phi}^{[1]}](\tau_0)+\int_{\C_{\tau_0}^{R,\infty}} r^2 (\pv (rT^{-2}\widehat{\phi}^{[1]}))^2 \dd \mu \lesV[\Vc]1.
	\end{equation}

	Moreover, if \cref{eq:time:ass2} holds for $\min(2+\delta,\beta)>\mathfrak{L} \in\N$, then we have, for $\ell_0$ sufficiently large, that
		\begin{equation}\label{eq:time:second:propsum2}
		\sum_{k=0}^{\mathfrak{L}-1} \int_{\C_{\tau_0}^{\rc,\infty}} r^2 (\pv (r^2\Xx)^k(rT^{-2}\phi_{\ell\geq \ell_0}))^2 \dd \mu \lesV[\Vc]1.
	\end{equation}
\end{prop}
\begin{proof}
	We write $\widehat{\phi}^{[1]}=\widehat{\phi}^{[1]}_{\ell\leq \ell_0}+\widehat{\phi}^{[1]}_{\ell>\ell_0}$ for $\ell_0$ sufficiently large. For the lower spherical harmonics, \cref{eq:time:second:propsum1} readily follows from \cref{prop:time:second:l=0,prop:time:second:highl}.
	For the higher spherical harmonics, we appeal to \cref{prop:time:first:elliptic}, \cref{prop:inhom:elliptic2} and Hardy's inequality; this also proves \cref{eq:time:second:propsum2}.
		\end{proof}

	\begin{cor}\label{cor:time:seconditerate}
	Let $\delta>1/2$.	If  \cref{eq:time:apriori} holds with $\alpha=1$, then it also holds with $\alpha=3/2$. More precisely, we have that \cref{eq:time:apriori} holds with $\alpha=2$ if $\phi$ is replaced by $\widehat{\phi}^{[1]}$ or with $\phi_{\ell \geq 2}$. In particular, \cref{eq:time:ass2} holds with $\mathfrak{L}=2$.
	\end{cor}
\begin{proof}
We again drop $\mathfrak{F}$ for brevity.	The proof proceeds similarly to the one of \cref{cor:time:seconditerate}, except that we now have a new inhomogeneity.
	More precisely, $T^{-2}\widehat{\phi}^{[1]}$ satisfies the equation $	\Box_{g_M}T^{-2}\widehat{\phi}^{[1]}=T^{-2}F$, where
	\begin{equation}
  T^{-2}F\lesV \underbrace{r^{-3}\tau^{-1/2+1-\alpha-\delta}\min_{q\in[-1,1]}\(\frac{\tau}{r}\)^q}_{=F_1} +\underbrace{T^{-1}\Box_{g_M}( \apphi[0,0]\cdot Y_{00})+\left.\sum_{m=-\ell}^{\ell}T^{-1-\ell}\Box_{g_M}(\apphi\cdot{Y_{\ell,m}})\right|_{\ell=1}}_{F_2} 
	\end{equation}
	Notice that we have, for any $\ell\in \N$, and for $\epsilon>0$ arbitrarily small:
	\begin{equation}\label{eq:time:inhomF2}
		T^{-1}\Box_{g_M}( \apphi[0,0]Y_{0,0})\lesV \frac{1}{r^2\tau^{\epsilon}(\tau+r)^{2-\epsilon}}, \quad T^{-1-\ell}\Box_{g_M}\apphi\cdot{Y_{\ell,m}}\lesssim \int_{\tau}^{\infty} \frac{(\tau'-\tau)^{\ell}r^{\ell}}{\tau^{\ell+1}(\tau+r)^{\ell+2}}\dd \tau' \lesssim \frac{r^{\ell-2}}{\tau^{\epsilon}(\tau+r)^{\ell+2-\epsilon}}.
	\end{equation}
	Thus, $F_2$ satisfies all assumptions of \cref{prop:inhom:energydecay} and \cref{cor:inhom:decay},\footnote{In fact, we have $\fbulkp{\tau_1}{\tau_2}{p}[F_2]\lesssim \tau^{p-3+}$ and, similarly, $F_2\lesV \frac{r^{\ell-2-}}{\tau^{\ell+2-}}$.} 
	which means that we can apply \cref{cor:inhom:decay} with $\eta=1/2+\delta$.
	This gives, after taking two time derivatives:
$
		\widehat{\psi}^{[1]}\lesV \min(\tau^{-5/2},r \tau^{-7/2}).
$
\end{proof}
If we also want to show improved decay for higher $\ell$-modes, then we have to appeal to \cref{ass:main:aux}:
\begin{cor}\label{cor:time:seconditeratefine}
If \cref{eq:time:apriori} holds with $\alpha=1$, and if moreover \cref{ass:main:aux} holds with $L=1$, then the if-clause of \cref{eq:main:ass:aux} holds with $\tilde{\alpha}=5/2$ and for $\ell_1=1$ if $\phi$ is replaced by $\widehat{\phi}^{[1]}$, i.e.~we have $r^{-\ell}\widehat{\phi}^{[1]}_{\geq \ell}\lesV \tau^{-7/2-\ell}$ for $\ell\leq 1$.
\end{cor}
\begin{proof}
	The proof is as exactly as above, except that we can now work with more $\tau$-decay coming from the term $\partial_\mu(G^{\mu\nu}\partial_\nu\phi)|_{\ell \geq 1}$ afforded by \cref{eq:main:ass:aux}. More precisely, we start with the estimate $\phi_{\ell\geq 1}\lesV \tau^{-3}$ and use \cref{eq:main:ass:aux}, which allows us to apply the second part of \cref{cor:inhom:decay} with $\eta_{\ell_0=1}=2+\delta$. The result then follows from a single iteration.
\end{proof}

\subsection{First conclusion: Leading-order global asymptotics}\label{sec:time:inbetween}
\begin{thm}\label{thm:time:first}
Let $\delta>0$ and assume \cref{ass:main} as well as \cref{ass:main:aux} with $0\leq L\leq 1$. Then the global leading-order asymptotics of $\phi$ are, for $\ell_0=0,1$ given by
\begin{equation}
|\widehat{\psi}^{[1]}_{\ell \geq \ell_0}|\lesV \tau^{-5/2}\min_{q\in[0,\min(\ell_0,L)+1]}\left(\frac{r}{\tau}\right)^{q}.
\end{equation}
\end{thm}
\begin{proof}
	If $\delta>1/2$, this result follows from \cref{lemma:time:zeroth},  \cref{cor:time:firstiterate,cor:time:seconditerate,cor:time:seconditeratefine}.
	
	It is left to explain how to deal with the range $0<\delta\leq 1/2$:
	For this, we revisit the arguments of \cref{sec:time:first}--\cref{sec:time:second}:
	In \cref{sec:time:first}, we have used the assumption $\delta>1/2$ in \cref{eq:time:l=0:proof2} (for $\ell=0$) and \cref{eq:time:higherl:bulklimit1} (for higher $\ell$) in order for the integrals in those expressions to be well-defined. In turn, this allowed us to infer that $\frac{r^2}{\log r}\Xx(T^{-1}\psi_\ell)$ remains bounded along $\Sigma_{\tau_0}$ (\cref{eq:time:l=0:NP} and \cref{eq:time:higherell:proof:r2Xx}). If, in these steps, we had merely assumed $1/2\geq \delta>0$, then the same computations would give us that $r^{3/2+\delta-}\Xx(T^{-1}\psi_{\ell})$ remains bounded along $\Sigma_{\tau_0}$. 
	In turn, this is sufficient for proving \cref{prop:time:firstcomplete}, and, after a single iteration also \cref{cor:time:firstiterate} (first from $\alpha=0$ to $\alpha=1/2+\delta-$, then to $\alpha=1$). With this improved knowledge, we can redo the proofs of \cref{prop:time:l=0:1} and \cref{prop:time:l>0:1} without changes, so they also hold for $\delta>0$.
	
	Similarly, in \cref{sec:time:second}, if we have $\delta\leq1/2$, then we can at first only prove boundedness of  $r^{3/2+\delta-}\Xx T^{-2}\widehat{\psi}^{[1]}_{\ell}$. This is again enough to prove \cref{prop:time:secondcomplete}, and, in \cref{cor:time:seconditerate}, deduce that \cref{eq:time:apriori} holds with $\alpha=3/2$ and with $\alpha=3/2+\delta-$ for $\phi$ replaced by $\widehat{\phi}^{[1]}$. One then revisits \cref{prop:time:second:l=0} and \cref{prop:time:second:highl} and sees that they still hold for $\delta>0$. (In particular, we can upgrade the RHS of \cref{eq:time:l=0NPconstant} to $r^{-1}$.)
\end{proof}
\subsection{Third time integration and improving $\mathfrak{L}$ in \cref{eq:time:ass2}}\label{sec:time:third}

We can now similarly take the third time integral to find the asymptotic behaviour of the $\ell=2$ mode. Notice that  since we already found the leading-order asymptotics of $\psi$, we cannot improve \cref{eq:time:apriori} beyond $\alpha=3/2$; therefore, we will henceforth improve \cref{eq:time:ass2} instead.
\textbf{To ease notation, we will assume henceforth assume that $\beta\geq2+\delta$.} (The reader may replace any appearance of $\delta$ below with $\min(\beta-2,\delta)$.)
\begin{prop}\label{prop:time:third}
	Let $\ell\geq 2$. 
\begin{enumerate}[label=\emph{\alph*)}]
\item Let \cref{eq:time:ass2} hold with $\mathfrak{L}=2$. Then we have that $r^2\pv(T^{-3}r\widehat{\phi}^{[2]}_{\ell})\lesV\max( r^{1-\delta+},\log r)$ for  $\ell=2$ and $\ell=3$, and $r^2\pv(T^{-3}r\widehat{\phi}^{[2]}_{\ell})\lesV\max( r^{1-\delta+},1)$ for $\ell>3$.

\item If \cref{eq:time:ass2} holds for some $2+\delta>\mathfrak{L}\geq 2\in \N$,\footnote{This assumption is necessary to define the third time integral $\ell\geq2$!} then we have, for $\ell\geq \mathfrak{L}+2$: $(r^2\Xx)^{\mathfrak{L-1}}(T^{-3}\psi_\ell)|_{\Sigma_{\tau_0}}\lesV \max( r^{1-\delta+},1)$. On the other hand, for $2\leq \ell<2+\delta$, we have $(r^2\Xx)^{\ell-2}(T^{-3}\psi_\ell)|_{\Sigma_{\tau_0}}\lesV \max( r^{1-\delta+},\log r)$.
\end{enumerate}
\end{prop}
\begin{proof}
	The proof is analogous to that of \cref{prop:time:second:highl}, with the only new element being the $r^{1-\delta+}$-upper bound. 
The reason of this upper bound is that, for $\delta\leq 1$, we have instead of \cref{eq:time:proof:loggenerator} the following expression contributing to $r^2Dw_{\ell}^2 \Xx T^{-3}\widehat{\phi}^{[2]}_{\ell}$, where $N=2$ (we again drop the $r^2 w_\ell \mathfrak{F}Y_{\ell,m}$ term):
\begin{nalign}\label{eq:time:proof:loggenerator2}
	&	\int_{\D_{\tau_0,\infty}^{\rc, r_{\infty}}}\frac{(\tau_0-\tau)^{N}}{N!}\left( -Dw_{\ell}' (G^{v\beta}-G^{u\beta})\partial_\beta\phi Y_{\ell,m}-w_{\ell} G^{A\beta}\partial_\beta \phi\partial_A Y_{\ell,m} \right)+w_\ell\frac{n(\tau_0-\tau)^{n-1}}{n!} G^{u\beta}\partial_\beta \phi \cdot Y_{\ell,m} \dd \mu \\
	&\lesssim	\int_{\D_{\tau_0,\infty}^{\rc, r_{\infty}}} r^{\max(\ell-1-\delta+,\ell-2)}u^{-1-}\dd \mu \lesssim \max (r^{\ell-\delta+}, r^{\ell-1}).
\end{nalign}
For $\ell=2$, the additional inhomogeneity coming from the subtracted approximate solution is bounded exactly as in \cref{eq:time:inhomogeneityintegral}. The result then follows.
\end{proof}
Analogously to \cref{prop:time:secondcomplete}, we then have
\begin{prop}\label{prop:time:third:sum}
Let $\tilde{p}=\max(2,1+2\delta-)$ and assume that \cref{eq:time:ass2} holds with $\mathfrak{L}=2$. Then:
		\begin{equation}\label{eq:time:third:propsum1}
		E_N[T^{-3}\widehat{\phi}^{[2]}_{\ell\geq 2}](\tau_0)+\int_{\C_{\tau_0}^{R,\infty}} r^{\tilde{p}} (\pv (T^{-3}\widehat{\psi}^{[2]}_{\ell\geq 2}))^2 \dd \mu \lesV[\Vc] 1.
	\end{equation}
	
	Moreover, we have for $2+\delta>\mathfrak{L}\in\N$ and for $\ell_0$ sufficiently large that
	\begin{equation}\label{eq:time:third:propsum2}
		\sum_{k=0}^{\mathfrak{L}-2} \int_{\C_{\tau_0}^{\rc,\infty}} r^{\tilde{p}}(\pv (r^2\Xx)^k(T^{-
		3}\psi_{\ell\geq \ell_0}))^2 \dd \mu \lesV[\Vc]1.
	\end{equation}
\end{prop}
\begin{proof}
	The proof of \cref{eq:time:third:propsum1} proceeds as for the second time inversion. Notice that the upper bound  $\pv(T^{-3}r\widehat{\phi}^{[2]})\lesV r^{-1-\delta}$  is sufficient for the finiteness of the $r^p$ flux with $p=\tilde{p}$.
	
	For the proof of \cref{eq:time:third:propsum2}, we argue as in \cref{prop:time:first:elliptic}. 
	Denoting $\lfloor 2+\delta -\rfloor=N$, we need to control, in particular, the integral
	\begin{equation}\label{eq:time:r2xinhomproof}
		\int_{\C_{\tau_0}} r^{\tilde{p}-4} (r^2\Xx)^{N-1}(r^3T^{-3}F)^2 \dd \mu.
	\end{equation}
	Since \cref{ass:time} holds with $\alpha=3/2$, we have
	\begin{equation}
		(r^2\Xx)^{N-1}(r^3F)\lesV \frac{1}{\tau^{2+\delta-N+2}}\min_{q\in[0,1]}\left(\frac{r}{\tau}\right)^q\lesV \frac{1}{\tau^{3+}} r^{N-1-\delta+}\implies T^{-3}(r^2\Xx)^{N-1}(r^3F)\lesV \tau^{-}r^{N-1-\delta+},
	\end{equation}
which gives that the integral \cref{eq:time:r2xinhomproof} is finite (the condition $\tilde{p}-2-\delta<-1$ is satisfied).
	\end{proof}
	We can finally infer improved decay:
	\begin{cor}\label{cor:time:third}
	If \cref{eq:time:ass2} holds with $\mathfrak{L}=2$, then it also holds with $\mathfrak{L}=\min(3,5/2+\delta-)$.

	If, in addition, \cref{ass:main:aux} holds with $L\leq 2$, then we also have the estimate 
	\begin{equation}
		\widehat{\psi}^{[2]}_{\ell\geq 2}\lesV (\tau^{-7/2}+\tau^{-3-\delta+})\min_{q\in[0,L+1]} \left(\frac{\tau}{r}\right)^{q}.
	\end{equation}
	\end{cor}
\begin{proof}
	The proof is entirely analogous to the proofs of \cref{cor:time:seconditerate} and \cref{cor:time:seconditeratefine} and therefore left to the reader. (Note that \cref{eq:time:inhomF2} holds for any $\ell$!)
\end{proof}
\subsection{Higher-order time integration and mode-coupling}\label{sec:time:nth}
 If we wanted to take the fourth time integral, then we would find that the integral in time of $H_{3m}^{(4)}$, which schematically looks like $u^{2-\delta}\cdot r\phi$, is no longer well-defined if $\delta<1$ (since the $\ell=0,1$ modes of $r\phi$ decay like $u^{-2}$).
  Indeed, if $\delta\leq1$, then the asymptotics of the $\ell\geq 3$ modes would now be determined by the behaviour of the $\ell=0,1$ modes, and one can at best prove the upper bound $\psi_{\ell\geq 3}|_{\scrip} \lesssim u^{-3-\delta+}$.
  
   \begin{figure}[htb]
 	\includegraphics[width=0.6\textwidth]{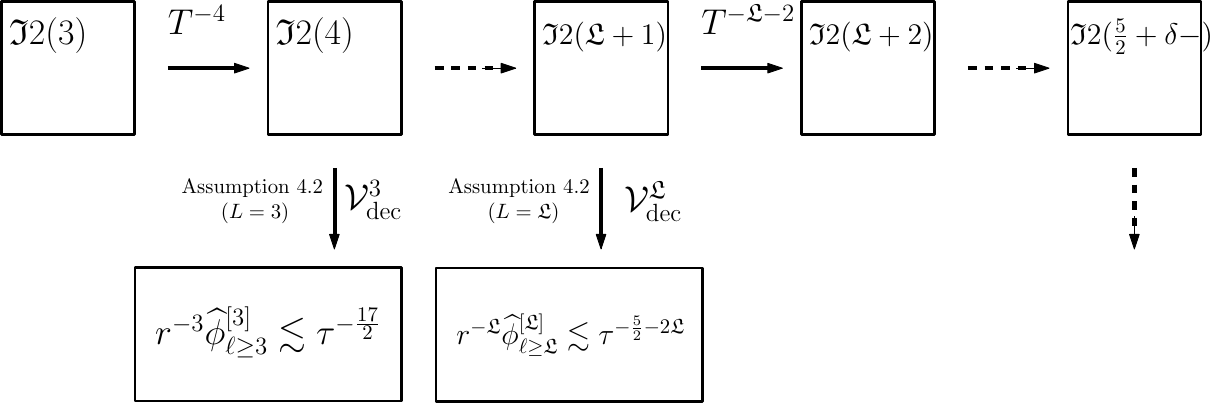}
 	\caption{An overview of iteration argument for improving $\mathfrak{L}$ in \cref{eq:time:ass2} (assuming  $\delta$ is sufficiently large).}
 \end{figure}

In general, we can prove the leading-order asymptotics for $\psi_{\geq \ell}\sim u^{-1-\ell}$ if and only if  $\ell<2+\delta$. This is the content of the following statements:
\begin{prop}\label{prop:time:nth}
	Assume that \cref{eq:time:ass2} holds with $2+\delta>\mathfrak{L}\in\N$, and let $\ell\geq \mathfrak{L}$. 
	
	\begin{enumerate}[label=\emph{\alph*)}]
	\item $(r^2\Xx T^{-\mathfrak{L}-1} \widehat{\psi}^{[\mathfrak{L}]}_\ell)|_{\Sigma_{\tau_0}}\lesV \max (r^{\mathfrak{L}-1-\delta},\log r)$, and
	
	\item For all $k\leq \mathfrak{L}$, we have, for $\ell\geq \mathfrak{L}+2$: $(r^2\Xx)^{\mathfrak{L}+1-k}(T^{-1-k}\psi_\ell)|_{\Sigma_{\tau_0}}\lesV \max( r^{k-1-\delta+},1)$. 
	On the other hand, for $\mathfrak{L}\leq \ell<2+\delta$, we have $(r^2\Xx)^{\ell-2}(T^{-\mathfrak{L}-1}\psi_\ell)|_{\Sigma_{\tau_0}}\lesV \max( r^{1-\delta+},\log r)$.
	\end{enumerate}
\end{prop}
\begin{proof}
	The proof is as in \cref{prop:time:third}, along with an induction in $k\leq \mathfrak{L}$.
\end{proof}

	\begin{prop}\label{prop:time:nthsum}
		Assume that \cref{eq:time:ass2} holds with $2+\delta>\mathfrak{L}\in\N$. 
		Then we have, for $\tilde{p}=\min(1+2(2+\delta-\mathfrak{L})-,2)$:
		\begin{equation}
			E_N[T^{-1-\mathfrak{L}}\widehat{\phi}^{[\mathfrak{L}]}_{\ell\geq\mathfrak{L}}](\tau_0)+\int_{\C_{\tau_0}} r^{\tilde{p}}(\pv T^{-1-\mathfrak{L}} \widehat{\psi}^{[\mathfrak{L}]}_{\ell\geq\mathfrak{L}})^2\dd \mu \lesV[\Vc] 1.
		\end{equation}
		Moreover, for $\ell_0$ sufficiently large, we have for all $k\leq \mathfrak{L}$
		\begin{equation}
			\sum_{j=0}^{\mathfrak{L}-k}\int_{\C_{\tau_0}} r^{\tilde{p}}(\pv(r^2\Xx)^j T^{-k-1} 	\widehat{\psi}^{[\mathfrak{L}]}_{\ell\geq\mathfrak{L}})^2\dd \mu \lesV[\Vc] 1
		\end{equation}
	\end{prop}
	\begin{proof}
		The proof is entirely analogous to that of \cref{prop:time:third:sum}.
	\end{proof}
	Analogously to \cref{cor:time:third}, we then have
	\begin{cor}\label{cor:time:nth}
		If \cref{eq:time:ass2} holds with $2+\delta>\mathfrak{L}\in\N$, then it also holds with $\max(\mathfrak{L}+1,5/2+\delta-)$.
		
		If, in addition, \cref{ass:main:aux} holds with $L\leq \mathfrak{L}$, then we also have
		\begin{equation}
		\widehat{\psi}^{[\mathfrak{L}]}_{\ell\geq \mathfrak{L}}\lesV (\tau^{-3-\delta}+\tau^{-3/2-\mathfrak{L}})\min_{q\in[0,L+1]}\left(\frac{r}{\tau}\right)^{q}.
		\end{equation} 
	\end{cor}
	
	Finally, inductively iterating \cref{prop:time:nth}, \cref{prop:time:nthsum} and \cref{cor:time:nth}, we finally deduce
	\begin{thm}
		\cref{eq:time:ass2} holds with $\mathfrak{L}=2\dots, \lfloor5/2+\delta-\rfloor$ and with  $\mathfrak{L}=5/2+\delta-$. 
			Furthermore, if \cref{ass:main:aux} holds for $L\in\N_{\geq 1}$, then \cref{eq:main:mainthm:higherell} holds for all $1\leq \ell_0\leq \min(2+\delta,L)$.
	\end{thm}
Together with \cref{thm:time:first}, this finishes the proof of \cref{thm:main}.

\subsection{Extensions}
\subsubsection{Improved error estimates}
\begin{obs}\label{obs:error}
	In \cref{thm:time:first}, we have proved an error term $\tau^{-5/2}$ using two time integrals. The reader may readily verify that we can, in fact, take a third time integral of $\widehat{\psi}^{[1]}$. Since $r^2\Xx T^{-2}\widehat{\psi}^{[1]}_{\ell}\sim \log r $ for $\ell=0,1$, the third time integral will satisfy $r^2\Xx T^{-3}\widehat{\psi}^{[1]}_\ell\sim r$. This is still enough to deduce the finiteness of the $r^p$-flux with $p<1$, which in turn translates into an error estimate of $\tau^{-3+}$ in \cref{thm:time:first}.
	Analogous improvements can be made in the other estimates of \cref{thm:main}.
\end{obs}

\subsubsection{Weakening the improved interior decay of the metric}
\begin{obs}\label{obs:weaken:away}
	Suppose that our \cref{ass:dyn:pol} did not contain improved decay away from $\scrip$, i.e.~suppose we merely had $g-g_M\lesssim r^{-1} \tau^{-\delta}$ rather than $(\tau+r)^{-1}\tau^{-\delta}$. Then, revisiting the proof of \cref{cor:time:seconditerate} with $\alpha=1$, we would merely have the bound $F_1\lesssim r^{-2} \tau^{-1/2-\delta}$ away from $\scrip$. We can then iteratively apply \cref{cor:inhom:decay} to eventually obtain the bound $\widehat{\psi}^{[1]}\lesV r\tau^{-\max(7/2,3+\delta)}$. In other words, we can still prove the global leading-order asymptotics for $\ell=0$ without the improved interior decay for the metric. 
\end{obs}

\subsubsection{Weakening the conformal regularity assumption on the metric}
\begin{obs}
	Using the expressions \cref{eq:dyn:Guv}--\cref{eq:dyn:GAB}, the reader will find that, in order to make sense of the expansion \cref{eq:time:proof:loggenerator} for the $\ell$-th mode, we merely require  conformal regularity $N=\ell$ in \cref{ass:dyn:pol}. More generally, if the metric $g$ is assumed to have a polyhomogeneous expansion in \cref{ass:dyn:pol}, then we will accordingly have a polyhomogeneous expansion in place of \cref{eq:time:proof:loggenerator}, which will then lead to new conformally irregular terms and thus determine the late-time asymptotics for sufficiently large $\ell$.
\end{obs}

\subsubsection{The case of non-compactly supported initial data}
\begin{obs}
	To simplify the presentation, we have assumed our initial data to be of compact support. If we instead had assumed that our initial data had a conformally regular expansion, i.e.~$\psi|_{\C_{\tau_0}}=1+1/r+\dots$, then we would now also have an initial data contribution in ${\integral}_{\ell,m}[\phi]$ for $\ell\geq1$, while, for $\ell=0$, the coefficient ${I}_0[\phi]$ would be completely determined from initial data: Indeed, for $\ell=0$, we need to subtract an approximate solution already at the level of the first time integral. The proof is otherwise unchanged.
	
	Of course, we may also consider conformally irregular data. This will simply mean less time-inversions (notice that even for compactly supported data for $\phi$, the associated data for $T^{-\ell}\phi_{\ell}$ are highly conformally irregular.)\todo{Perhaps further comments here later on.}
\end{obs}
\section{Proofs of \cref{thm:main:F,thm:main:phi3,thm:main:phi4,thm:main:puphipvphi}}\label{sec:additionalproof}
	In this section, we provides proofs for \cref{thm:main:F,thm:main:phi3,thm:main:phi4}.
	The proofs will be simpler variations of the proof of \cref{thm:main}. We will therefore be brief:	\begin{proof}[\textbf{Proof of \cref{thm:main:F}}]
		\textit{Starting point:}
		First of all, by an argument analogous to \cref{lemma:time:zeroth}, we may assume that \cref{eq:time:apriori} holds with $\alpha=0-$.
		The reason why we cannot directly infer $\alpha=0$ is simply that we can only perform $r^p$-estimates with $p<2$, since only then will the spacetime integral over $r^{p+1}(r\mathfrak{F})^2\lesssim r^{p-3}$ be well-defined. This arbitrarily small loss in decay is entirely inconsequential to our argument. (In fact, we can still take the full hierarchy of $r^p$-estimates up to $p=2$ if we simply estimate the bulk term as described in \cref{rem:inhom:rp:u}!)
		
		\textit{Initial data of $T^{-1}\phi_{\ell=0}$:} 
		Repeating the computations of \cref{prop:time:l=0:1}, we find directly that
		\begin{equation}
			\lim_{r\to\infty}\frac{r^2}{\log r} D\Xx T^{-1}\phi_{\ell=0}|_{\Sigma_{\tau_0}}=-\int_{\scrip} [\mathfrak{F}]^{(3)}\dd \mu\implies \lim_{r\to\infty}  T^{-1}\psi_{\ell=0}|_{\Sigma_{\tau_0}} =\int_{\scrip} [\mathfrak{F}]^{(3)}\dd \mu.
		\end{equation}
		Thus, for $\ell=0$, we can only deduce a finite $r^p$-flux with $p<1$. 
		
		\textit{Subtraction of approximate solution for $\ell=0$:}
		At this stage, we recall from \cref{prop:approx} that $\lim_{\scrip} T^{-1}\appsi[0,0]=\log r$. We may then infer that
		\begin{equation}
	\Xx	T^{-1}\big(\psi_{\ell=0}-\int_{\scrip}[\mathfrak{F}]^{(3)}\dd \mu \cdot \appsi[\ell=0]\big)\big|_{\Sigma_{\tau_0}}\lesV r^{-2}\log r,
		\end{equation}		
		which allows to infer the finiteness of the $r^p$-flux of the time integral of the difference with $p=2$. 
		
	\textit{Completing the first time-inversion:} Since the proof of \cref{prop:time:l>0:1} is unaffected by the slower decay-rate in $r$ of $\mathfrak{F}$, we may at this stage already deduce the leading-order asymptotics of the solution; in particular, we deduce the $\ell=0$-part of \cref{thm:main:F}.
	
	\textit{Further time-integrals for higher $\ell$:} 
	For higher $\ell$, we can proceed exactly as in the previous section.
	\end{proof}

\textbf{Proofs of \cref{thm:main:phi3,thm:main:phi4}:}
We need a decay statement as a starting point. 
For this, we quote~\cite{tohaneanu_pointwise_2022}:
\begin{thm}[Theorem 2.1 of \cite{tohaneanu_pointwise_2022}]
	Let $3\leq P\in\N$. Given sufficiently small, smooth compactly supported initial data to the equation 
	\begin{equation}
		\Box_{g_M}\phi= \phi^P
	\end{equation}
	 along $\Sigma_{\tau_0}$, then the solution $\phi$ is global-in-time and satisfies for some $\delta>0$:
	 \begin{equation}\label{eq:tohan}
	 \psi\lesV[\Vc] \tau^{-\delta}\min_{q\in[0,1]}\left(\frac{r}{\tau}\right)^q.
	 \end{equation}
\end{thm}
\begin{rem}
	We make two comments on this theorem. Firstly, \cite{tohaneanu_pointwise_2022} does not show that  \cref{eq:tohan} commutes with $r\pv$-vector fields, nevertheless, such commutations are easy to recover, arguing analogously to \cref{prop:inhom:rX}. 
	Secondly, \cite{tohaneanu_pointwise_2022} proves \cref{eq:tohan} with the sharp value of $\delta$; we shall not make use of this to illustrate the fact that it is sufficient for the methods of the present paper to apply \cref{eq:tohan} with an arbitrarily small $\delta>0$. 
	In fact, we also do not actually need the full strength of the improved interior decay encoded in the $\min_{q\in[0,1]}$.
\end{rem}
	\begin{proof}[\textbf{Proof of \cref{thm:main:phi3}}]
		Firstly, we may immediately upgrade \cref{eq:tohan} by simply iterating the results of \cref{sec:inhom}:
		\begin{lemma}
			If $\psi=r\phi$ satisfies \cref{eq:tohan}, then $\phi$ satisfies \cref{eq:time:apriori} with $\alpha=0-$.
		\end{lemma}
		In particular, the improved $T$-regularity, as well as the improved interior decay away from $\scrip$ follow from merely knowing \cref{eq:tohan}. 
		
		At this stage, we want to proceed as in the proof of \cref{thm:main:F}, with the only difference being that we need to iteratively improve the decay of $\mathfrak{F}=\phi^3$ as follows: 
		
	\textit{a)} Since $\phi$ satisfies \cref{eq:time:apriori} with $\alpha=0-$, we may apply \cref{thm:main:F} with $\beta=1/2-$, i.e.~we may take $\ell_0=0$ in \cref{thm:main:F}. This already gives that $\psi$ satisfies \cref{eq:time:apriori} with $\alpha=1/2$. 
	
	\textit{b)} We may apply \cref{thm:main:F} with $\beta=3$ and therefore with $\ell_0=0,1,2$. 
	
	\textit{c)} For higher $\ell$, we now make use of the fact that the projection $(\phi^3)_{\ell\geq \ell_0}$ will only contain products $\phi_{\ell_1}\phi_{\ell_2}\phi_{\ell_3}$ satisfying, in particular, $\sum {\ell_i}\geq \ell_0$.
	Thus, when considering the equation satisfied by $\psi_{\ell\geq \ell_0}$ with $\ell_0=3$, we may in fact take $\beta=6$ and apply \cref{thm:main:F} up until $\ell_0=5$. 
	
	\textit{d)} The full proof of \cref{thm:main:phi3} now follows by an inductive iteration. 
	\end{proof}
	
	\begin{proof}[\textbf{Proof of \cref{thm:main:phi4}}]
		We argue precisely as in the previous proof, except that we now appeal to \cref{thm:main} (with $g=g_M$) rather than to \cref{thm:main:F}.
	\end{proof}
	
	\textbf{Proof of \cref{thm:main:puphipvphi}:}
Note that global existence for $\Box_{g_M}\phi=\pu\phi\pv\phi$ is well-known. In particular, one may revisit the argument of \cite{dafermos_quasilinear_2022} and include $r\pv$-commutations to establish \cref{eq:main:additional:ass} for $\epsilon=1/2-$.
The additional regularity with respect to $\tau T$ is then obtained by iteratively applying the results of \cref{sec:inhom}.
	\begin{proof}[\textbf{Proof of \cref{thm:main:puphipvphi}}]
		For the higher $\ell>0$-modes, the proof is analogous to the previous ones, the only difference is for $\ell=0$. 
		Taking the first time-integral for $\ell=0$, we obtain that
		\begin{equation}
		4\pi	\lim_{r\to\infty}r^2\Xx(T^{-1}\psi_{\ell=0})=\frac{C_1+2MC_0}{2}, \qquad\text{where } C_1=\int_{\scrip}(\pu\phi\pv\phi)^{(4)}\dd\mu,
		\end{equation}
		and where $C_0$ is as in the stationary case (see~\cref{eq:time:C0}). 
		Here, we already used that $(\pu\phi\pv\phi)^{(3)}$ is a total derivative (see~\cref{eq:main:tralala}) and therefore does not contribute. 
		
		We next take the second time integral, where we now get two logarithmic contributions:
		\begin{equation}
		4\pi	r^2D\Xx T^{-2}\phi_{0}=\int_{\rc}^r \frac{C_1+2MC_0}{2r'}\dd r' -\int_{\rc}^{r}\frac1{r'}\dd r' \int_{\scrip}(\tau_0-\tau')\underbrace{(\pu\phi\pv\phi)^{(3)}}_{=-\frac12 \pu((\psi^{(0)})^2)}\dd\mu,
		\end{equation}
		from which the result follows after integration by parts.
	\end{proof}
\section{Comments on the genericity of the tails}\label{sec:gen}
In this section, we give a proof of the statement that the coefficients $\integral_{\ell,m}[\phi]$ featuring in \cref{eq:main:mainthm:higherell} are each non-vanishing in the complement of a set of initial data of codimension at least 1, within the space of smooth, compactly supported initial data on $\tilde{\Sigma}_{\tilde{\tau}_0}$.
The argument here is inspired by the genericity argument of \cite{luk_strong_2019}, and can also be extended to the nonlinear examples \cref{thm:main:phi3,thm:main:phi4,thm:main:puphipvphi}; we give a sketch of this as well.

\subsection{Linear equations} For concreteness, we focus on the setting of \cref{thm:main}. 
The idea is the following: Given a solution $\phi$ for which the integral $\integral_{\ell,m}[\phi]$ vanishes, we want to add a $\lambda$-multiple of a solution $\phi^f$ to $\Box_g\phi^f=0$ for which the integral $\integral_{\ell,m}[\phi^f]$ does not vanish and then argue by linearity that $\integral_{\ell,m}[\phi+\lambda \phi^f]$ only vanishes for $\lambda=0$.
 
In principle, we could construct $\phi^f$ via a backwards scattering construction starting from an appropriate, compactly supported radiation field $\psi^f|_{\scrip}=r\phi^f|_{\scrip}=f(u,\theta^A)$ along $\scrip$ and trivial data along $\hplus$; however, the arising solution would not induce compactly supported data on $\Sigma_{\tilde{\tau}_0}$. We will therefore construct an solution from a backwards scattering argument with initial data imposed on a finite ingoing null hypersurface, such that $\phi^f|_{\Sigma_0}$ is compactly supported at the expense of $\psi^f|_{\scrip}$ only \emph{approximating} a fixed function $f$.

Let $U\gg u_0$ and $R\gg 2M$ be sufficiently large, and let $V=U+r^{\ast}(R)$.
Let $f(u,\theta^A)$ be compactly supported in $(u_0,u_0+1)$.
Then we define $\phi^f$ as the solution to $\Box_{g}\phi^f=0$ arising from the following data:\footnote{Note that even when studying the equation $\Box_g\phi=\mathfrak{F}$ for some inhomogeneity $\mathfrak{F}$, we still want $\phi^f$ to solve the homogeneous equation! Note also that $\phi^f$ is defined in the entire future domain of dependence of $\tilde{\Sigma}_{\tilde{\tau}_0}$}
\begin{align*}
		\psi^f|_{\tilde{\Sigma}_{\tilde{\tau}=U}\cap\{v\leq V\}}\equiv &\:0,\\
	\psi^f(u,V)=&\: f(u,\theta^A)\quad &(u_0\leq u\leq U),\\
	\psi^f(u_0,v)=&0\quad &(v\geq V+1),\\
\psi^f|_{\mathcal{H}^+\cap J^-(\tilde{\Sigma}_{\tilde{\tau}=U})}\equiv &\:0.
\end{align*}
See Figure \ref{fig:gen} for a pictorial representation of the above set-up.
\begin{figure}[htb]
	\includegraphics[width=0.5\textwidth]{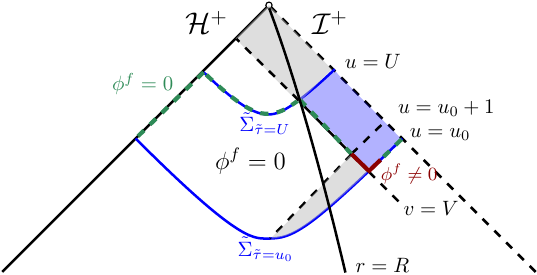}
	\caption{Set-up of the solution $\phi^f$. It is only nonvanishing in the shaded regions.}\label{fig:gen}
\end{figure}
Notice that all $\pv$-derivatives of $\psi^f(u_0,v)$ at $v=V$ are determined by integrals over $f(u)$ via the equation $\Box_{g}\phi^f=0$ and the first of the conditions above, and we then smoothly extend the data along $u=u_0$ to vanish for $v\geq V+1$.
By a standard semi-global scattering argument (see e.g.~\cite{kadar_scattering_2025}), we may directly obtain the  boundedness of our $r^p$-weighted and $\Vc$-as well as $r^2\partial_v$-commuted energies along $\tilde{\Sigma}_{\tilde{\tau}=u_0}$ \textit{uniformly in $U$ and $R$}. In turn, by \cref{thm:main}, we may then infer the following global decay estimate as $\tau\to\infty$:
\begin{equation}
	\psi^f \leq C(f) \tau^{-2}.
\end{equation}
We will now show that $f$ is a good approximation for $\psi^f|_{\scrip}$ for $u_0\leq u\leq U$.
\begin{lemma}\label{lem:gen}
	Let $\psi^f=r\phi^f$ be the solution to $\Box_{g}\phi^f=0$ constructed above. Then we have in $[u_0,U]\times [V,\infty)\times S^2$:
	\begin{equation}
	|	\psi_\ell^f(u,v,\theta^A)-f_\ell(u,\theta^A)|\leq \frac{C_{\ell}}{R}||f||_{L^{\infty}((u_0,u_0+1)\times S^2)}+C(f) R^{-1+}.
	\end{equation}
\end{lemma}
\begin{proof}
	The proof proceeds in two steps, we first prove the estimate for the corresponding Minkowksian problem, and then estimate the difference between the Minkowskian solution and the actual solution.
	We define the Minkowskian wave operator similarly as in \cref{sec:approx}, but now with respect to $(u,v)$-coordinates (simply to keep the null cones the same), and we let $\phi^\eta$ denote the solution to $\Box_\eta\phi^\eta=0$ arising from data as above (restricted to the region $[u_0,U]\times [V,\infty)\times S^2$). 	Then we have the following identities: (cf.~\cref{eq:inhom:r2Xcommuted}:)
	\begin{equation}
		\pu(r^{-2k}\pv(r^2\pv)^k\psi^\eta_\ell)=(k(k+1)-\ell(\ell+1))r^{-2k-2}(r^2\pv)^k\psi^\eta_\ell.
	\end{equation}
	We may use these identities to determine the $\pv$-derivatives along $v=V$ of $\psi^\eta_\ell$; these will be bounded by $||f||_{L^{\infty}((u_0,u_0+1)\times S^2)}$. 
		We then set $k=\ell$ in the above to infer that
		\begin{equation}
			(r^2\pv)^\ell \psi^\eta_\ell(u,v)=(r^2\pv)^\ell \psi^\eta_\ell(u,V)+\int_{V}^{{v}} r(u,v')^{2\ell} r(u_0,v')^{-2\ell}\pv(r^2\pv)^\ell\psi^\eta_\ell(u_0,v')\dd v', 
		\end{equation}
(notice the the integrand above is compactly supported)	 and, after an inductive argument consisting of repeated integration, we obtain:
		\begin{equation}
		|	\psi^\eta_\ell(u,\bar{v})-f(u)|\leq C_\ell R^{-1} ||f||_{L^{\infty}((u_0,u_0+1))}.
		\end{equation}
		The estimate still holds after commuting with vector fields in $\V$.
		
		We now estimate the difference between Minkowskian and actual solution: 
	Notice that we have
		\begin{equation}
			\Box_{g_M}(\phi^f-\phi^\eta)= O(r^{-3})|\V^2\phi^\eta|+F \implies \pu\pv(\psi^f-\psi^\eta)=\frac{\Dl (\psi^f-\psi^\eta)}{r^2}+rF +O(r^{-3}\log r)|\V^2\psi^\eta|,
		\end{equation}
		from which the result follows after integrating in $u$ and $v$. (Since we actually have $\psi^f\lesssim u^{-2}$ at our disposal, we have a stronger estimate, but we wrote the estimate of \cref{lem:gen} to also cover the case where $\psi^f\lesssim u^{-1}$.)
\end{proof}

The upshot it is that we can construct solutions whose radiation fields are approximately $f(u)$ for $u\leq U$ (with error term $O(R^{-})$), and which are of size $U^{-2}$ for $u\geq U$.  Thus, if we want to deduce the generic nonvanishing of, say, $\integral_0[\phi]$ (and setting $\mathfrak{F}=0$ for simpler notation), we can then pick $f$ such that
\begin{equation}
	\int_{\scrip\cap\{u\leq U\}} \mathfrak{m} f \dd \mu\geq c_f
\end{equation}
for some $c_f>0$. Then, choosing $R$ and $U$ sufficiently large depending only on $f$ and $\mathfrak{F}$, we obtain that
\begin{equation}
	\int_{\scrip}\mathfrak{m}\psi^f \dd \mu \geq c_f/2.
\end{equation}

Therefore, given any solution to $\Box_{g}\phi=\mathfrak{F}$ such that $\integral_0[\psi]$ vanishes, we can deduce by linearity that $\integral_0[\psi+\lambda \psi^f]\neq 0$ for $\lambda\neq 0$, ensuring the generic nonvanishing of $\integral_0$. Entirely analogous arguments hold for the nonvanishing $\integral_{\ell=1,m}$, so we have:
\begin{prop}\label{prop:gen}
	In the setting of \cref{thm:main}, the subset of smooth, compactly supported data leading to vanishing $I_{0}[\phi]$ has at least codimension 1, and the set of initial data leading to vanishing $I_{\ell=1,m}[\phi]$ for all $m\in\{-1,0,1\}$ has codimension at least 3.
\end{prop}
The argument also directly generalises to higher $\ell$, except that we now first have to re-write the integrals $\integral_{\ell,m}[\phi]$ using the wave equation and integration by parts so that they only feature tangential derivatives of $\phi$ (and the vanishing for all $m$ will then imply codimension at least $2\ell+1$).

\subsection{Nonlinear equations}
The argument of \cref{lem:gen} still holds for nonlinear equations, so we can still construct nonlinear solutions $\psi^f$ whose radiation field is approximately $f$, for some freely prescribable compactly supported function $f$. 
However, $\psi+\lambda\psi^f$ need no longer be a solution to a linear equation, so we have to modify the argument. We will sketch a modification that will apply to general nonlinear equations.

For concreteness, we consider the equation $\Box_{g_M}\phi=\phi^3$, and denote by $P$ the nonlinear operator defined by:
\begin{equation}
	P\phi=\Box_{g_M}\phi-\phi^3.
\end{equation}
Recall from \cref{thm:main:phi3} that the $\ell=0$-late-time behaviour is determined by the integral
\begin{equation}
J_0[\phi]=\frac12\frac{1}{4\pi}	\int_{\scrip}\psi^3 \dd \mu.
\end{equation}
Suppose then that we have a  solution $\phi$ for which $J_0[\phi]=0$.
Roughly, we want to perturb the radiation field $\psi|_{\scrip}$ to $\psi_{\scrip}+\lambda f$, which would perturb the integral above according to
\begin{equation}\label{eq:lambda}
	\int_{\scrip}\psi^3+3\lambda \psi^2 f+3\lambda^2 \psi f^2+\lambda^3 f^3 \dd\mu.
\end{equation}
Notice then that unless $\psi|_{\scrip}$ vanishes identically along $\scrip$, the $\lambda f \psi^2$-term dominates for sufficiently small $\lambda$. We therefore need to distinguish between the following two cases:

\textbf{Case A): $\psi|_{\scrip}\equiv 0$.} 
If the radiation field is identically zero, then we first let $\phi^f$ be a size-$\epsilon$ solution to $P\phi^f=0$ such that $J_0[\phi^f]>\epsilon c_f/2>0$; this exists by \cref{lem:gen}. 
Secondly, we define $\tilde{\phi}_\lambda$ as a solution to $P \tilde{\phi}_{\lambda}=0$ arising from the initial data induced by the sum $\phi+\lambda \phi^f$ on $ \tilde{\Sigma}_{\tilde{\tau_0}}$.

All three functions $\psi=r\phi$, $\psi^f=r\phi^f$ and $\tilde{\psi}_{\lambda}=r\tilde{\phi}_{\lambda}$ can be globally, uniformly estimated by $C\cdot \epsilon \tau^{-1}$, where $\epsilon$ depends purely on a choice of initial data norm and is sufficiently small to guarantee global existence, and we take all solutions to lie in the same ball $B(0,\frac{\epsilon}{3})$ with respect to this initial data norm.

In particular, if we then study the equation satisfied by $\phi^{\Delta}:=\tilde{\phi}_{\lambda}-(\phi+\lambda\phi^f)$,
\begin{equation}
	\Box_{g_M}\phi^{\Delta}=\tilde{\phi}_{\lambda}^3-\phi^3-(\lambda\phi^f)^3, 
\end{equation}
then, since $\phi^{\Delta}$ has trivial data, on $\tilde{\Sigma}_{\tilde{\tau}=u_0}$ by construction, we can prove that $\psi^{\Delta}\leq  C\lambda  \epsilon^3\tau^{-1}$. Thus, using that $\psi|_{\scrip}\equiv 0$, we may deduce that
\begin{equation}
	J_{0}[\tilde{\phi}_{\lambda}] =J_0[\lambda\phi^f]+O(\lambda^4 \epsilon^5),
\end{equation}
from which it follows that $I_{0}[\tilde{\phi}_{\lambda}]\geq \epsilon \frac{c_f}{4}$, and from which the property of being codimension at least 1 follows for sufficiently small $\lambda$.

\textbf{Case B): $\psi|_{\scrip}$ not vanishing identically.}
If $\psi|_{\scrip}$ does not vanish identically, then there is an open set contained in $(u_1,u_2)\times S^2$ along which $\psi|_{\scrip}$ is bounded strictly away from zero, for some $u_2>u_1>u_0$. 
We then want to perturb the radiation field by function $f$ supported in $(u_1,u_2)$ such that the term linear in $\lambda$ in \cref{eq:lambda} is non-vanishing; in particular the choice of $f$ now depends on the value of $\psi|_{\scrip}$ in $(u_1,u_2)\times S^2$.

Rather than letting $\psi^f$ be a solution to $P \psi^f=0$, we instead consider the linearisation of $P$ around $\psi$; that is, we let $\dot{\phi}^f$ be solution to the linear problem
\begin{equation}
	\Box_{g_M}\dot{\phi}^f=3\dot{\phi}^f\phi^2
\end{equation}
such that the initial data of $\dot{\phi}^f$ induced on $\tilde{\Sigma}_{\tilde{\tau}_0}$ lies in a ball $ B(0,\frac{\epsilon}{3})$ and
\begin{equation}
	\int_{\scrip\cap \{u\leq U\}} 3\dot{\psi}^f \psi^2 \dd \mu >\delta(\epsilon) \cdot c_f, 
\end{equation}
where $\delta(\epsilon)$ depends on the size of $\psi$ in $(u_1,u_2)\times S^2$.

The construction of such a $\dot{\psi}^f$ is analogous to the construction in the previous subsection.
Notice that we have, in particular, that
\begin{equation}
	\Box_{g_M}(\phi+\lambda \dot{\phi}^f)=(\phi+\lambda\dot{\phi}^f)^3+O(\lambda^2).
\end{equation}
In fact, by defining $\tilde{\phi}_{\lambda}$ to be the solution to $P\tilde{\phi}_{\lambda}$ arising from the initial data induced by the sum $\phi+\lambda\dot{\phi}^f$ along $\tilde{\Sigma}_{\tilde{\tau}_0}$, analogously to the set-up of Case A), and by then considering the equation satisfied by the difference 
 $\phi^{\Delta}=\tilde{\phi}_{\lambda}-(\phi+\lambda\dot{\phi}^f)$, namely
 \begin{equation}
 	\Box_{g_M}\phi^{\Delta}=\phi^{\Delta}(\tilde{\phi}_{\lambda}^2+\tilde{\phi}_{\lambda} (\phi+\lambda\dot{\phi}^f)+ (\phi+\lambda\dot{\phi}^f)^2)+O(\lambda^2\epsilon^3),
 \end{equation} we can now deduce an improved estimate of the form $\psi^{\Delta}\leq C \lambda^2 \epsilon^3 \tau^{-1}$, from which it follows that
\begin{equation}
	J_0[\tilde{\phi}_{\lambda}]=\frac{1}{8\pi} 	\int_{\scrip} 3\lambda \dot{\psi}^f \psi^2 \dd \mu+O(\lambda^2\epsilon^5),
\end{equation}
from which it follows that we $J_0[\tilde{\phi}_{\lambda}]>\frac{1}{2}\delta(\epsilon) \cdot c_f$  for sufficiently small $\lambda$, and we can deduce the property of being codimension at least 1 (for sufficiently small $\lambda$). 

Notice that the above argument is not particular to the $\phi^3$ nonlinearity and generalises to other nonlinearities.

\appendix
\section{Motivation of assumptions on the dynamical metrics}\label{app:motivation}
The aim of this appendix is to assist the reader in motivating and interpreting \cref{ass:dyn:pol} defining the class of dynamical metrics studied in the present paper. 

\paragraph{Near null infinity}
We begin with a discussion of the assumptions near $\scrip$. 
First of all, our dynamical metrics satisfy, in particular, the central global gauge conditions near $\scrip$:
\begin{equation}\label{eq:app:gauge}
	\Omega^2=\Omega^2_M+O(r^{-2}), \qquad \slashed{g}_{AB}-\slashed{g}_{M,AB}=O(r), \qquad b^A=O(1),
\end{equation}
where the $O$-terms may include arbitrary $u$-dependency for now.

At this stage, we assume that $g$ satisfies the Einstein vacuum equations and we use the usual notation for the connection coefficients $\et, \etb, \chi, \underline{\chi}$:
\begin{equation}
	\chi_{AB}=g(\nabla^g_A e_4, e_B), \quad 	\underline{\chi}_{AB}=g(\nabla^g_A e_3, e_B), \quad 	\et_{AB}=-\frac12 g(\nabla^g_3 e_A, e_4),\quad \etb_{AB}=-\frac12 g(\nabla^g_4 e_A, e_3);
\end{equation}
furthermore,  $\xh$ and $\xhb$ denote the traceless parts of $\chi$ and $\underline{\chi}$, respectively. With respect to this decomposition, the Einstein vacuum equations (see, for instance, \cite{dafermos_non-linear_2021}[Section~1.2]) then, when also assuming \cref{eq:app:gauge} and assuming the curvature components $\rho, \sigma, \beta,\alpha$ to all decay faster than $r^{-2}$, imply the following relations along $\scrip$:
\begin{align}
	&	\pu \hat{\slashed{g}}^{(1)}_{AB}=2\xhb^{(1)}_{AB} \qquad &\text{\big(using \cite{dafermos_non-linear_2021}[eq. (1.2.5)]\big)},\\
	& \xh^{(2)}=-\frac12 \hat{\slashed{g}}^{(1)}  \qquad &\text{\big(using \cite{dafermos_non-linear_2021}[eq. (1.2.8)]\big)},\\
	&\pu \otrx^{(2)}=0\qquad &\text{\big(using \cite{dafermos_non-linear_2021}[eq. (1.2.10)]\big)},\\
	&\pu \otrxb^{(2)}=-|\xhb^{(1)}|^2 \qquad &\text{\big(using \cite{dafermos_non-linear_2021}[eq. (1.2.7)]\big)},\\
	&	{\pu \tr g^{(1)}=0}	\qquad &\text{\big(using \cite{dafermos_non-linear_2021}[eq. (1.2.5)]\big)}	,\\
	&\pu \et^{(2)}_A=-(\mathring{\div}\xhb^{(1)})_{A} \qquad &\text{\big(using \cite{dafermos_non-linear_2021}[eq. (1.2.14) and (1.2.19)]\big)},\\
	&\et^{(2)}_A=-\etb^{(2)}_A \qquad &\text{\big(using \cite{dafermos_non-linear_2021}[eq. (1.2.21)]\big)},\\
	&b^{(1)}_A=-2\et^{(2)}_A \qquad &\text{\big(using \cite{dafermos_non-linear_2021}[eq.~(1.2.34)]\big)}.
\end{align}
If we now assume that along a single cone, the quantities  $\otrxb^{(2)}$ and $\tr g^{(1)}$ vanish. Then they necessarily vanish everywhere, and we moreover have that $\otrxb^{(2)}=4m(u)$ (in view of \cref{lem:dyn:Bondi}). 
In addition, if we assume that $\eta^{(2)}$ vanishes as $u\to\infty$ (this, i.e.~the equivalent statement for $b^{(1)}$ is contained in \cref{ass:dyn:pol}!), then we have $\eta^{(2)}=-\frac12 \div \gsh^{(1)}$ and thus $b^{(1)}=\mathring{\div}\gsh^{(1)}$.

Thus, under these assumptions, the leading-order behaviour imposed in \cref{ass:dyn:pol} corresponds exactly to the assumption that $\gsh^{(1)}$ decays like $u^{-\delta}$. 
The behaviour of the higher-order terms in the expansions of \cref{ass:dyn:pol} then follows similarly. In fact, this will depend entirely on the conformal regularity, i.e.~the $1/r$-expansions, of the initial data.

Now, the assumptions \cref{eq:app:gauge}, along with the aforementioned assumptions on the limit of $\eta^{(2)}$ vanishing and  $\otrxb^{(2)}$ and $\tr g^{(1)}$ vanishing are satisfied in the $\scrip$-gauge of \cite{dafermos_non-linear_2021}; see \cite{dafermos_non-linear_2021}[Definition 2.2.1] (as well as the computations in \cite{dafermos_non-linear_2021}[Section 14], or the final pointwise estimates \cite{dafermos_non-linear_2021}[(eq. 6.1.22)]).

The work \cite{dafermos_non-linear_2021} thus essentially constructs spacetimes as in \cref{ass:dyn:pol} with $\delta=1$, and with some finite $N$ (though this depends solely on the assumption on the initial data). \textit{However,} improved decay with respect to time-derivatives are not proved in \cite{dafermos_non-linear_2021}, and, in the same vain, it is not shown that $\xhb^{(1)}$ decays one power faster than $\gsh^{(1)}$. That is to say, \cite{dafermos_non-linear_2021} merely establishes $\xhb^{(1)}\lesssim u^{-1}$ (though this could easily be improved to $u^{-2}$ even within the existing framework of \cite{dafermos_non-linear_2021}).

Notice that commutation estimates with $r\partial_r|_{\Sigma}$ are already contained in \cite{dafermos_non-linear_2021}, which is the crucial step towards getting improved decay for time derivatives (cf.~\cref{lem:inhom:Tconversion}).

\paragraph{Away from null infinity}
Away from null infinity, our \cref{ass:dyn:pol} assumes one additional power in $\tau$-decay. This additional power is not present in \cite{dafermos_non-linear_2021}. It would be recovered by showing that $T$-derivatives decay one power faster, and then performing an elliptic estimate analogous to the one leading to \cref{eq:inhom:cor:all:l} of \cref{cor:inhom:decay}.
Similarly, the assumption on faster decay for higher $\ell$-modes (\cref{ass:main:aux}) would follow from an estimate analogous to \cref{eq:inhom:cor:high:l}.

\paragraph{The global coordinate system}
The class of metrics in \cref{ass:dyn:pol} is given with respect to one global coordinate system (modulo the two coordinate charts necessary to cover the sphere). This is not the case for the metrics constructed in \cite{dafermos_non-linear_2021}, where there are two different coordinate systems (ignoring the two different coordinate systems for the sphere), one near $\hplus$ and one near $\scrip$. 
In order to construct out of these one global coordinate system, one can, for instance, glue the two coordinate systems along a suitable spacelike hypersurface sitting at ``$t=n$'', $t$ denoting a suitably glued time-coordinate, then backwards construct a full coordinate system from this spacelike hypersurface, and then sending $n$ to infinity. \todo{Maybe comment on purpose.}

\section{Deriving the a priori assumptions on $\phi$}\label{app:derivapriori}
In this appendix, we give a short sketch of how to derive (III) of \cref{ass:main}, i.e.~we give a sketch of the proof of the following
\begin{prop}
Let $g$ be as in \cref{ass:dyn:pol}, and let $\phi$ arise from smooth compactly supported data to $\Box_g\phi=0$.
Then \cref{eq:time:apriori} holds with $\alpha=0$.
\end{prop}
\begin{proof}[Sketch of proof]The cleanest way to prove this is to use the methods of compatible currents, where the argument simply becomes a stability argument around the Schwarzschildean wave equation. Indeed, in the notation of \cite{dafermos_quasilinear_2022}, we may construct a quadruple $V,w,q,\varphi$ such that the divergence of the associated compatible current $J^{V,w,q,\varphi}[g_M,\phi]$ satisfies appropriate coercivity properties, and such that the fluxes control the nondegenerate energy and the $r^p$-flux. 
After applying the divergence theorem\footnote{For this, we have to work with the smoothened out level sets $\tilde{\Sigma}_{\tilde{\tau}}$ of $\tilde{\tau}$. }, the resulting identity, being geometric in nature, then turns out to be robust with respect to perturbations of $g_M$ to $g$ and is sufficient to give the estimate $\psi\lesssim \tau^{-1/2}$. One proves an analogous estimate for commutations with $\sl, T$ and $r\partial_r|_{\Sigma}$. 
The improved $\tau$-decay of time-derivatives, and the improved $\tau$-decay away from $\scrip$, is then obtained exactly as in \cref{sec:inhom} (after a few iterations). (In fact, one may at this stage treat $\Box_g \phi=0$ as inhomogeneous equation on Schwarzschild again, which however requires one extra iteration if $\delta\leq 1/2$.)
\end{proof}

\small{\bibliographystyle{alpha}
\bibliography{allBib}}
\end{document}